\documentclass[12pt]{article}
\usepackage{amsmath}
\usepackage{graphicx,psfrag,epsf}
\usepackage{enumerate}
\usepackage{natbib}

\usepackage{authblk}            

\usepackage{url} 

\usepackage{tikz}
\usepackage{pgfplots}
\usetikzlibrary{spy,calc}
\usetikzlibrary{automata}

\newcommand{\blind}{1}

\usepackage{scalerel,stackengine}
\stackMath
\newcommand\reallywidehat[1]{%
\savestack{\tmpbox}{\stretchto{%
  \scaleto{%
    \scalerel*[\widthof{\ensuremath{#1}}]{\kern.1pt\mathchar"0362\kern.1pt}%
    {\rule{0ex}{\textheight}}
  }{\textheight}%
}{2.4ex}}%
\stackon[-6.9pt]{#1}{\tmpbox}%
}

\usepackage{amssymb,bbm,amsmath} 
\usepackage{bm}

\usepackage{amsfonts,amsmath,amssymb,bbm,amsthm}

\usepackage[font=footnotesize,labelfont=bf]{caption}
\usepackage{subcaption}
\usepackage{float,mathtools}
\usepackage[shortlabels]{enumitem}
\usepackage{xcolor}
\usepackage{multirow}
\usepackage{booktabs}
\usepackage{varwidth}

\usepackage{natbib}
\usepackage{chngcntr}

\usepackage{ulem}
\usepackage{nameref,xr}
\usepackage{tikz}
\usepackage{pgfplots}
\usetikzlibrary{spy,calc}

\usepackage[plain,noend]{algorithm2e}
\makeatletter
\renewcommand{\algocf@captiontext}[2]{#1\algocf@typo. \AlCapFnt{}#2} 
\def\@algocf@capt@plain{top}
\renewcommand{\algocf@makecaption}[2]{%
  \addtolength{\hsize}{\algomargin}%
  \sbox\@tempboxa{\algocf@captiontext{#1}{#2}}%
  \ifdim\wd\@tempboxa >\hsize
    \hskip .5\algomargin%
    \parbox[t]{\hsize}{\algocf@captiontext{#1}{#2}}
  \else%
    \global\@minipagefalse%
    \hbox to\hsize{\box\@tempboxa}
  \fi%
  \addtolength{\hsize}{-\algomargin}%
}
\makeatother

\newcommand{\w}{\xi}     
\newcommand{\rd}{\color{red}}
\newcommand{\bl}{\color{blue}}
\newcommand{\bk}{\color{black}}
\newcommand{\var}{\text{var}}
\newcommand{\E}{E}

    \newcommand{\shat}{\hat{\sigma}_n} 

\newcommand{\cei}{\theta^{(i)}}

\newcommand{\bootsd}{\hat{\nu}_n}

\newcommand{\noisesd}{\tilde{\sigma}_{n}}
\newcommand{\noisesdax}{\tilde{\sigma}_{n \ |  A,X}}

\newif\ifblackandwhitecycle
\gdef\patternnumber{0}

\pgfkeys{/tikz/.cd,
    zoombox paths/.style={
        draw=orange,
        very thick
    },
    black and white/.is choice,
    black and white/.default=static,
    black and white/static/.style={ 
        draw=white,   
        zoombox paths/.append style={
            draw=white,
            postaction={
                draw=black,
                loosely dashed
            }
        }
    },
    black and white/static/.code={
        \gdef\patternnumber{1}
    },
    black and white/cycle/.code={
        \blackandwhitecycletrue
        \gdef\patternnumber{1}
    },
    black and white pattern/.is choice,
    black and white pattern/0/.style={},
    black and white pattern/1/.style={    
            draw=white,
            postaction={
                draw=black,
                dash pattern=on 2pt off 2pt
            }
    },
    black and white pattern/2/.style={    
            draw=white,
            postaction={
                draw=black,
                dash pattern=on 4pt off 4pt
            }
    },
    black and white pattern/3/.style={    
            draw=white,
            postaction={
                draw=black,
                dash pattern=on 4pt off 4pt on 1pt off 4pt
            }
    },
    black and white pattern/4/.style={    
            draw=white,
            postaction={
                draw=black,
                dash pattern=on 4pt off 2pt on 2 pt off 2pt on 2 pt off 2pt
            }
    },
    zoomboxarray inner gap/.initial=5pt,
    zoomboxarray columns/.initial=2,
    zoomboxarray rows/.initial=2,
    subfigurename/.initial={},
    figurename/.initial={zoombox},
    zoomboxarray/.style={
        execute at begin picture={
            \begin{scope}[
                spy using outlines={%
                    zoombox paths,
                    width=\imagewidth / \pgfkeysvalueof{/tikz/zoomboxarray columns} - (\pgfkeysvalueof{/tikz/zoomboxarray columns} - 1) / \pgfkeysvalueof{/tikz/zoomboxarray columns} * \pgfkeysvalueof{/tikz/zoomboxarray inner gap} -\pgflinewidth,
                    height=\imageheight / \pgfkeysvalueof{/tikz/zoomboxarray rows} - (\pgfkeysvalueof{/tikz/zoomboxarray rows} - 1) / \pgfkeysvalueof{/tikz/zoomboxarray rows} * \pgfkeysvalueof{/tikz/zoomboxarray inner gap}-\pgflinewidth,
                    magnification=3,
                    every spy on node/.style={
                        zoombox paths
                    },
                    every spy in node/.style={
                        zoombox paths
                    }
                }
            ]
        },
        execute at end picture={
            \end{scope}
            \node at (image.north) [anchor=north,inner sep=0pt] {\subcaptionbox{\label{\pgfkeysvalueof{/tikz/figurename}-image}}{\phantomimage}};
            \node at (zoomboxes container.north) [anchor=north,inner sep=0pt] {\subcaptionbox{\label{\pgfkeysvalueof{/tikz/figurename}-zoom}}{\phantomimage}};
     \gdef\patternnumber{0}
        },
        spymargin/.initial=0.5em,
        zoomboxes xshift/.initial=1,
        zoomboxes right/.code=\pgfkeys{/tikz/zoomboxes xshift=1},
        zoomboxes left/.code=\pgfkeys{/tikz/zoomboxes xshift=-1},
        zoomboxes yshift/.initial=0,
        zoomboxes above/.code={
            \pgfkeys{/tikz/zoomboxes yshift=1},
            \pgfkeys{/tikz/zoomboxes xshift=0}
        },
        zoomboxes below/.code={
            \pgfkeys{/tikz/zoomboxes yshift=-1},
            \pgfkeys{/tikz/zoomboxes xshift=0}
        },
        caption margin/.initial=4ex,
    },
    adjust caption spacing/.code={},
    image container/.style={
        inner sep=0pt,
        at=(image.north),
        anchor=north,
        adjust caption spacing
    },
    zoomboxes container/.style={
        inner sep=0pt,
        at=(image.north),
        anchor=north,
        name=zoomboxes container,
        xshift=\pgfkeysvalueof{/tikz/zoomboxes xshift}*(\imagewidth+\pgfkeysvalueof{/tikz/spymargin}),
        yshift=\pgfkeysvalueof{/tikz/zoomboxes yshift}*(\imageheight+\pgfkeysvalueof{/tikz/spymargin}+\pgfkeysvalueof{/tikz/caption margin}),
        adjust caption spacing
    },
    calculate dimensions/.code={
        \pgfpointdiff{\pgfpointanchor{image}{south west} }{\pgfpointanchor{image}{north east} }
        \pgfgetlastxy{\imagewidth}{\imageheight}
        \global\let\imagewidth=\imagewidth
        \global\let\imageheight=\imageheight
        \gdef\columncount{1}
        \gdef\rowcount{1}
        
    },
    image node/.style={
        inner sep=0pt,
        name=image,
        anchor=south west,
        append after command={
            [calculate dimensions]
            node [image container,subfigurename=\pgfkeysvalueof{/tikz/figurename}-image] {\phantomimage}
            node [zoomboxes container,subfigurename=\pgfkeysvalueof{/tikz/figurename}-zoom] {\phantomimage}
        }
    },
    color code/.style={
        zoombox paths/.append style={draw=#1}
    },
    connect zoomboxes/.style={
    spy connection path={\draw[draw=none,zoombox paths] (tikzspyonnode) -- (tikzspyinnode);}
    },
    help grid code/.code={
        \begin{scope}[
                x={(image.south east)},
                y={(image.north west)},
                font=\footnotesize,
                help lines,
                overlay
            ]
            \foreach \x in {0,1,...,9} { 
                \draw(\x/10,0) -- (\x/10,1);
                \node [anchor=north] at (\x/10,0) {0.\x};
            }
            \foreach \y in {0,1,...,9} {
                \draw(0,\y/10) -- (1,\y/10);                        
                \node [anchor=east] at (0,\y/10) {0.\y};
            }
        \end{scope}    
    },
    help grid/.style={
        append after command={
            [help grid code]
        }
    },
}

\newcommand\phantomimage{%
    \phantom{%
        \rule{\imagewidth}{\imageheight}%
    }%
}
\newcommand\zoombox[2][]{
    \begin{scope}[zoombox paths]
        \pgfmathsetmacro\xpos{
            (\columncount-1)*(\imagewidth / \pgfkeysvalueof{/tikz/zoomboxarray columns} + \pgfkeysvalueof{/tikz/zoomboxarray inner gap} / \pgfkeysvalueof{/tikz/zoomboxarray columns} ) + \pgflinewidth
        }
        \pgfmathsetmacro\ypos{
            (\rowcount-1)*( \imageheight / \pgfkeysvalueof{/tikz/zoomboxarray rows} + \pgfkeysvalueof{/tikz/zoomboxarray inner gap} / \pgfkeysvalueof{/tikz/zoomboxarray rows} ) + 0.5*\pgflinewidth
        }
        \edef\dospy{\noexpand\spy [
            #1,
            zoombox paths/.append style={
                black and white pattern=\patternnumber
            },
            every spy on node/.append style={#1},
            x=\imagewidth,
            y=\imageheight
        ] on (#2) in node [anchor=north west] at ($(zoomboxes container.north west)+(\xpos pt,-\ypos pt)$);}
        \dospy
        \pgfmathtruncatemacro\pgfmathresult{ifthenelse(\columncount==\pgfkeysvalueof{/tikz/zoomboxarray columns},\rowcount+1,\rowcount)}
        \global\let\rowcount=\pgfmathresult
        \pgfmathtruncatemacro\pgfmathresult{ifthenelse(\columncount==\pgfkeysvalueof{/tikz/zoomboxarray columns},1,\columncount+1)}
        \global\let\columncount=\pgfmathresult
        \ifblackandwhitecycle
            \pgfmathtruncatemacro{\newpatternnumber}{\patternnumber+1}
            \global\edef\patternnumber{\newpatternnumber}
        \fi
    \end{scope}
}

\newtheorem{theorem}{Theorem}
\newtheorem{lemma}[theorem]{Lemma}
\newtheorem{proposition}[theorem]{Proposition}
\newtheorem{corollary}[theorem]{Corollary}

\newtheorem{remark}{Remark}
\newtheorem{assumption}{Assumption}     

\newcommand{\tnt}{\Tilde{T}_n}
\newcommand{\tnhat}{\hat{T}_n}

\newcommand{\gnhat}{\hat{G}_n}

\newcommand{\tnas}{T_{n,1}^{*}}
\newcommand{\tnbs}{T_{n,2}^{*}}

\newcommand{\gi}{\hat{g}_1(i)}
\newcommand{\gj}{\hat{g}_1(j)}
\newcommand{\ggij}{\hat{g}_2(i,j)}

\newcommand{\hi}{\hat{H}_1(i)}

\newcommand{\hhij }{\hat{H}_2(i,j)}

\newcommand{\cfhat}{\hat{T}_n}
\newcommand{\cfbm}{\hat{T}_{n,M}^*}
\newcommand{\cfbq}{\hat{T}_{n,Q}^*}
\newcommand{\cfbl}{\hat{T}_{n,L}^*}

\newcommand{\mbm}{\texttt{MB-M}\xspace}
\newcommand{\mbq}{\texttt{MB-Q}\xspace}
\newcommand{\mbl}{\texttt{MB-L}\xspace}
\newcommand{\ew}{\texttt{EW}\xspace}
\newcommand{\smg}{\texttt{SM-G}\xspace}
\newcommand{\sbm}{\texttt{SBM-G}\xspace}
\newcommand{\sub}{\texttt{SS}\xspace}
\newcommand{\ls}{\texttt{LS}\xspace}
\newcommand{\eg}{\texttt{EG}\xspace}

\newcommand{\mblbts}{\texttt{MB-L-SWR}\xspace}
\newcommand{\mblapx}{\texttt{MB-L-apx}\xspace}

\newcommand{\dnm}{\Delta_n(m)}

 \newcommand{\cov}{\mathrm{cov}}

 \newcommand{\mzx}{\mathcal{M}(n, \rho_n , R)}

\newcommand{\gjt}{\gamma_j(t)}
\newcommand{\y}{Y}

\newcommand{\KL}[2]{\text{err}(#1,#2)}
\newcommand{\ione}{I_{1,n}}
\newcommand{\itwo}{I_{2,n}}
\newcommand{\ithr}{I_{3,n}}
\newcommand{\ifour}{I_{4,n}}
\newcommand{\ro}{\lambda}
\newcommand{\lij}{\ell_{i,j}}

\newcommand{\tij}{\theta_{1,i,j}}
\newcommand{\ttij}{\theta_{2,i,j}}
\newcommand{\tttij}{\theta_{3,i,j}}

\newcommand{\kn}{K_{2,n}}
\newcommand{\et}{e^{-t^2/2}}
\newcommand{\del}{\Delta}
\newcommand{\lnone}{L_{1,n}}
\newcommand{\lntwo}{L_{2,n}}
\newcommand{\rnf}{R_{n,4}}

\newcommand{\ttnl}{\Tilde{T}_{n,L}^*}
\newcommand{\tdu}{\tilde{T}_n}
\newcommand{\tdh}{\tilde{H}_1}

\newcommand{\hpi}{H_{\pi}}

\newcommand{\hsxi}{H(A_{S\cup i})}
\newcommand{\hoi}{H_1(i)}

\newcommand{\si}{\mathbb{S}_{-i}}
\newcommand{\bb}[1]{\left(#1\right)}    
\newcommand{\bbb}[1]{\left[#1\right]}

\newcommand{\hatsigf}{\hat{\sigma}_f}
\newcommand{\tdsigf}{\tilde{\sigma}_f}
\newcommand{\tdsigi}{\tilde{\sigma}_{n,i}}
\newcommand{\bxx}{x}

\newcommand{\bu}{u}
\newcommand{\bhatu}{\hat{u}}
\newcommand{\bmu}{\mu}

\newcommand{\cbu}{\check{u}}
\newcommand{\cbmu}{\check{\mu}}
\newcommand{\pbu}{\bold{u}'}

\newcommand{\lm}{\lambda}

\newcommand{\sigf}{\sigma_f}

\usepackage{amsmath}
\usepackage{graphicx,psfrag,epsf}
\usepackage{enumerate}
\usepackage{natbib}

\usepackage{authblk,verbatim}            

\usepackage{url} 

\usepackage{amssymb,bbm,amsmath} 
\usepackage{bm}

\usepackage{amsfonts,amsmath,amssymb,bbm}

\usepackage[font=footnotesize,labelfont=bf]{caption}
\usepackage{subcaption}
\usepackage{float,mathtools}
\usepackage[shortlabels]{enumitem}
\usepackage{xcolor}
\usepackage{multirow}
\usepackage{booktabs}
\usepackage{varwidth}
\usepackage{array}

\usepackage{natbib}

\usepackage[plain,noend]{algorithm2e}
\makeatletter
\renewcommand{\algocf@captiontext}[2]{#1\algocf@typo. \AlCapFnt{}#2} 
\def\@algocf@capt@plain{top}
\renewcommand{\algocf@makecaption}[2]{%
  \addtolength{\hsize}{\algomargin}%
  \sbox\@tempboxa{\algocf@captiontext{#1}{#2}}%
  \ifdim\wd\@tempboxa >\hsize
    \hskip .5\algomargin%
    \parbox[t]{\hsize}{\algocf@captiontext{#1}{#2}}
  \else%
    \global\@minipagefalse%
    \hbox to\hsize{\box\@tempboxa}
  \fi%
  \addtolength{\hsize}{-\algomargin}%
}
\makeatother

\usepackage[plain,noend]{algorithm2e}

\makeatletter
\renewcommand{\algocf@captiontext}[2]{#1\algocf@typo. \AlCapFnt{}#2} 
\def\@algocf@capt@plain{top}
\renewcommand{\algocf@makecaption}[2]{%
  \addtolength{\hsize}{\algomargin}%
  \sbox\@tempboxa{\algocf@captiontext{#1}{#2}}%
  \ifdim\wd\@tempboxa >\hsize
    \hskip .5\algomargin%
    \parbox[t]{\hsize}{\algocf@captiontext{#1}{#2}}
  \else%
    \global\@minipagefalse%
    \hbox to\hsize{\box\@tempboxa}
  \fi%
  \addtolength{\hsize}{-\algomargin}%
}
\makeatother

\begin{document}

\def\spacingset#1{\renewcommand{\baselinestretch}%
{#1}\small\normalsize} \spacingset{1}


\newcommand{\mytitle}{Trading off Accuracy for Speedup: Multiplier Bootstraps for Subgraph Counts}

\if1\blind
{
  \title{\bf \mytitle}
  
  \author{Qiaohui Lin \vspace{-0.2cm}\\
    Department of Statistics and Data Sciences,\\
The University of Texas at Austin\\

    Robert Lunde\\
    Department of Statistics,\\
University of Michigan
    and \\
    Purnamrita Sarkar \\
    Department of Statistics and Data Sciences,\\
The University of Texas at Austin}
  
  





  \maketitle
} \fi

\if0\blind
{
  \bigskip
  \bigskip
  \bigskip
  \begin{center}
    {\LARGE\bf \mytitle}


    
\end{center}
  \medskip
} \fi

\addtolength{\topmargin}{-10pt}%


\begin{abstract}
We propose a new class of multiplier bootstraps for count functionals, ranging from a fast, approximate linear bootstrap tailored to sparse, massive graphs to a quadratic bootstrap procedure that offers refined accuracy for smaller, denser graphs.   For the fast, approximate linear bootstrap, we show that $\sqrt{n}$-consistent inference of the count functional is attainable in certain computational regimes that depend on the sparsity level of the graph.  Furthermore, even in more challenging regimes, we prove that our bootstrap procedure  offers valid coverage and vanishing confidence intervals.  For the quadratic bootstrap, we establish an Edgeworth expansion and show that this procedure offers higher-order accuracy under appropriate sparsity conditions. We complement our theoretical results with a simulation study and real data analysis and verify that our procedure offers state-of-the-art performance for several functionals.        
\end{abstract}

\noindent%
\vfill


\newpage
\spacingset{1.45} 


 \section{Introduction}

Count functionals play a pivotal role in the analysis of network data.  In biological networks, it is believed that certain subgraphs  may  represent functional subunits within the larger system \citep{milo-network-motifs,chen-yuan-protein-protein-interaction-network,daudin-exceptionality-motifs,kim2014}.  In social networks, the frequency of triangles provides information about the likelihood of mutual friendships \citep{newman-collaboration-network,meyers-social-network-twitter,ugander-facebook-graph}. At a more theoretical level, count functionals may be viewed as network analogs of the moments of a random variable. Thus, a method of moments approach may be used to estimate the underlying model under suitable conditions~\citep{Bickel-Chen-Levina-method-of-moments}.  


While real-world networks share many qualitative features (see e.g. \citet{newman-structure-networks}), they often vary substantially in terms of size, given by the number of vertices in the network, and sparsity level, given by the number of edges relative to the number of vertices. For networks of small to moderate size, inferential methods that are highly accurate are advantageous; for sparse, massive networks, one needs to simultaneously consider computational tractability. 

To meet these diverse needs in real-world applications with user-friendly computational tools, we develop a new family of bootstrap procedures for uncertainty estimation of count functionals of networks. Our methods range from a very fast approximate linear bootstrap to a fast quadratic bootstrap procedure that offers improved accuracy for moderately sparse networks. Both procedures may be viewed as approximations\footnote{More precisely, the linear and quadratic bootstraps may be viewed as first and second-order terms of a Hoeffding decomposition for the multiplier bootstrap, respectively.} to a multiplier bootstrap method in which each potential subgraph in the network is perturbed by the product of independent multiplier random variables. This multiplier bootstrap is closely related to a bootstrap method for U-statistics (see for example, \citet{bose-chatterjee-u-stats-resampling}). Under the sparse graphon model (see Section \ref{sec:sparse-graphon}), subgraph counts may be viewed as U-statistics perturbed by asymptotically negligible noise, allowing the adaptation of bootstrap methods for U-statistics to the network setting.

One of the main theoretical contributions of our paper is that we establish the statistical properties of a computationally efficient randomized approximation for computing local subgraph counts.   This  approximation alleviates the main computational bottleneck for our linear bootstrap procedure. By invoking these results in our analysis of the bootstrap, we characterize the relationship between the sparsity level of the graph and the degree of randomization allowed to achieve either asymptotically efficient confidence intervals or the less ambitious aim of consistent inference via vanishing confidence intervals.
 In the worst case, counting a local subgraph of size $r$ (e.g. number of r-cycles a node is involved in) in a $n$ node network has a computation complexity of $O(n^{r-1})$. In contrast, our methods can be used to reduce this time substantially depending on the nature of the subgraph and sparsity level; for triangles for example, the computational complexity of our approximation is $O(Nn^2 \rho_n\log(n\rho_n))$, where $N$ is the number of random permutations and $\rho_n$ is the probability of an edge.
\bk

While randomized linear bootstrap is more scalable, it does not yield higher-order correctness, which can be achieved with certain methods for sufficiently dense graphs (see~\citet{zhang-xia-network-edgeworth}). As the second theoretical contribution of our paper, we propose the quadratic bootstrap and establish conditions under which it is higher order correct. 
To this end, we establish an Edgeworth expansion for the quadratic bootstrap and show that it is close to the Edgeworth expansion of the sampling distribution. 
We bypass the typical Cram\'{e}r's condition required by~\cite{zhang-xia-network-edgeworth} in their analysis of existing bootstrap methods by considering a continuous multiplier that matches the first three moments of the data. It is well-known that continuous random variables satisfy Cram\'{e}r's condition. 

\bk

\bk

 In addition to obtaining Edgeworth expansions for count statistics, we also obtain Edgeworth expansions for smooth functions of U-statistics. We show that, under suitable sparsity assumptions, the cumulative distribution function of smooth functions arising from the quadratic bootstrap match this asymptotic expansion and are therefore higher-order correct. In this setting, obtaining analytical expressions for Edgeworth expansions are cumbersome, whereas the bootstrap is automatic and user-friendly.  

We will now provide a roadmap for the rest of the paper.  In Section~\ref{sec:related-work}, we discuss related work, focusing on the emerging area of resampling methods for network data.  The problem setting and our bootstrap proposal is introduced in Section~\ref{sec:setup}.  In Section~\ref{sec:main}, we present our main results, which establish higher-order correctness for our bootstrap procedures.  In Section~\ref{sec:exp}, we present experiments on both simulations and real data, which show that our procedures exhibit strong finite-sample performance in various settings.  
\section{Related work}
\label{sec:related-work}
The first theoretical result for resampling network data was attained by~\citet{Bhattacharyya-subsample-count-features}.  Their subsampling proposals involve expressing the variance of a count functional in terms of other count functionals and estimating the non-negligible terms through subsampling.  \citet{subsampling-sparse-graphons} show that it is also possible to conduct inference using quantiles of the subsampling distribution as in~\citet{politis-romano-subsampling-minimal-assumptions}.   \citet{green-shalizi-network-bootstrap} propose a bootstrap based on the empirical graphon.  \citet{network-jackknife-theory} establish the validity of the network jackknife for count functionals. 

  \citet{levin-levina-rdpg-bootstrap} study a two-step procedure that is closely related to our linear bootstrap procedure. The above authors propose estimating the latent positions with the adjacency spectral embedding in the first step (see, for example, \citet{athreya-rdpg-survey}) and resampling the corresponding U-statistic with the estimated positions in the second step. They derive theoretical results under the assumption that the rank of the random dot product graph model is known and finite.  In contrast, our procedures do not impose assumptions on the spectral properties of the underlying graphon. \bk



~\citet{zhang-xia-network-edgeworth} establish conditions under which the empirical graphon bootstrap exhibits higher-order correctness.  For establishing higher-order correctness for resampling methods studied in this work, Cram\'{e}r's condition is required, which is restrictive for network models.  
In contrast, their empirical Edgeworth expansion proposal does not require Cram\'{e}r's condition. 
On the mathematical side, the analysis of our multiplier bootstrap involves Edgeworth expansions for weighted sums. 
Prior work (c.f. \citet{bai-zhao-independent-edgeworth} and \citet{liu-non-iid-bootstrap}) suggests that establishing sharp rates of convergence for the independent but non-identically distributed sequences is more difficult, with the above references establishing a  $o(n^{-1/2})$ error bound instead of the $O(n^{-1})$ bound for i.i.d. sequences.  Therefore, it is not surprising that our quadratic bootstrap proposal, while offering improved performance over first-order correct methods, does not quite match the rate of the empirical Edgeworth expansions studied by \citet{zhang-xia-network-edgeworth}.   


From the computational standpoint, \citet{chen2019randomized} presented a randomized algorithm to estimate high-dimensional U-statistics from a subsample of subsets from the set of all subsets of a given size. We propose a different sampling method that exploits the structure of U-statistics and draws random permutations instead of subsets. Empirically, we show that this method provides faster computation over subset sampling.

\section{Problem setup and notation}
 \label{sec:setup}
\subsection{The sparse graphon model}
\label{sec:sparse-graphon}
Let $\{A^{(n)}\}_{n \in \mathbb{N}}$ denote a sequence of $n \times n$ binary adjacency matrices and let $w:[0,1]^2 \mapsto \mathbb{R}$ be a symmetric measurable function such that $\int_0^1 \int_0^1 w(u,v) \ du \ dv = 1$ and $ w(u,v) \leq C$ for some $1 \leq C < \infty$.  We assume that $A^{(n)}$ is generated by the following model:
\begin{align}
\label{eq:sparse-graphon-model}
A_{ij}^{(n)} = A_{ji}^{(n)} \sim \mathrm{Bernoulli}(\rho_n w(X_i, X_j))
\end{align}
where $X_i, X_j \sim \mathrm{Unif}[0,1]$, $\rho_n \rightarrow 0$, and $A_{ii}^{(n)} = 0$.  While closely related models were considered by \citet{bollobas-inhomogenous-graphs}, \citet{hoff-raftery-handcock-latent-space-model}, and \citet{borgs-lp-part-one} this particular parameterization was introduced by  \citet{Bickel-Chen-on-modularity}. We will refer to (\ref{eq:sparse-graphon-model}) as the sparse graphon model.  Sparse graphons are a very rich class of models, subsuming many widely used models, including stochastic block models and their variants \citep{holland-sbm, karrer-newman-dcsbm, airoldi-mmsb}, and  (generalized) random dot product graphs \citep{young-schneiderman-rdpg, generalized-rdpg}.  More generally, sparse graphons are natural models for graphs that exhibit vertex exchangeability; the functional form is motivated by representation theorems for exchangeable arrays established by \citet{aldous-representation-array} and \citet{hoover-exchangeability}.  The parameter $\rho_n = P(A_{ij} = 1)$ determines the sparsity level of the sequence $\{ A^{(n)} \}_{n \in \mathbb{N}}$.  Many real world graphs are thought to be sparse, with $o(n^2)$ edges; $\rho_n \rightarrow 0$ is needed for graphs generated by (\ref{eq:sparse-graphon-model}) to exhibit this behavior. 

While boundedness of the graphon is a common assumption in the statistics literature (see, for example, the review article by \citet{gao-ma-minimax-network-analysis}), it should be noted that unbounded graphons are known to be more expressive.  As noted by \citet{borgs-lp-part-one}, unboundedness allows graphs that exhibit power-law degree distributions, a property that bounded graphons fail to capture. For mathematical expedience, in the present article, we focus on the bounded case, but we believe that our analysis may be extended to sufficiently light-tailed unbounded graphons as well. We use standard asymptotic notations such as $O()$,$o()$ and $\omega()$ throughout the paper. 

\subsection{Count functionals}
Now we will introduce notation related to our functional of interest. 
Let $R$ denote the adjacency matrix of a subgraph of interest, with $r$ vertices and $s$ edges. Let $A_{i_1, \ldots, i_r}^{(n)}$ denote the adjacency matrix formed by the node subset $\{i_1, \ldots, i_r\}$ and for each such $r$-tuple, define the following function: 
\begin{align*}
H(A_{i_1, \ldots, i_r}^{(n)}) := \mathbbm{1}(A_{i_1, \ldots, i_r}^{(n)} \cong R)
\end{align*}
where we say that $A_{i_1, \ldots, i_r}^{(n)} \cong R$ if there exists a permutation function $\pi$ such that $A_{\pi(i_1), \ldots, \pi(i_r)} = R$.  Our count functional, which we denote $\hat{T}_n(R)$, or $\hat{T}_n$ when there is no ambiguity, is formed by averaging over all $r$-tuples in the graph.
\begin{align}\label{eq:cf}
\cfhat := \frac{1}{{n \choose r} } \sum_{1 \leq i_1 < i_2 < \ldots < i_r \leq n} H(A_{i_1, \ldots, i_r}^{(n)}) 
\end{align}
Let $h(X_{i_1}, \ldots X_{i_r})$ denote the conditional expectation of $H(A_{i_1, \ldots, i_r}^{(n)})$.
Now, define the following (conventional) U-statistic:
\begin{align*}
T_n := \frac{1}{{n \choose r}} \sum_{1 \leq i_1 < i_2 < \ldots < i_r \leq n} h(X_{i_1}, \ldots X_{i_r})
\end{align*}
 For notational convenience we will refer to $h(X_{i_1}, \ldots, X_{i_r})$ by $h(X_S)$, where $S$ is the subset $\{i_1,\dots, i_r\}$.  Denote $\theta_n:=\E\{h(X_S)\}$. We see that $\theta_n/\rho_n^s\rightarrow\mu$. This can be thought of as a normalized subgraph density that we want to infer. The normalization by $\rho_n^s$ is to ensure that our functional converges to an informative non-zero quantity.

  By a central limit theorem for U-statistics \citep{hoeffding1948}, it can be shown that $(T_n - \theta_n)/\sigma_n$ is asymptotically Gaussian.  Here we have:
  \begin{align}\label{eq:taun}
      \tau_n^2=\var[\E\{h(X_S)\mid X_1\}],\ \  \  \sigma_n^2=r^2\tau_n^2/n
  \end{align}
 Furthermore, \citet{Bickel-Chen-Levina-method-of-moments} show that, 
  $( \cfhat - T_n)/\sigma_n = o_P(1)$  under mild sparsity conditions for a wide range of subgraphs.   Thus, we may view $(\cfhat - \theta_n)/\sigma_n = (\cfhat- T_n)/\sigma_n + (T_n - \theta_n)/\sigma_n$ as a U-statistic perturbed by asymptotically negligible noise.

\subsection{Preliminaries of proposed bootstrap procedures}
\label{subsec:proposed-bootstrap}
In order to estimate the subgraph density, we will consider the following multiplier bootstrap procedures. In what follows let $\w_1, \ldots \w_n$ be i.i.d. continuous random variables with mean $\mu = 1$ and central moments $\mu_2 =1$, and $\mu_3 =1$.  An example of such a random variable is the product $Z$ of two independent Normal random variables $X$ and $Y,$ defined below: 
\begin{align}\label{eq:gaussproduct}
    X\sim N(1,1/2) \qquad Y\sim N(1,1/3) \qquad Z=XY
\end{align}

Let $\w_{i_1 \cdots i_r}$ denote $\w_{i_1} \times \ldots \times \w_{i_r}$ and define the following multiplicative bootstrap:
\begin{align}
\label{eq:multiplicative-bootstrap}
\cfbm = \cfhat+ \frac{1}{{n\choose r}} \sum_{1 \leq i_1 < i_2 < \ldots i_r} \xi_{i_1 \cdots i_r}  \cdot \left\{ H(A_{i_1, \ldots, i_r}^{(n)}) - \cfhat \right\}
\end{align}

Our multiplicative bootstrap is motivated by Hoeffding's decomposition (see Supplement Section~\ref{sec:supp:lemmaone}).  The first two terms of the decomposition for $T_n - \theta_n$ are given by:
\begin{align*}
g_1(X_i) &= E\{h(X_i, X_{i_2} \ldots X_{i_r}) \mid X_i\} - \theta_n 
\\ g_2(X_i, X_j) &= E\{h(X_i, X_j, X_{i_3} \ldots X_{i_r}) \mid X_i, X_j\} -g_1(X_i)-g_1(X_j) - \theta_n,
\end{align*}
leading to the representation:
\begin{align}\label{eq:hoeffding-noiseless}
T_n - \theta_n = \frac{r}{n}\sum_{i=1}^n g_1(X_i) + \frac{r(r-1)}{n(n-1)} \sum_{i < j}  g_2(X_i,X_j)  + o_p\left(\frac{\rho_n^s}{n}\right) 
\end{align}
Similarly, conditional on the data, it can be shown that we have the following bootstrap analog. Let:
\begin{align}
\hat{g}_1(i) &= \frac{1}{{n-1 \choose r-1 }} \sum_{1 \leq i_2 < \ldots i_r \leq n, i_u \neq i } \bigl\{H(A_{i,i_2, \ldots i_r}) - \cfhat \bigr\}\label{eq:g1hat}
\\ \tilde{g}_2(i,j) &= \frac{1}{{n-2 \choose r-2 }} \sum_{1 \leq i_3 < \ldots i_r \leq n, i_u \neq i,i_u \neq j }  \bigl\{H(A_{i,i_2, \ldots i_r}) - \cfhat \bigr\}   \\
\hat{g}_2(i,j)&=\tilde{g}_2(i,j)-\hat{g}_1(i) -\hat{g}_1(j) \label{eq:g2hat}
\end{align}
Furthermore, Eq~\ref{eq:g1hat} can be used to standardize the bootstrap replicates using the following estimate of $\tau_n$ (Eq~\ref{eq:taun}):
\begin{align}\label{eq:tauhat}
\hat{\tau}_n^2=\sum_i \frac{\hat{g}_1(i)^2}{n}
\end{align}
We now present the Hoeffding decomposition for our bootstrap statistic. The proof is deferred to the Supplement Section~\ref{sec:supp:lemmaone}.



 \begin{lemma}
 \label{lemma:hoeffding-bootstrap}
  We have the following decomposition: 
  \begin{align}\label{eq:boot-hoeff}
 \begin{split}
 \cfbm - \cfhat &= \frac{r}{n}\sum_{i=1}^n (\w_i-1) \cdot \hat{g}_1(i)  + \frac{r(r-1)}{n(n-1)}\sum_{i<j } \ (\w_i \ \w_j - \w_i-\w_j+1) \cdot \tilde{g}_2(i,j)    
 \\ & \ \ \ \ \ +  O_P\left(\rho_n^sn^{-1/2}\delta(n,\rho_n,R) \right),
 \end{split}
 \end{align}
 where $\delta(n,\rho_n,R)$ is defined as follows: 
 
 $\begin{aligned}
     \delta(n,\rho_n,R) =\begin{cases} 
       \dfrac{1}{n\rho_n} & \text{R is acyclic} \\
       \dfrac{1}{n\rho_n^{3/2}} & \text{R is a simple cycle} .
  \end{cases}
\end{aligned}$
 \end{lemma}

Although the quadratic term in the above expansion may seem different from Eq~\ref{eq:hoeffding-noiseless}, Some manipulation yields that $\sum_{i<j }( \ \w_i \ \w_j - \w_i-\w_j+1) \cdot \tilde{g}_2(i,j)  $ is equivalent to $\sum_{i<j}(\w_i\w_j-1) \tilde{g}_2(i,j)-(\w_i-1) \cdot \hat{g}_1(i) - (\w_i-1) \cdot \hat{g}_1(j)$, which is similar to the corresponding term in the Hoeffding decomposition of the U statistic (see Eq~\ref{eq:hoeffding-noiseless}).

Viewing $\hat{g}_1(i)$ and $\hat{g}_2(i,j)$ as estimates of $g_1(X_i)$ and $g_2(X_i, X_j)$, respectively, it is clear that that our weighted bootstrap version encapsulates important information about $\cfhat - \theta_n$. The above decomposition also suggests that one may approximate the non-negligible terms more directly.  Ignoring the remainder term, we arrive at the linear and quadratic bootstrap estimates:
\begin{align}
\cfbl &= \cfhat +  \frac{r}{n}\sum_{i=1}^n (\w_i-1) \cdot \hat{g}_1(i)   \label{eq:cfbl}\\ \cfbq &= \cfbl+ \frac{r(r-1)}{n(n-1)}\sum_{i<j } \ (\w_i \ \w_j - \w_i-\w_j+1) \cdot \tilde{g}_2(i,j). \label{eq:cfbq} 
\end{align}


Now that we have introduced the main concepts, we are ready to present the  our bootstrap procedures. We first present the results on our fast linear bootstrap method.
\bk
\section{Proposed algorithms}\label{sec:proposed-alg}
In this section, we present a fast linear bootstrap method using Eq~\ref{eq:cfbl}. Recall that the multiplicative bootstrap requires to precompute $\cfhat$ and $\hat{g}_1(i)$ for all $i$. This computation is $O(n^r)$ in the worst case. In addition to this, the computation complexity for MB-L is $O(Bn)$. Therefore, in what follows, our goal is to reduce the precomputation time.

\subsection{Fast linear bootstrap}\label{subsec:fast-linear-boot}
We propose a randomized approximation for $\cfhat$ and $\hat{g}_1(i)$. The main idea is that an average over all size $r$ subset can be written as an average over $n!$ permutations (see~\cite{hoeffding1948,subsampling-sparse-graphons}). 
For any $i \in \{1,\ldots,n\}$, denote the set of all subsets of size $r-1$ taken from  $\{1,\ldots,i-1,i+1,\ldots n\}$ as $\si$. 
Denote  $H(A_{i,i_2, \ldots i_r})$ for $S=\{i_2, \ldots i_r\}  \in \mathbb{S}\{-i\}$ as $\hsxi$. 
Denote 
\begin{align*}
     &\hoi = \frac{1}{{n-1 \choose r-1 }} \sum_{S \in \mathbb{S}\{-i\}} \hsxi.
\end{align*}
\bk 

One can also write $H_1(i)$ as follows:
\begin{align*}
    H_1(i)=\frac{1}{(n-1)!}\sum_{\pi} H_\pi(i).
\end{align*}
Here $H_\pi(i)=\frac{\sum_{S\in\mathbb{S}_{\pi} }\hsxi}{\frac{n-1}{r-1}},$ where $\mathbb{S}_{\pi}$ denotes the set of all disjoint subsets  $\{\pi_{(i-1)(r-1)+1},\dots, \pi_{i(r-1)}\}$, $i=1,\dots, \frac{n-1}{r-1}$ obtained from permutation $\pi$. 
Now let  $\pi_j$ be a permutation picked with replacement and uniformly at random from the set of all permutations of $\{1,\dots,n\}\setminus i$. 

Our randomized algorithm makes use of this structure and draws $j=1,\ldots,N$ independent permutations $\pi_j$. 
We compute
\begin{equation}\label{eq:define-h1-tilde}
\tdh(i)=\frac{\sum_{j} H_{\pi_j}(i)}{N},\qquad \tdu=\frac{1}{n}\sum_{i=1}^n\tdh(i).
\end{equation}

To calculate $\tdh(i)$, for each $i$, we permute the node set excluding $i$ for $N$ times and for each of these permutations $\pi$ we check the disjoint set $\mathbb{S}_{\pi}$ for count functionals.  Thus, the complexity for calculating $\tdh(i)$ is now $O\left(N\frac{n}{r}\right)$. From $\{\tdh(i)_{i=1}^n\}$, $\tdu$ is calculated from their mean and $\tilde{\tau}_n$ is defined as
\begin{align}\label{eq:define-taun-tilde}
    \tilde{\tau}_n^2= \frac{\sum_{i=1}^n\{\tdh(i)-\tdu\}^2}{n^2}.
\end{align}
The linear bootstrap uses $\tilde{T}_{n,L}$ by plugging in $\tilde{g}_1(i)=\tilde{H}_1(i)-\tilde{T}_n$ and $\tilde{T}_n$ in Eq~\ref{eq:cfbl}.
\begin{align}\label{eq:define-tnl-tilde}
  \ttnl  =  \Tilde{T}_n + \frac{c_n}{n}\sum_{i=1}^n (\xi_i-1)\{\tdh(i)-\tdu\},
\end{align}
where $c_n=r$ when $Nn\rho_n^s\gg 1$ and $c_n=1$ otherwise.

\begin{remark}
As we show in Theorem~\ref{thm:clt-approx}, the variance of the approximate count statistic consists of the true limiting variance of the underlying U-statistic and the noise variance. When  $N= \Theta(1/n\rho_n^s)$, there is a phase transition from the case where the noise variance dominates to one where the signal variance dominates. The factor $c_n$ helps match the two variances before and after this phase transition. 
\end{remark}

We denote the bootstrap based on Eq~\ref{eq:define-tnl-tilde} by \mblapx and explicitly provide the algorithm in the Supplement (see Algorithm~\ref{alg:apx}). 

\begin{remark}\label{remark:sparsity}
We can adapt the computation of \mblapx to exploit the sparsity in networks. Take for example the simple case of triangles. Instead of scanning over a given permutation of $n-1$ nodes to calculate $\tdh(i)$ for node $i$ (see Eq~\ref{eq:define-h1-tilde}), we can ``sample'' the positions of the neighbors of $i$ in a random permutation, and then check for consecutive neighbors among these positions, thus reducing the computation from $O(n)$ per permutation to $O(n\rho_n\log(n\rho_n))$.  This idea can be extended to other subgraphs as well. We provide details of this approach for counting triangles, two stars and three paths in the Supplement Section~\ref{sec:suppalgapx}. For large sparse graphs (with about 50,000 nodes and 750,000 edges) our approach computes three paths in about 850 seconds, where brute force counting did not finish in a day (see Section~\ref{sec:realdata}). This method reduces the complexity of approximating count statistics to $O(n^2N\rho_n\log(n\rho_n))$, while brute force counting for a three-path can be as large as $O(n^4)$ and become too slow and impractical for large sparse networks. 
\end{remark}

\bk

\subsection{Higher-order correct bootstrap procedures}
In this section, we present our proposed quadratic, and multiplicative algorithms (\mbq, and \mbm). The detailed pseudocode can be found in the Supplement (Algorithm~\ref{alg:ho} in Section~\ref{sec:suppalgapx}).

For a given network, we first compute $\cfhat$ and $\hat{\tau}_n$ (see Eqs~\ref{eq:cf},~\ref{eq:tauhat}).
For each algorithm, we generate $B$ samples of $n$ weights $\{\w^{(j)}_i,i=1,\dots, n\}_{j=1}^B$ from the \texttt{Gaussian Product} distribution (see beginning of Section~\ref{subsec:proposed-bootstrap}). For each of these, \mbm, \mbq, and \mbl respectively values $\cfbm,\cfbq$ and $\cfbl$. From the $B$ values one then constructs the CDF of the statistic in question, after shifting and normalizing it appropriately. 
While we divide by $\frac{r}{n^{1/2}}\hat{\tau}_n$, our statistic is not studentized, yielding a different expansion from previous work. Conditioned on the data, $\hat{\tau}_n$ is constant for the bootstrap samples.

Note that \mbm is computationally expensive since it involves computing the expression in Eq~\ref{eq:multiplicative-bootstrap} for each sample of the bootstrap. The worst-case complexity of evaluating all ${n\choose r}$ subsets of nodes is $n^r$. For $B$ bootstrap samples, the worst-case timing of \mbm will be $B n^r$. In comparison, for \mbl and \mbq, we can precompute the $\hat{g}_1(i)$ and $\hat{g}_2(i,j)$ values in $O(n^r)$ time. After that, the time per bootstrap sample is linear for \mbl and quadratic for \mbq. Thus worst-case computational complexity for a dense network for \mbm, \mbq, and \mbl is $O(Bn^r)$, $Bn^2$ and $Bn$ respectively, \textit{excluding} precomputation time (O($n^r$) in the worst case). In contrast, the approximate linear bootstrap is much faster (see Section~\ref{subsec:fast-linear-boot}).

\section{Main results}
\label{sec:main}



  \subsection{Theoretical guarantees for approximate linear bootstrap}
 In this section, we show that the linear bootstrap statistic using the approximate moments in Eq~\ref{eq:define-h1-tilde} is indeed first-order correct under appropriate sparsity conditions as long as $N$ is large enough.  In this section, we make the following assumptions:
 \begin{assumption}\label{ass:linear}
  \begin{enumerate}
      \item[(a)] $\tau_n/\rho_n^s\geq c>0$, for some constant $c$.
    \item[(b)] $0<w(u,v)<C$, for some constant $C$.
  \end{enumerate}
 \end{assumption}
 The first condition is a standard non-degeneracy assumption for U-statistics.  Boundedness of the graphon is also a common assumption; this assumption can be relaxed at the cost of longer proofs.  
 In what follows, let:
\begin{align}\label{eq:def-sigmatildeAX}
\tilde{\sigma}^2_{|A,X}:= \mathrm{Var}(\tilde{T}_n - \hat{T}_n \ | \ A, X), \ \ \ \tilde{\sigma}_n^2 := E(\tilde{\sigma}^2_{|A,X})
\end{align}
Our first result is a central limit theorem for approximate count statistics.  We show that, so long as $N \gg 1/n^2 \rho_n^s$, the appropriately scaled and centered count functional converges in distribution to a standard Normal.  This threshold ensures that the random conditional variance $\tilde{\sigma}^2_{|A,X}$ concentrates appropriately around its expectation $\tilde{\sigma}_n^2$.  Furthermore, our theorem below establishes two computation/sparsity regimes. When $N \gg 1/n \rho_n^s$, the limiting variance corresponds to the limiting variance of $\hat{T}_n$, given by $\sigma_n^2$, whereas when $1/n^2 \rho_n^s \ll  N \ll 1/n \rho_n^s$, the limiting variance is dominated by $ \tilde{\sigma}_n^2$, which captures the the effect of randomization.  The main technical difficulty is establishing appropriate concentration of $\tilde{\sigma}^2_{|A,X}$ when $N$ is small; our proof involves delicate combinatorial analysis and exploiting the degeneracy of an appropriate U-statistic.  See Supplement S2 for details. 
 \begin{theorem}\label{thm:clt-approx}
 Suppose that $\rho_n \rightarrow 0$ and either $R$ is acyclic and $n \rho_n \rightarrow \infty$ or $R$ is a simple cycle and $n^{r-1} \rho_n^r \rightarrow \infty$.  Moreover, suppose that $N \gg 1/n^2\rho_n^s$.  
 \begin{enumerate}
\item[(a)] (Variance of approximate count statistic $\tilde{T}_n$) 
\begin{align*}
\var(\tilde{T}_n) = 
\tilde{\sigma}_n^2 + \sigma_n^2, \text{ \ where \ } \sigma_n^2 = \Theta\left( \frac{\rho_n^{2s}}{n} \right) \text{ \ and \ } \tilde{\sigma}_n^2 = \Theta\left( \frac{\rho_n^{s}}{Nn^2} \right)
\end{align*}
Consequently, when $N\gg \frac{1}{n\rho_n^s}$, $\var(\tilde{T}_n)/\sigma_n^2\rightarrow 1$ and when $\frac{1}{n^2\rho_n^s}\ll N\ll \frac{1}{n\rho_n^s}$, 
$\var(\tilde{T}_n)/\tilde{\sigma}_n^2\rightarrow 1$.
\item[(b)] (CLT for approximate count statistic $\tilde{T}_n$)
 \begin{align*}
 \sup_{u} \left| P\left( \frac{\tilde{T}_n - \theta_n}{ \sqrt{\noisesd^2 + \sigma_n^2}} \leq u  \right) - \Phi(u)    \right| \rightarrow 0
 \end{align*}
 \end{enumerate}
 \end{theorem}
 
 Our next result establishes the consistency of our bootstrap procedure under appropriate conditions.  For a broad range of computational regimes, ranging from $1/n^2 \rho_n^s \ll N \ll 1/n \rho_n^s$ to $N \gg 1/n \rho_n^s$, we show that bootstrap consistently approximates the distribution of $\tilde{T}_n$. In the regime $N \gg 1/n \rho_n^s$, the bootstrap is $\sqrt{n}$-consistent whereas in the other aforementioned regime, the bootstrap offers faster computation while still providing vanishing confidence intervals.   Our proof here again involves combinatorial analysis and tight bounds for appropriate higher-order moments; see Supplement S2 for details. \bk
 
 \begin{theorem}\label{thm:boot-linear-apx}

Suppose that $\rho_n \rightarrow 0$, $r \geq 2$ and either $R$ is acyclic and $n \rho_n \rightarrow \infty$ or $R$ is a simple cycle and $n^{r-1} \rho_n^{r} \rightarrow \infty$. Let $\mathcal{R}$ denote the independent random variables determining the random permutations chosen with replacement.
\begin{enumerate}
    \item[(a)] When $1/n^2\rho_n^s\ll N\ll 1/n\rho_n^s$, we have $\var(\tilde{T}_{n,L}^*|A,X,\mathcal{R})/\tilde{\sigma}_n^2\stackrel{P}{\rightarrow} 1$. In contrast, when $ N\gg 1/n\rho_n^s$, we have $\var(\tilde{T}_{n,L}^*|A,X,\mathcal{R})/\sigma_n^2\stackrel{P}{\rightarrow} 1$.

\item [(b)]
When $1/n^2\rho_n^s\ll N\ll 1/n\rho_n^s$ or $N\gg 1/n\rho_n^s$, we have 
\begin{align*}
\sup_{u} |P( \tilde{T}_{n,L}^* - \tilde{T}_n   \leq u) - P( \tilde{T}_n - \theta_n   \leq u)     | \xrightarrow{P} 0.    
\end{align*}
\end{enumerate}
\end{theorem}

\begin{remark}

     As seen in Theorem~\ref{thm:clt-approx}, $N=\Theta(1/n\rho_n^s)$ is where a phase transition occurs. When $N$ is a smaller (or larger) order than this threshold, the variance of the bootstrapped statistic $\tilde{T}_{n,L}^*$ converges to variance of the approximate count statistic $\tilde{T}_n$. However, at the point of the phase transition, the bootstrap variance does not match the variance of the approximate count statistic.  In some sense, this is unavoidable since U-statistics and averages of independent random variables have fundamentally different structure owing to the dependence of U-statistics; thus, one cannot mimic the structure of both components with this approach simultaneously. From a practical standpoint, this is not a problem since the user can choose $N$ to be either in the smaller or larger order regimes.  Moreover, at the boundary, one can always opt for a conservative estimate of the variance.      
     
     
    
 \end{remark}

\begin{remark}
     The advantage of the fast linear bootstrap over a direct Normal approximation is that it is automatic and more user-friendly.
     Estimating parameters of the limiting normal for count statistics using estimated sparsity parameter or smooth functions of counts such as the transitivity (suitably normalized ratio of triangles and two stars) requires additional analytic calculations which can be cumbersome for a practitioner. The bootstrap presents a user-friendly alternative, where one replaces burdensome calculations by computer simulations.  However, our analysis also sheds light on the properties of a Normal approximation using approximate count statistics. Therefore, the approximate moment computation proposed in this section can be also used to estimate the limiting variance in both the small and large $N$ setting.  
     
\end{remark}

\begin{remark}[Comparison to existing work on approximating $U$- statistics]
	In~\cite{chen2019randomized}, the authors draw $\omega(n)$ subsets of size $r$ from all ${n\choose r}$ subsets with replacement to estimate an incomplete $U$-statistic. The total number of subsets
	we examine for approximating a local count statistic is also $\omega(n)$. Comparing our Theorem~\ref{thm:boot-linear-apx} with their result shows that both
	methods require similar computation to achieve consistency. However, practically, drawing $Nn/r$ subsamples with replacement seems to be slower than drawing a $N$ permutations
	and then dividing each into disjoint subsets (see Fig~\ref{fig:4cycle-timing}). While they also consider the effect of randomization on the Gaussian approximation, our results take into account the sparsity level, and require delicate and novel combinatorial analysis.
	\end{remark}

\begin{remark}[Approximation quality]
Let $\hat{\nu}_n^2$ denote $\var(\ttnl \mid A,X,\mathcal{R})$.
The standardized  bootstrap CI with coverage $\alpha$ constructed from the approximated linear bootstrap $\ttnl$ can be written as $$\bb{\frac{\tilde{T}_n-\hat{\nu}_nz_{\frac{1-\alpha}{2}}}{\rho_n^s},\ \ \frac{\tilde{T}_n-\hat{\nu}_nz_{\frac{1+\alpha}{2}}}{\rho_n^s}}$$ where $z_{\frac{1-\alpha}{2}}$ and $z_{\frac{1+\alpha}{2}}$ are the $\frac{1-\alpha}{2}$ and $\frac{1+\alpha}{2}$ quantiles from the standard normal distributions.
When $1/n^2\rho_n^s\ll N\ll 1/n\rho_n^s$, Theorem~\ref{thm:boot-linear-apx} states that $\var(\ttnl \mid A,X,\mathcal{R})=\Theta(\rho_n^{s}/nN)$ and hence the CI width vanishes, thus leading to consistent estimation of the parameter of inferential interest. 
For estimated $\rho_n$ these intervals can be established using a delta method approach.\bk 
\end{remark}

\begin{remark}[Broader Sparsity Regime]\label{rem:apx}
Theorem~\ref{thm:boot-linear-apx} suggests that the linear bootstrap with a suitably large $N$ gives a consistent estimate of variance under much weaker conditions on sparsity than that needed for establishing higher order correctness (Theorem~\ref{thm:mbq-firstorder} and Corollary~\ref{thm:mbm-firstorder}).  It should be noted that the arguments in \cite{zhang-xia-network-edgeworth} require $\rho_n = \omega(1/\sqrt{n})$ for acyclic graphs and therefore, their convergence rates for empirical Edgeworth expansions do not hold in sparser regimes.  \bk
\end{remark}

\bk

\subsection{Results on higher-order correct bootstrap procedures}
Below, we establish an Edgeworth expansion normalized by the true standard deviation, which is more appropriate for our purposes.  Since estimating the variance leads to a non-negligible perturbation, the polynomials in our expansion differ from those established by \cite{zhang-xia-network-edgeworth}.  All proofs and details are deferred to Supplement Section \ref{sec:suppprop3} and Section \ref{sec:suppthm23}. \bk In what follows, let $F_n(u)$ denote the CDF of $\hat{T}_n$ and $G_n(u)$ denote the Edgeworth expansion of interest, given by:

\begin{equation}\label{eq:edgeworth-gn}
    G_n(u)= \Phi(u)- \phi(u)\frac{(u^2-1)}{6n^{1/2}\tau_n^3}\biggl[E\{g_1^3(X_1)\}+3(r-1)E\{g_1(X_1)g_1(X_2)g_2(X_1,X_2)\} \biggr].
\end{equation}

Furthermore, recall $\tau_n^2=\var[\E\{h(X_S)\mid X_1\}]$ denotes the asymptotic variance of the  U-statistic.  Throughout this section, we will impose the following condition:

 \begin{assumption}\label{ass:sparse}
  For acyclic $R$, $\rho_n = \omega(n^{-1/2})$ and for cyclic $R$,  $\rho_n = \omega(n^{-1/r})$.
 \end{assumption}

  The above is a nontrivial sparsity assumption that we require for higher-order correctness.
We have the following result:
\bk 

\begin{proposition} 
\label{prop:ew}
Let $G_n$ be the Edgeworth expansion defined in Eq \ref{eq:edgeworth-gn} and let $R$ be a fixed subgraph. Suppose that Assumptions \ref{ass:linear} and \ref{ass:sparse} hold.  Further suppose that $\rho_n=O(1/\log n)$ or Cram\'{e}r's condition holds, i.e. $\limsup_{t\rightarrow \infty} \left|\E\left\{e^{itg_1(X_1)/\tau_n}\right\}\right|<1$ then we have,
\begin{align}
\sup_u \left| F_{n}(u) - G_n(u) \right| = O(\mathcal{M}(n, \rho_n , R))
\end{align}
where $F_{n}$ is the cumulative distribution function of $\frac{\hat{T}_n-\theta_n}{\sigma_n}$ and 
\begin{align}
\label{eq:M-def}
\mathcal{M}(n,\rho_n,R) = \begin{cases} 
      \frac{1}{n\rho_n} & \text{R is acyclic} \\
       \frac{1}{n\rho_n^{r/2}} & \text{R is cyclic} 
  \end{cases}
\end{align}
\end{proposition}

Now, we will state our bootstrap approximation results.  
We will first show that conditioned on the network, and latent variables, the CDF of \mbq matches the asymptotic expansion in Eq \ref{eq:edgeworth-gn}, where the true moments are replaced by their empirical versions. In what follows, let $\hat{E}_n(\cdot)$ denote the expectation operator with respect to the empirical measure of $A$ and $X$. Define
 \begin{equation}\label{eq:gnhat}
 \gnhat(u) = \Phi(u) -\frac{(u^2-1)\phi(u)}{6n^{1/2}\widehat{\tau}_n^3}\left[\widehat{E}_n\left\{g_1(i)^3\right\}+ 3(r-1) \widehat{E}_n\{g_1(i)g_1(j)g_2(i,j)\} \right],
 \end{equation}
where we have:
 \begin{equation}\label{eq:empmoment}
 \begin{split}
  &\widehat{E}_n\left\{g_1(i)^3\right\} = \frac{1}{n}\sum_{i=1}^n \gi^3, \\
    &\widehat{E}_n\{ g_1(i)g_1(j)g_2(i,j) \}=\frac{1}{{n \choose 2}}\sum_{1 \leq i<j\leq n}\tilde{g}_2(i,j)\gi\gj. 
 \end{split}
 \end{equation}
\begin{theorem}\label{thm:mbq-firstorder}
If  Assumptions \ref{ass:linear} and \ref{ass:sparse} are satisfied, 
 the weights $\w_1,\dots,\w_n$ are generated from a non-lattice distribution (see \cite{feller-vol-2} page 539) such that $\E[\xi_1] = 1$, $\E[(\xi_1-1)^2] = 1$,  $\E[(\xi_1-1)^3] = 1$, then 
$$\sup_u\left|P^*\left(\frac{\cfbq-\tnhat}{\shat} \leq u \right)-\gnhat(u)\right|= o_P(n^{-1/2})  + O_P\left(\frac{\log n}{n^{2/3}\rho_n}\right),$$
where $P^*(\cdot)$ denotes the conditional probability of event $(\cdot)$ conditioned on $A$ and $X$.

\end{theorem}
%

\begin{remark}
	\label{remark:standardvsstudent}
	While the above theorem is for standardized bootstraps, our proof may be adapted to yield an analogous statement for bootstraps studentized by a variance estimator inspired by the Delta Method.  In essence, the studentized bootstrap may also be expressed as a weighted U-statistic and a negligible remainder term, allowing the use of similar proof techniques.    
\end{remark}

\begin{remark}
	While \cite{zhang-xia-network-edgeworth} establish higher-order correctness under milder sparsity conditions for subsampling and the empirical graphon, our result here does not require Cram\'{e}r's condition for $g_1(X_i)$, which is an important feature for network applications.  Our simulation study suggests that our rate here can be improved, but we leave this to future work. \bk
\end{remark}

Combining Theorem \ref{thm:mbq-firstorder} with the Hoeffding decomposition in Eq \ref{eq:boot-hoeff}, we obtain the corollary below for the multiplicative bootstrap.  Since the remainder term in the Hoeffding decomposition concentrates slowly for well-connected subgraphs of sparser networks, we impose additional assumptions on the subgraph to maintain the rate from the previous theorem. 
\begin{corollary}\label{thm:mbm-firstorder}
Suppose Assumption \ref{ass:linear} is satisfied and either R is acyclic and $\rho_n = \omega(1/\sqrt{n})$ or R is a simple cycle and $\rho_n = \omega(n^{-1/r})$.  Further suppose that the weights $\w_1,\dots,\w_n$ are generated from a non-lattice distribution with such that $\E(\xi_1) = 1$, $\E\{(\xi_1-1)^2\} = 1$,  $\E\{(\xi_1-1)^3\} = 1$, then, 

$$\sup_u\left|P^*\left(\frac{\cfbm-\tnhat}{\shat} \leq u \right)-\gnhat(u)\right|= o_P(n^{-1/2})  + O_P\left(\frac{\log n}{n^{2/3}\rho_n}\right),$$

where $P^*(\cdot)$ denotes the conditional probability of event $(\cdot)$ conditioned on $A$ and $X$ and $\hat{\sigma}_n=r\hat{\tau}_n/n^{1/2}$.
\end{corollary}



The proof of Theorem \ref{thm:mbq-firstorder} build upon results from \citet{wang-jing-weighted-bootstrap-u-statistics}\bk, which establish higher-order correctness of the weighted bootstrap for order-2 U-statistics.  However, certain terms that appear as constants in their work blow up when perturbed by sparse network noise. To deal with this issue, we control various terms unique to the network setting and use different arguments to control the overall error rate.  

Similar to \cite{zhang-xia-network-edgeworth}, we now consider an empirical Edgeworth expansion in which the expected quantities of interest are estimated.
   \begin{lemma}\label{lem:convEW}
Under the assumptions in Assumption \ref{ass:linear} and \ref{ass:sparse}, we have 
\begin{align*}
    \sup_u|\gnhat(u)-F_n(u)|= O_P(\mzx)
\end{align*}
\end{lemma}
 The lemma above suggests that the empirical Edgeworth expansion achieves a better rate than the bootstrap procedures considered.  In the experimental section, we see that the empirical Edgeworth expansion (\ew) in fact achieves the smallest error when the network is dense enough. 
  However, for smooth functions of counts, it is cumbersome to derive such expansions and the bootstrap emerges as a strong practical alternative that offers improved accuracy over a Normal approximation in certain regimes.       \bk

 
 
\subsection{Smooth functions of count statistics}\label{sec:smf}

In network science, the transitivity coefficient, which may be defined as a smooth function of triangles and two-stars, is commonly used to quantify how much nodes in the network cluster together. Given the importance of such functions in applications, accurate inference for these parameters is naturally of substantial interest. Our results in this section establish the quadratic and multiplicative bootstraps as  accurate and user-friendly methods for smooth functions of counts that sidestep the cumbersome computation of gradients and moments required by empirical Edgeworth expansions.

To the best of our knowledge, below we provide the first result on Edgeworth expansions for smooth functions of U-statistics.  It turns out that arguments to derive Edgeworth expansions for smooth functions of IID means such as those in \citet{hall-bootstrap-edgeworth} depend heavily on the properties of cumulants of independent random variables and require multivariate Edgeworth expansions, complicating extensions even to U-statistics.  In contrast, we adapt flexible Edgeworth expansion results of \citet{jing2010unified} to approximate non-negligible terms arising from a Taylor approximation of the smooth functional.

To state our result, we need to introduce some additional notation.   Let $\bu$ denote a d-dimensional vector of count functionals, let $\bu^*$ be a vector of corresponding bootstrap statistics generated by either the multiplier bootstrap $\hat{T}_{n,M}^*$ or the the quadratic bootstrap $\hat{T}_{n,Q}^*$.  Furthermore, let $f: \mathbb{R}^d \mapsto \mathbb{R}$ denote the function of interest.  
Consider the following smooth function of bootstrapped count frequencies:
\begin{align}
\label{eq-snstar}
S_n^*=n^{1/2}\{f(\bu^*)-f(\bu)\}/\tdsigf
\end{align}
 where  $\tdsigf$ is used to standardize the bootstrap version and will be defined shortly.  The standard Delta Method involves a first-order Taylor expansion; to attain higher-order correctness, we need to consider a second-order expansion.  We use the following notation to denote the derivatives of interest  evaluated at the expectation $E(\bu) = \bmu$: 
 \begin{align}
\label{eq:ai}
  &  a_i=\frac{\partial f(\bxx)}{\partial \bxx^{(i)}}\bigg{|}_{\bxx=\bmu}, \ 
    a_{ij}=\frac{\partial^2 f(\bxx)}{\partial \bxx^{(i)}\partial \bxx^{(j)}}\bigg{|}_{\bxx=\bmu},
\end{align}
Define corresponding gradients for the bootstrap evaluated at the count functional $\bu$:
\begin{align}
\label{eq-ahat}
  &  \hat{a}_i=\frac{\partial f(\bxx)}{\partial \bxx^{(i)}}\bigg{|}_{\bxx=\bu}, \ 
    \hat{a}_{ij}=\frac{\partial^2 f(\bxx)}{\partial \bxx^{(i)}\partial \bxx^{(j)}}\bigg{|}_{\bxx=\bu}.
\end{align}
Define the asymptotic variance of the smooth function as: 
\begin{align}
\label{eq:sigmaf}
\sigf^2&=\sum_{i=1}^d\sum_{j=1}^da_ia_j E \left(\frac{r_{i}\hat{g}_1^{(i)}(l)}{\rho_n^{s_{i}}}\frac{r_{j}\hat{g}_1^{(j)}(l)}{\rho_n^{s_{j}}}\right).
\end{align}
and define the empirical analogue of the asymptotic variance as:
\begin{align}\label{eq:sigmahatfemp}
\tdsigf^{2}&=\sum_{i=1}^d\sum_{j=1}^d\hat{a}_i\hat{a}_j \widehat{E}_n \left(\frac{r_{i}\hat{g}_1^{(i)}(l)}{\rho_n^{s_{i}}} \frac{r_{j}\hat{g}_1^{(j)}(l)}{\rho_n^{s_{j}}}\right).
\end{align}

  We are now ready to state our Edgeworth expansion for the smooth function of the bootstrapped statistics. For simplicity, we state the Edgeworth expansion for $u^*$ resulting from the quadratic bootstrap procedure \mbq. A similar result holds for \mbm, albeit under stronger conditions on the subgraph like those imposed in Corollary~\ref{thm:mbm-firstorder}. 

\begin{theorem}
\label{thm:boot_smooth}
Suppose that $\sigf>0$, the function $f$ has three continuous derivatives in a neighborhood of $\bmu$ and suppose that the weights $\w_1,\dots,\w_n$ are generated from a non-lattice distribution such that $\E(\xi_1) = 1$, $\E\{(\xi_1-1)^2\} = 1$,  $\E\{(\xi_1-1)^3\} = 1$.  Further suppose that Assumptions \ref{ass:linear} and \ref{ass:sparse} are satisfied. Then, we have:
 \begin{align}\label{eq:smf-bt-ew}
    P^*(S_n^* \leq x) = \Phi(x) + n^{-1/2}\phi(x)\{\tilde{A}_1\tdsigf^{-1}+\frac{1}{6}\tilde{A}_2\tdsigf^{-3}(x^2-1)\} + O_P\left(\frac{\log n}{n^{2/3}\rho_n}\right).
\end{align}
where:
 \begin{align*}
\tilde{A}_1=& \ \frac{1}{2}\sum_{i=1}^d\sum_{j=1}^d\hat{a}_{ij} \widehat{E}_n \left(\frac{r_{i}\hat{g}_1^{(i)}(l)}{\rho_n^{s_{i}}} \frac{r_{j}\hat{g}_1^{(j)}(l)}{\rho_n^{s_{j}}}\right),\\
    \tilde{A}_2=& \ \sum_{i=1}^d\sum_{j=1}^d\sum_{k=1}^d \hat{a}_i \hat{a}_j\hat{a}_k \widehat{E}_n \left(\frac{r_{i}\hat{g}_1^{(i)}(l)}{\rho_n^{s_{i}}} \frac{r_{j}\hat{g}_1^{(j)}(l)}{\rho_n^{s_{j}}} \frac{r_{k}\hat{g}_1^{(k)}(l)}{\rho_n^{s_{k}}}\right)
    \\ + & \ 3 \ \sum_{i=1}^d\sum_{j=1}^d\sum_{k=1}^d\sum_{t=1}^{d}\hat{a}_i\hat{a}_j\hat{a}_{kt} \widehat{E}_n \left(\frac{r_{i}\hat{g}_1^{(i)}(l)}{\rho_n^{s_{i}}} \frac{r_{k}\hat{g}_1^{(k)}(l)}{\rho_n^{s_{k}}}\right) \widehat{E}_n \left(\frac{r_{j}\hat{g}_1^{(j)}(l)}{\rho_n^{s_{j}}} \frac{r_{j}\hat{g}_1^{(t)}(l)}{\rho_n^{s_{t}}}\right)
    \\ + & \ 3 \ \sum_{i=1}^d\sum_{j=1}^d\sum_{k=1}^d \hat{a}_i\hat{a}_j\hat{a}_k \widehat{E}_n\left( \frac{r_i \hat{g}_1^{(i)}(l)}{\rho_n^{s_i}} \frac{r_j \hat{g}_1^{(j)}(m)}{\rho_n^{s_j}} \frac{r_k(r_k-1) \tilde{g}_2^{(k)}(l,m)}{\rho_n^{s_k}} \right).
\end{align*}
\end{theorem}
In the Supplementary Material, we derive Edgeworth expansions for smooth functions of U-statistics corresponding to the non-negligible component of the count functional in Proposition~\ref{prop:smooth-pop} and show that our bootstrap version of the Edgeworth expansion is close to this expansion in Proposition~\ref{prop:smf-bt-ew-close}. To derive Edgeworth expansions for the U-statistic, we impose a non-lattice condition; however, it is likely that this assumption can be removed for count functionals if a smoothing argument used in \cite{zhang-xia-network-edgeworth} is adapted.  
\bk
\section{Simulation study}
\label{sec:exp}
Here we provide simulated and real data examples demonstrating the effectiveness of our methods.
\subsection{Simulation study}
\label{sec:sim}
We consider two graphons in our simulation study. The first graphon we consider is a Stochastic Blockmodel (SBM), introduced by \cite{holland-sbm}. The SBM is a popular model for generating networks with community structure. The SBM may be parameterized by a $K \times K$ probability matrix $B$ and a membership probability vector $\pi$ that takes values in the probability simplex in $\mathbb{R}^K$. Let $Y_1, \ldots Y_n \in \{1, \ldots, K\}$ be random variables indicating the community membership of the corresponding node, with probability given by the entries of $\pi$. Under this model, we have that $P(A_{ij}^{(n)} = 1 \ \mid \ Y_i = u, Y_j = v) = \rho_n B_{uv}$.  In our simulations, we consider a two block SBM (\sbm) with $B_{ij}=0.6$ for $i=1,j=1$ and $0.2$ for the rest, with $\pi=(0.65, 0.35)$.
%
We also do experiments on a smooth graphon model(\smg). For space considerations, additional experiments for the \sbm model and experimental results for the \smg model are deferred to Supplement Section~\ref{sec:addlexp}. 

 Define $\text{err}(F,G)$ as the maximum of $|F(x)-G(x)|$ over a grid on $[-3,3]$ with grid size $0.1$; this will serve as an approximation to the Kolmogorov distance between $F$ and $G$.  In order to study this error, we first need an estimate of the true CDF. \bk To this end, we conduct Monte Carlo simulations with $M$ samples generated from each model. Note that, since our goal is to show that the error is better than the Normal approximation, we need $M=\omega((n\rho_n)^2)$, which ensures that the error from the Monte Carlo samples is $o(1/n\rho_n)$.  To ease the computational burden, we perform simulations on small networks with $n=160$ nodes.  We generate $M=10^6$ Monte Carlo simulations, so that the higher-order correctness is not obscured by error from the simulations.   In addition to this, we also compare the coverage for different resampling methods in Figure~\ref{fig:coverage}, where we use $n=500$, since the true CDF does not have to be estimated. In this setting we estimate the true parameter of inferential interest from a 15000 node network. In the next subsection we show how to obtain higher-order correct confidence intervals.\bk

\subsection{Higher-order correct confidence intervals}


In this paper, we have studied the properties of bootstrap methods for standardized count functionals of networks.  While our Edgeworth expansions establish that the standardized bootstrap is higher-order correct in the Kolmogorov norm, in general that does not imply higher-order correct confidence intervals (see Chapter~3 page~93 of \citet{hall-theoretical-comparison-bootstrap-ci}). There are two approaches in the literature to deal with this issue, one is studentization and the other is using the bias-corrected standardized intervals.  While our theoretical results can be extended to studentized count functionals (see Remark \ref{remark:standardvsstudent}), it is well-known that for statistics such as correlation coefficients,  studentization may not be suitable because the variance estimate can be unstable ~\citep{hall-bootstrap-edgeworth}. 
Bias-corrected standardized CIs have been extensively investigated in other settings; see for example, \citet{efron-jackknife-resampling}, \citet{efron-better-ci}, and \citet{hall-theoretical-comparison-bootstrap-ci}. In this section, we show how to correct standardized intervals to attain higher-order accuracy for coverage.  \bk

We define function $\hat{L}(t)=P(\hat{T}_n^* < t \mid  A,X)$ for bootstrap samples $\{\hat{T}^*_n\}$ and for nominal coverage $\alpha$. Let $z_{\alpha}$ denote the standard normal critical point, and let $\hat{y}_\alpha=\hat{L}^{-1}(\alpha)$. 
Our bias-corrected standardized CI is given as follows (see~\cite{hall-bootstrap-edgeworth} for more details):
\begin{align}\label{eq:ci-correction}
    \mathcal{I}'_1 = \left(\hat{y}_{\frac{1-\alpha}{2}}+n^{-1}\hat{\sigma}_n\left\{\hat{p}_1(z_{\frac{1-\alpha}{2}})+\hat{q}_1(z_{\frac{1-\alpha}{2}})\right\},\  \hat{y}_{\frac{1+\alpha}{2}}+n^{-1}\hat{\sigma}_n\left\{\hat{p}_1(z_{\frac{1+\alpha}{2}})+\hat{q}_1(z_{\frac{1+\alpha}{2}})\right\} \right),
\end{align}
where for any $x \in \mathbb{R}$, $\hat{p}_1(x)$ and   $\hat{q}_1(x)$ are estimates for $p_1(x)$ and $q_1(x)$. Recall that $p_1(x)$ and $q_1(x)$ are polynomial coefficients of the second order term in the  Edgeworth expansions of the standardized and studentized statistic. 

When the statistic is a count functional like the triangle or two-star density, then we already know the form of $p_1(x)$ and $q_1(x)$. For smooth functions of count statistics like the transitivity, we derive the standardized and studentized Edgeworth expansion for the smooth function of the corresponding U statistics in the Supplement Sections~\ref{sec:supp:sub-proof-smf-ew} and~\ref{sec:supp:proof-smf-ew-stu-and-ci} respectively.  We also show that the Edgeworth expansion of the bootstrapped smooth function converges to this population version in the Supplement Section~\ref{sec:supp:sub-proof-smf-bt-close-to-ew}. 
\bk
\subsection{Competing methods}
We compare our algorithms with the following network resampling procedures. 

\textit{Empirical Graphon (\eg).} We draw $B$ size $n$ resamples $S_i^{*}$ with replacement from $1,\dots, n$.  We compute the count functional $\hat{T}_{n,i}^*$ on $A^{(n)}(S_i^*,S_i^*)$. We compute $\tnhat$ and $\hat{\sigma}_n^2$ on the whole graph. Now for triangles and two-stars we compute the CDF of $\{(\hat{T}_{n,i}^*-\tnhat)/\hat{\sigma}_n\}_{i=1}^B$. For functions of counts, we compute the function for each resampled graph, center using the function computed on the whole network.

\textit{Subsampling (\sub).} We draw $B$ size $b$ ($b=0.5n$) subsamples $S_i^{*}$ without replacement  from $1,\dots, n$ . We compute the count functional $\hat{T}_{b,i}^*$ on $A^{(n)}(S_i^*,S_i^*)$. We also compute $\tnhat$ and $\hat{\sigma}_n^2$ on the whole graph. We set $\hat{\sigma}_b^2=n/b\hat{\sigma}_n^2$. Now for triangles and two-stars we compute the CDF of $\{(\hat{T}_{b,i}^*-\tnhat)/\hat{\sigma}_b\}_{i=1}^B$. For functions of counts, we compute the function for each subsampled graph, center using the function computed on the whole network.

\newcommand{\bhx}{\hat{\boldsymbol{X}}}
\textit{Latent Space  (\ls).} We first estimate the latent variables $\hat{X}:=\{\hat{X}_1,\dots, \hat{X}_n\}$ from the given network. For \sbm, we use the true number of blocks. 
We compute $g_1(\hat{X}_i)$ for $i=1\dots n$, and then compute $T_n(\hat{X})=T_n(\hat{X}_1,\dots, \hat{X}_n)$. Now we simply use the additive variant of bootstrap $T_n(\hat{X})+\frac{r}{n}\sum_i (g_1(\hat{X}_i)-T_n(\hat{X}))$ (see~\cite{levin-levina-rdpg-bootstrap}). For triangles and two-stars, we normalize by the square root of $r^2/n\sum_i g_1(\hat{X}_i)^2$. For functions of counts, we center using the function computed on $\hat{X}$. 


We compare the performance of the resampling methods for two-stars, triangles and a variant of the transitivity coefficient defined in Example 3 of \citet{Bhattacharyya-subsample-count-features}, which is essentially an appropriately defined ratio between triangles and two-stars.

\subsection{Results}
 In Figure \ref{fig:cdf-error}, we plot the maximum of (absolute) difference of the bootstrap CDF $F_n^*$ over the $[-3,3]$ range  (\KL{$F_n$}{$F_n^*$}) for triangle density from the true CDF $F_n$ for sparsity parameter $\rho_n$ varying from $0.05$ to $1$.   We show the average of the expected difference over 30 independent runs along with error-bars. In this figure we see an interesting phenomenon. For sparse networks with $\rho_n\leq 0.2$, the linear method outperforms the higher-order correct methods. As the networks become denser, the higher-order correct methods start performing better. For $\rho_n\leq 0.2$, we also see that the empirical Edgeworth expansion performs worse than the linear bootstrap method.
 
 We also compare the bootstrap samples centered at the bootstrap mean (first row) with bootstrap samples centered at the subgraph density computed on the whole network. We see that the latest space method (LS) and approximate linear bootstrap method (\mblapx) behave differently under these two centerings, with LS performing much worse. This suggests that while both suffer from bias, LS suffers from it to a higher degree. For all parameter settings, we used $N=50\log n$. It is possible that increasing $N$ for sparser settings may lead to reduced bias of \mblapx.
 
 In Figure \ref{fig:coverage}, we show the coverage of  95\% bootstrap percentile CI with correction for triangles (left) and transitivity coefficient (right) of the \sbm model  in $\rho_n$ from $0.05$ to $1$. We simulate $200$ graphs for each $\rho_n$ from \sbm model, construct CI from bootstrap percentiles and correct the CI using Eq~\ref{eq:ci-correction}  for triangles and transitivity respectively.  For smooth functions, computing the bootstrap distribution is straightforward. One simply computes $u^*$ which is now a vector of bootstrapped triangles and two-star densities. Now a standardized bootstrap replicate is given by $\{f(u^*)-f(u)\}/\tdsigf$, where $u$ is the vector of triangles and two-star densities computed on the whole graph, and $\tdsigf$ is given by Eq~\ref{eq:sigmahatfemp}. For transitivity, $f(x,y)=3x/y$.
\begin{figure}[H]
\vspace{-5mm}
\centering
\begin{tikzpicture}[zoomboxarray,  zoomboxarray inner gap=.5cm, zoomboxarray columns=1, zoomboxarray rows=1]
    \node [image node] { \includegraphics[width=0.35\textwidth]{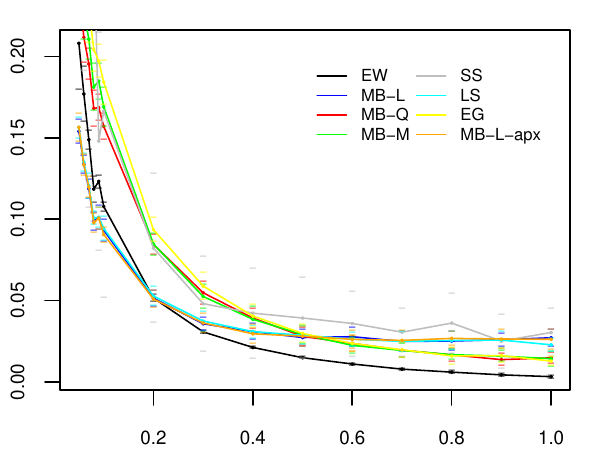}};
    \zoombox[magnification=2]{.15,.7}
\end{tikzpicture}
\begin{tikzpicture}[zoomboxarray,
zoomboxarray inner gap=.5cm, 
zoomboxarray columns=1, 
zoomboxarray rows=1]
    \node [image node] { \includegraphics[width=0.35\textwidth]{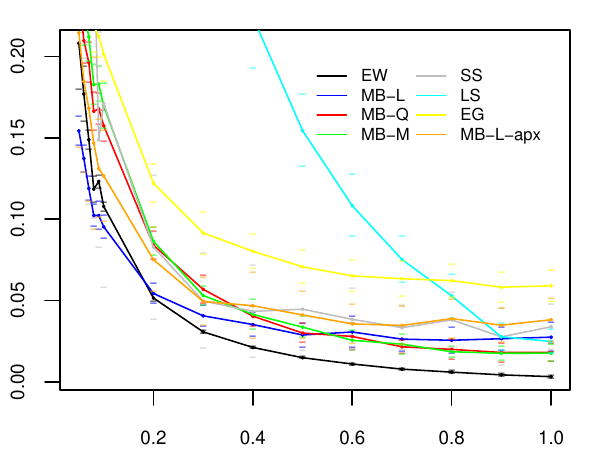} };
    \zoombox[magnification=2]{.15,.7}
\end{tikzpicture}\\
\vspace{-3mm}
\caption{We plot \KL{$F_n$}{$F_n^*$} ($Y$ axis) vs sparsity parameter $\rho_n$ ($X$ axis) for triangle density on networks generated from \sbm, where $F_n^*(t)$ corresponds to the appropriate resampling distribution. The first row is centered at bootstrap mean and normalized by variance estimation from each method $\hat{\sigma}_n$. The second row is centered by triangles density estimated on the whole graph (\mblapx is centered at approximate triangle density estimated from the whole graph) and normalized by $\hat{\sigma}_n$. 
}
\label{fig:cdf-error}
\end{figure}

 \begin{figure}[!htb]
     \includegraphics[trim=2em 2em 0em 0em, clip, width=  .5\textwidth]{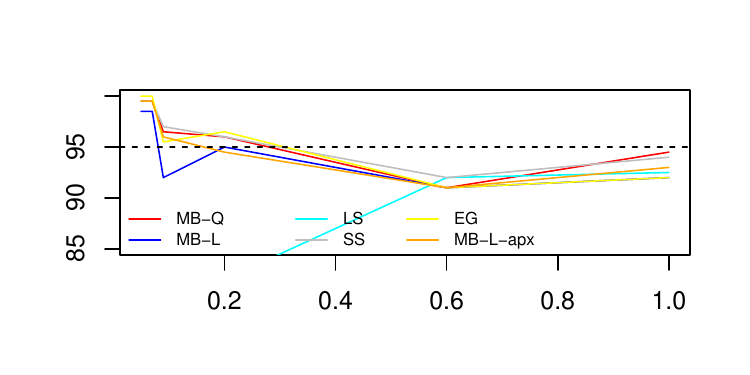}
     \includegraphics[trim=2em 2em 0em 0em, clip,width=  .5\textwidth]{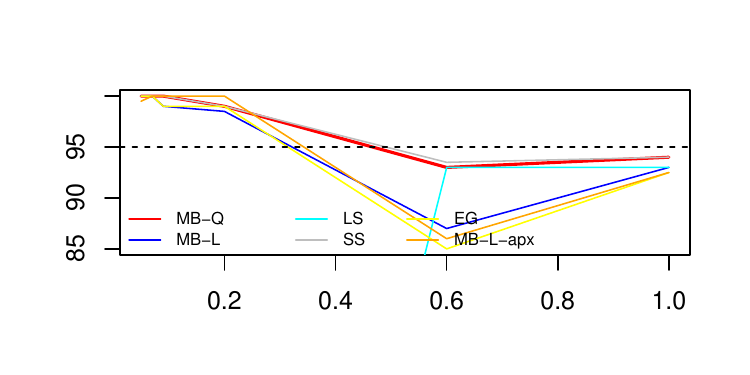}
\caption{We present coverage of  95\% bootstrap percentile CI with correction for triangles (left) and transitivity coefficient (right) of the \sbm model  with $\rho_n$ varying from $0.05$ to $1$. 
 } 
\label{fig:coverage}
\end{figure}

 \bk 

\subsection{Computation time}
In Figure \ref{fig:4cycle-timing}  we show logarithm of running time for 4-cycles count against growing $n$ for 
\sbm model.  We compare our approximate linear method \mblapx with \mblbts, which uses a randomized algorithm for approximating U-statistics proposed in Section~2.2 of~\cite{chen2019randomized}, for precomputation of the local network statistics. 
\begin{figure}[h]
\centering
    \includegraphics[width=  .4\textwidth]{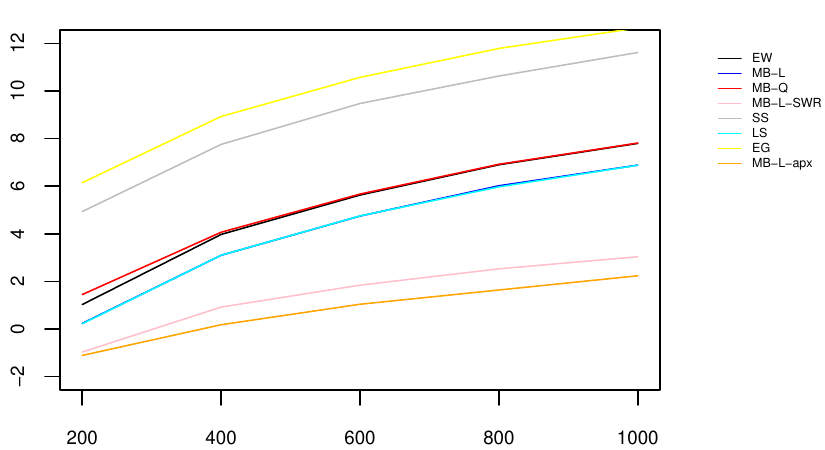} 
\rd \caption{Logarithm of running time for  four-cycles in  \sbm  against sample size $n$.}
\label{fig:4cycle-timing}
\end{figure}
We see that among higher-order correct methods, \mbq offers strong computational performance, outperforming methods such as the \eg and \sub. We see that while \eg has comparable performance to \mbq, it requires recomputation of the count statistic for every bootstrap iteration, making it about $500$ times slower than \mbq for $n=500$ for four-cycle counting. \mbm is the slowest one we do not show here as the weights make symmetric counting shortcuts not as simple to apply. \ew is the fastest among higher-order correct methods, but it cannot be readily adapted for smooth functions of count statistics and is much slower compared to additive methods \ls, \mbl, \mblbts, \mblapx. The four additive methods, i.e. \ls, \mbl, \mblbts, \mblapx are the fastest four of all methods, but they are not higher-order accurate; among them \mblapx is the fastest method in all with our proposed approximate precomputation. The procedure \mblbts draws around $N(n-1)/(r-1)$ size $r-1$ subsets with replacement, whereas we draw $N$ permutations at random and then divide each into consecutive disjoint subsets of size $r-1$.   We observe that \mblapx appears to be faster empirically even if theoretically it has similar computational complexity as \mblbts. 
 The simulation experiments are run on the Lonestar super computer (1252 Cray XC40 compute nodes, each with two 12-core Intel® Xeon® processing cores for a total of 30,048 compute cores) at the Texas Advance Computing Center.

\section{Real data application}
\label{sec:realdata}
\subsection{Real data application}
\label{sec:realdata}
To illustrate our approximated linear bootstrap algorithm on large sparse networks, we apply \mblapx (with the speedup for sparse graphs discussed in Remark~\ref{remark:sparsity}) on the Facebook page network dataset from \cite{rozemberczki2019gemsec}. Each network represents blue verified Facebook pages of different categories. Nodes represent the pages and edges are mutual likes among them. We compare the categories of athletes, politicians and artists also using \mblapx constructed $90\%$ CIs using triangles, two-stars and three-paths. Figure~\ref{fig:fb-pages} shows that both three-paths and two-stars can separate all three pairs of the networks:  athletes vs politicians, athletes vs artists, politicians vs artists, while three-paths provide more clearly separate CIs than two-stars. Three-paths and two-stars are significantly more prevalent in the network of artists compared to athletes and politicians. The sizes of the networks in this example range from around $5000$ to $50000$, whereas the edge densities range from $0.0006$ to $0.0024$. For example, for the artist network ($50515$ nodes, edge density is $0.0006$), our approximate three-path calculation takes around $14$ minutes  while brute force counting does not finish in $24$ hours. See detailed description in Supplement Section~\ref{sec:addlexp}. \bk

\begin{figure}
    \centering
        \scriptsize
    \begin{tabular}{ccc}
    \includegraphics[width=0.3\textwidth]{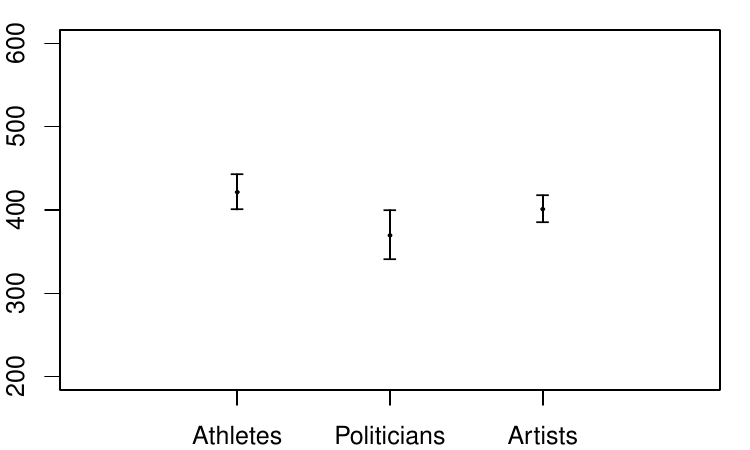}
    &\includegraphics[width=0.3\textwidth]{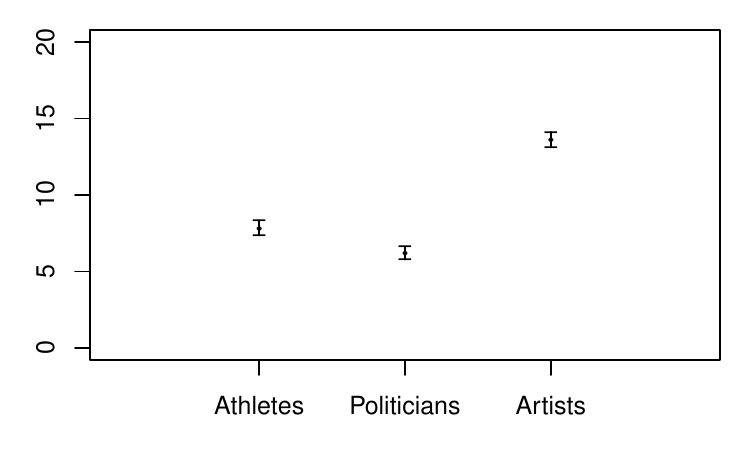}
    &\includegraphics[width=0.3\textwidth]{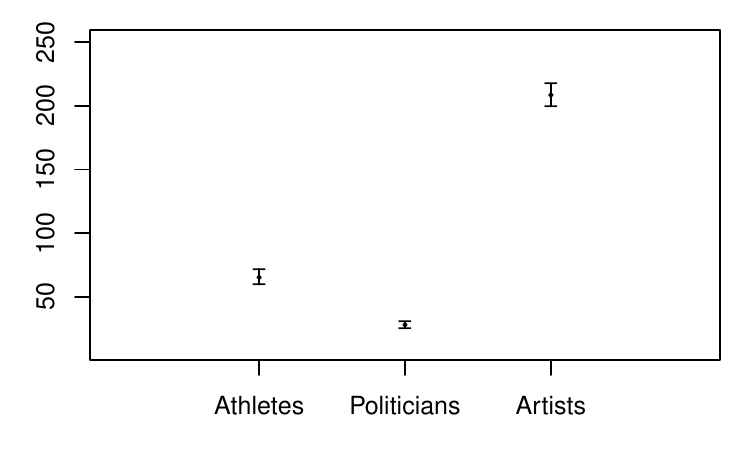}\\
    Triangles & Two-stars & Three-paths \\
    \end{tabular}
    \caption{Bonferroni-adjusted $90\%$ CI for three categories of Facebook pages with triangles, two-stars and three-path using \mblapx algorithm.}
    \label{fig:fb-pages}
\end{figure}

We also compare networks representing the voting similarities of U.S. Congress, which we defer to the Supplement Section~\ref{sec:addlexp} for space constraints.


\bk

\vspace{-5mm}
\section{Conclusion}
In this paper, we propose a family of multiplier bootstraps for network count functionals.  
For large sparse graphs, we  present and analyze an approximate linear bootstrap method that uses randomized sketching algorithms for subgraph counting.  Our theoretical results establish a phase transition phenomena in which different computational regimes lead to limiting behavior. We also present the quadratic bootstrap, which is higher-order correct for moderately sparse graphs; therefore, this method offers a computationally feasible alternative for smaller graphs where accuracy is paramount. For the first time in the literature, we also derive Edgeworth expansions for smooth functions of counts.  
Taken together, we establish the multiplier bootstrap as a user-friendly, automatic procedure that can be tailored to yield higher-order correctness or scalable and consistent inference.  
\bk




\clearpage

\appendix

\setcounter{theorem}{0}
\setcounter{remark}{0}
\setcounter{table}{0}
\setcounter{figure}{0}
\setcounter{equation}{0}
\setcounter{algocf}{0}

\renewcommand{\thetheorem}{S\arabic{theorem}}  
\renewcommand{\thelemma}{S\arabic{lemma}} 
\renewcommand{\theremark}{S\arabic{remark}} 
\renewcommand{\theproposition}{S\arabic{proposition}} 
\renewcommand{\thesection}{S\arabic{section}}   
\renewcommand{\thetable}{S\arabic{table}}   
\renewcommand{\thefigure}{S\arabic{figure}}
\renewcommand{\theequation}{S.\arabic{equation}}
\renewcommand{\thealgocf}{S\arabic{algocf}}

\bigskip
\begin{center}
{\large\bf SUPPLEMENTARY MATERIAL}
\end{center}





\begin{abstract}
We provide detailed proofs of our theoretical results in the supplementary. Section~\ref{sec:supp:lemmaone} presents the Hoeffding decomposition for our multiplier bootstrap count statistic. Section~\ref{sec:sup:proof-be} provides theoretical guarantee for approximate linear bootstrap, showing it is indeed first-order correct under appropriate sparsity conditions.  Section~\ref{sec:suppprop3} establishes an Edgeworth expansions for count statistics normalized by population standard deviation. Section~\ref{sec:suppthm23} proves that our multiplier bootstrap and quadratic bootstrap are higher-order correct under properly generated weights and our sparsity assumption.  Section~\ref{sec:supp:proof-smf-ew-and-smf-boot} provides detailed results for Edgeworth expansion of the corresponding smooth function of the U statistics and show that the Edgeworth expansion of the bootstrapped smooth function converges to its population version. Section~\ref{sec:supp:proof-smf-ew-stu-and-ci} provides examples for confidence interval correction for smooth functions of count statistics such as transitivity. Section~\ref{sec:addlexp} presents additional simulation results and real data results from U.S. House of Representatives roll call vote data. Section~\ref{sec:suppalgapx} presents detailed pseudocode for algorithms.  
\end{abstract}

\newcommand{\bcom}{ \kappa_n }
\newcommand{\lfn}{ l_{5,n} \bk}


%


\section{Proof of Lemma~\ref{lemma:hoeffding-bootstrap}}
\label{sec:supp:lemmaone}
In what follows, we will consider a projection of $T_{n,M}^*$ with respect to the random variables $ \w_1, \ldots \w_n$, conditional on $A$ and $X$. 

Recall that $\xi_i$ follows the Gaussian Product distribution. First, we may express $T_{n,M}^*$ as:
\begin{align*}
T_{n,M}^* = \frac{1}{{n \choose r}} \sum_{1 \leq i_1 < i_2 < \ldots i_r} \left(\xi_{i_1 \cdots i_r} - 1\right)  \cdot \left\{ H(A_{i_1, \ldots, i_r}^{(n)}) - \hat{T}_n \right\}
\end{align*}
where $\xi_{i_1 \cdots i_r}$ denotes the product $\xi_{i_1} \times \cdots \times \xi_{i_r}$. It turns out that applying the Hoeffding decomposition directly to $T_{n,M}^*$ leads to tedious combinatorial calculations; following \citet{Bentkus-edgeworth-symmetric}, let $\Omega_r$ denote an $r$-tuple of $\{1, \ldots, n\}$.  For each summand, we will consider a Hoeffding representation with respect to $\Omega_r$.   Note that using the Hoeffding projection (also see~\cite{Bentkus-edgeworth-symmetric} section 2.8), 
\begin{align*}
    \prod_{1\leq i\leq r}\w_i-1=\sum_{k=1}^r\sum_{1\leq i_1<\dots< i_k\leq r} h_k(\w_{i_1},\dots,\w_{i_k}),
\end{align*}
where for $\Omega_k=\{1,\dots, k\}$,
\begin{align*}
    h_k(\w_{1},\dots,\w_{k})=\sum_{B\in \Omega_k}(-1)^{k-|B|}\E\left\{\prod_{1\leq i\leq r}\w_i-1\mid B\right\}
\end{align*}
Thus the first two terms are given by:
\begin{align*}
    h_1(\w_{1})&:=(\w_1-1)\\ h_2(\w_{1},\w_{2})&:=(\w_1\w_2-1)-(\w_1-1)-(\w_2-1)=(\w_1-1)(\w_2-1)
\end{align*}
In what follows, we will also denote $A^{(n)}_{i_1,\dots,i_r}$ by $A^{(n)}_S$, where $S=\{i_1,\dots, i_r\}$.
Let 
\begin{align}\label{eq:hhij}
    &\hhij= \frac{1}{{n-2 \choose r-2}}\sum_{S \ | \ i,j\in S}H(A_S),\\
    & \hat{H}_u(i_1,\ldots,i_u)=\frac{1}{{n-u \choose r-u}}\sum_{S \ | \ i_1,\ldots,i_u \in S}H(A_S) \notag.
\end{align}
Thus $T_{n,M}^*$ can be written as follows:
\begin{align}\label{eq:hofftnm}
    T_{n,M}^* &= \frac{1}{{n \choose r}} \sum_{1 \leq i_1 < i_2 < \ldots < i_r} \left(\xi_{i_1 \cdots i_r} - 1\right)  \cdot \left\{ H(A_{i_1, \ldots, i_r}^{(n)}) - \hat{T}_n \right\}\notag\\
    &=\frac{1}{{n \choose r}} \sum_{1 \leq i_1 < i_2 < \ldots < i_r}\sum_{k=1}^r\sum_{1\leq i_1<\dots< i_k\leq i_r} h_k(\w_{i_1},\dots,\w_{i_k})\cdot \left\{ H(A_{i_1, \ldots, i_r}^{(n)}) - \hat{T}_n \right\}\notag\\
    &=\frac{1}{{n \choose r}} \sum_{k=1}^r \sum_{1\leq i_1<\dots< i_k\leq n} h_k(\w_{i_1},\dots,\w_{i_k})\cdot\sum_{S}  \left\{ H(A_S^{(n)}) - \hat{T}_n \right\}1 (i_1,\dots, i_k\in S)\notag\\
    &=\frac{1}{{n \choose r}} \sum_{k=1}^r {n-k\choose r-k}\sum_{1\leq i_1<\dots< i_k\leq n}h_k(\w_{i_1},\dots,\w_{i_k}) \frac{\sum_{S}  \left\{ H(A_S^{(n)}) - \hat{T}_n \right\}1 (i_1,\dots, i_k\in S)}{{n-k\choose r-k}}\notag\\
    &=\frac{r}{n}\sum_i (\xi_i-1)\gi+\frac{r(r-1)}{n(n-1)}\sum_{1\leq i<j\leq n} (\w_i-1)(\w_j-1)\underbrace{\{\hat{H}_2(i,j)-\tnhat\}}_{\tilde{g}_2(i,j)}+R_n
\end{align}

 Now, it remains to show that the remainder of $(T_{n,M}^* - \hat{T}_n)/\hat{\sigma}_n$ is $O(\delta(n,\rho_n,R))$, where:
 \begin{align*}
    \delta(n,\rho_n,R) =\begin{cases} 
      \frac{1}{n\rho_n} & \text{R is acyclic} \\
       \frac{1}{n\rho_n^{3/2}}& \text{R is a simple cycle} .
  \end{cases}
\end{align*}

 The residual $R_n$ is a sum of higher order Hoeffding projections, which are all uncorrelated. %
 Therefore, we see that the variance of the $u^{th}$ order term is $\dfrac{\sum_{1\leq i_1<i_2\dots<i_u}  \tilde{g}_u(i_1,\dots,i_u)^2}{\hat{\sigma}_n^2 {n\choose u}^2}$. We will now obtain expressions for $3\leq u\leq r$.

Consider any term $\tilde{g}_u(1,\dots,u)$. We will now bound $\E\{\tilde{g}_u(1,\dots,u)^2\}$.
\begin{align*}
    \E\{\tilde{g}_u(1,\dots,u)^2\}\leq 2[\var\{\hat{H}_u(1,\dots,u)\}+\underbrace{\var(\tnhat)}_{O(\rho_n^{2s}/n)}]
\end{align*}
The bound on the second term follows from~\cite{Bickel-Chen-Levina-method-of-moments} and will be smaller than that of the first term.
Let $\mathcal{S}_{r,u}$ denote all subsets of size $r-u$, not containing $1,\dots u$. For any subset $S\in \mathcal{S}_{r,u}$, also define, $S_u=S\cup \{1,\dots,u\}$.
For the first part, we have:
\begin{align*}
    \var\{\hat{H}_u(1,\dots,u)\}=\frac{\sum_{S,T\in \mathcal{S}_{r,u}}\cov\{H(A_{S_u}),H(A_{T_u})\}}{{n-u\choose r-u}^2}
\end{align*}
Note that the dominating term here will indeed be the one where $|S\cap T| =0$. The number of such terms is ${n\choose 2r-u}$. Also the covariance of those terms will be $\rho_n^{2s-E(A_{1,\dots,u})}$, where $E(A_{1,\dots,u})$ denotes the intersection of the edgeset of $A_{1,\dots,u}$ and the subgraph we are counting. This number can be at most $u-1$ for acyclic $R$ and $u$ for a simple cycle $R$. For $|S\cap T|=k$, the number of terms is ${n\choose 2r-2u-k}$ and the exponent on $\rho_n$ is at most $2s-(u+k-1)$. Thus, for an acyclic subgraph, we have,
\begin{align*}
    \var\{\hat{H}_u(1,\dots,u)\}&\leq \frac{\sum_{k=0}^r {n\choose 2r-2u-k}\rho_n^{2s-(u+k-1)}}{{n-u\choose r-u}^2}\\
    &\leq \sum_{k=0}^r n^{-k}\rho_n^{2s-(u+k-1)}=\rho_n^{2s-(u-1)}\left(1+\sum_{k>0}\frac{1}{(n\rho_n)^k}\right)
\end{align*}
The cyclic one is worse by a factor of $\rho_n$.
Thus the contribution of the $u^{th}$ element of the Hoeffding decomposition is 
\begin{align*}
    n\cdot \frac{\sum_{i_1,\dots, i_u} \tilde{g}_u(i_1,\dots,i_u)^2}{{n\choose u}^2 \hat{\tau}_n^2}&=\begin{cases}O_P\left(\frac{1}{(n\rho_n)^{u-1}}\right)&\mbox{$R$ acyclic}\\
    O_P\left(\frac{\rho_n^{-1}}{(n\rho_n)^{u-1}}\right)&\mbox{$R$ a simple cycle}\end{cases}
\end{align*}

This shows that the third term contributes the most to $R_n$ in Eq~\ref{eq:hofftnm}.  By Markov's inequality, and the definition of $O_P(.)$ notation, it is easy to see that $R_n=O_P(\delta(n,\rho_n,R))$. 

\section{Proof of Theorems~\ref{thm:clt-approx} and~\ref{thm:boot-linear-apx}}\label{sec:sup:proof-be}
To establish the central limit theorem, we will need the following intermediary result, which establishes  that the conditional variance concentrates around its expectation.  In what follows, let:
\begin{align}\label{eq:def-sigmatildeAX}
\tilde{\sigma}^2_{|A,X}:= \mathrm{Var}(\tilde{T}_n - \hat{T}_n \ | \ A, X), \ \ \ \tilde{\sigma}_n^2 := E(\tilde{\sigma}^2_{|A,X})
\end{align}

\begin{theorem}\label{thm:clt-var-conc}
Suppose that $\rho_n \rightarrow 0$, $n \rho_n \rightarrow \infty$, $R$ is acyclic or a simple cycle, and $r \geq 2$.  Then the following holds:
\begin{align*}
\frac{\tilde{\sigma}_{|A,X}^2}{\tilde{\sigma}_n^2} \xrightarrow{P} 1    
\end{align*}
\end{theorem}
\begin{proof}
By Chebychev's inequality, it suffices to show that $\var(\tilde{\sigma}_{|A,X}^2)/\tilde{\sigma}_n^4 \rightarrow 0$.  We start by lower bounding the expectation $\tilde{\sigma}_n^2$.  Observe that:  
\begin{align}
\begin{split}
E(\var( \tilde{H}_1(i) - H_1(i) \ | \ A,X)) &=  \frac{1}{N} \left\{ \frac{1}{(n-1)!}\sum_{\pi} E(H_\pi - \theta_i)^2 - E(H_1(i) - \theta_i)^2 \right\} 
\\ &= \frac{1}{N}\{ E( \var( H_\pi) \ | \ X_i) - E( \var( H_1(i) \ | \ X_i))\}
\\ &= \Theta\left( \frac{\rho_n^s}{Nn} \right) - O\left(\frac{\rho_n^{2s}}{Nn}\right)
\end{split}
\end{align}

Therefore,
\begin{align*}
E\{\var(\tilde{T}_n - \hat{T}_n \ | \ A,X)\} = \frac{1}{n^2} \sum_{i=1}^n E\{\var(\tilde{H}_1(i) - H_1(i) \ | \ A,X )\} = \Theta\left( \frac{\rho_n^s}{Nn^2} \right)
\end{align*}

Now, we analyze $\var(\tilde{\sigma}_{|A,X}^2)$. Let $\theta_n^{(i)} = \mathbb{E}[H_\pi(i) \ | \ X_i]$. We have that:
\begin{align*}
\tilde{\sigma}^2_{|A,X} &= \frac{1}{n^2 N} \sum_{i=1}^n \frac{1}{(n-1)!} \sum_{\pi} \{ H_\pi(i) - H_1(i)\}^2   
\\ &= \underbrace{\frac{1}{n^2 N} \sum_{i=1}^n \frac{1}{(n-1)!} \sum_{\pi} \{ H_\pi(i) - \cei \}^2}_{P_1} - \underbrace{\frac{1}{n^2 N} \sum_{i=1}^n \{ H_1(i)- \cei) \}^2}_{P_2} %
\end{align*}
%
%
For $P_1$, we further expand the sum into two parts, consisting of squared terms and cross terms:
\begin{align*}
P_1 =& \underbrace{\frac{(r-1)^2}{n^2(n-1)^2N} \sum_{i=1}^n \frac{1}{(n-1)!}\sum_\pi \sum_{S \in \mathbb{S}_\pi} \{H(A_{S \cup i}) - \theta_n^{(i)}\}^2}_{P_{1,1}} \\ + \ & \ \underbrace{\frac{(r-1)^2}{n^2(n-1)^2N} \sum_{i=1}^n \frac{1}{(n-1)!} \sum_\pi \sum_{S \not= T, S,T \in \mathbb{S}_\pi} \{ H(A_{S \cup i}) - \cei\} \cdot \{ H(A_{T \cup i}) - \cei\}}_{P_{1,2}}
\end{align*}


For $P_{1,1}$, observe that $H(A_{S \cup i})^2 = H(A_{S \cup i})$. Distributing the square and re-writing each term as (weighted) U-statistics, we have that, for some constants $c_1,c_2,c_3$ depending only on $r$, 
\begin{align*}
P_{1,1} = \frac{1}{n^3 N} \left[ \frac{c_1}{{n \choose r}} \sum_{|S| = r} H(A_S) -  \frac{2c_2}{{n \choose r}} \sum_{|S|=r} \left(\frac{1}{r}\sum_{i \in S} \theta_n^{(i)}\right) H(A_S) + \frac{c_3}{n} \sum_{i=1}^n (\theta_n^{(i)})^2  \right]
\end{align*}

It can readily be seen that first summand dominates; therefore, $\var(P_{1,1}) = O( \rho_n^{2s}/n^7N^2  )$.

We proceed to bounding $\var(P_{1,2})$.  For a given node i, observe that each subset $S \cup T$ appears in $O(n^2 (n-1-2(r-1))!=O(n^2 (n-2r+1)!$ 
permutations.  We may now define a new kernel corresponding with $2r-1$ nodes on the entire graph; thus, for some $c$ depending only on $r$,
\begin{align*}
P_{1,2} =   \frac{c}{Nn{n \choose 2r-1}}\sum_{|K| = 2r-1  } \frac{1}{2r-1}\sum_{S \cup T = K, i \in S \cap T }[H(A_S) - \theta_n^{(i)}] \cdot [H(A_T) -\theta_n^{(i)}]    
\end{align*}
Using Lemma~\ref{lem:p12} we see that
\begin{align*}
\var(P_{1,2})=O\left( \frac{\rho_n^{4s}}{N^2n^4} \right)+O\left( \frac{\rho_n^{4s-2}}{N^2n^5} \right) + O\left( \frac{\rho_n^{2s}}{N^2n^{2r+1}} \right)
\end{align*}

For $P_2$, we again expand the square and consider squared and cross terms:
\begin{align*}
P_2 =& \underbrace{\frac{1}{n^2N}\sum_{i=1}^n \frac{1}{{n-1 \choose r-1}^2}\sum_{S \in \mathbb{S}\{-i\}} \{H(A_{S \cup i}) -\cei\}^2}_{P_{2,1}}   \\ + & \underbrace{\frac{1}{n^2N} \sum_{i=1}^n \frac{1}{{n-1 \choose r-1}^2}\sum_{S \neq T, S,T \in \mathbb{S}\{-i\}} \{H(A_{S \cup i}) -\cei\} \cdot \{ H(A_{T \cup i}) - \cei \} }_{P_{2,2}}
\end{align*}
For $P_{2,1}$, we immediately have that, for $c_1,c_2,c_3$ depending only on $r$,
\begin{align*}
P_{2,1} = \frac{1}{ N n^r}  \left[\frac{c_1}{{n \choose r}} \sum_{|S| = r} H(A_S) - \frac{2c_2}{{n \choose r}} \sum_{|S|=r} \left(\frac{1}{r}\sum_{i \in S} \theta_n^{(i)}\right) H(A_S) + \frac{c_3}{n} \sum_{i=1}^n (\theta_n^{(i)})^2\right]   
\end{align*}
    Therefore, $\var(P_{2,1}) = O(\rho_n^{2s}/n^{2r+1}N^2)$.  

Finally using Lemma~\ref{lem:p12},
\begin{align*}
    \var(P_{2,2})=O\bb{\frac{1}{N^2}\bb{\frac{\rho_n^{2s}}{n^7}+\frac{\rho_n^{2s}}{n^{2r+1}}+\frac{\rho_n^{4s}}{n^4}+\frac{\rho_n^{4s-2}}{n^5}}}
\end{align*}
Now, putting everything together, we see that $\var( \tilde{\sigma}_{|A,X}^2) = o\left( \frac{\rho_n^{2s}}{N^2n^4} \right)$ and thus the result follows.
\end{proof}

Below, we state and prove lemmas that establish the order of the cross-terms that may be written as noisy U-statistics. Let:
\begin{align*}
\hat{U}_n^{(k)} = \frac{1}{n^k} \frac{1}{{n \choose 2r-k} }\sum_{K :|K| = 2r-k } \ \sum_{ \substack{S,T : \ S \cup T = K, \\|S| = r,\ |T| = r }} \ \sum_{i \in S \cap T} \{H(A_S) - \cei\} \cdot \{(H(A_T) - \cei \}  
\end{align*}
We have the following lemma:
\begin{lemma}
\label{lem:p12}
Suppose that $n \rho_n \rightarrow \infty$ and $S$ is acyclic or a simple cycle with $r \geq 2$. Then, for any $k \in \{1, \ldots, r-1 \}$,
\begin{align*}
\var(U_n^{(k)}) = O\left( \frac{\rho_n^{4s}}{n^4} \right)+O\left( \frac{\rho_n^{4s-2}}{n^5} \right) +  O\left( \frac{\rho_n^{2s}}{n^{2r+1}} \right)
\end{align*}

\end{lemma}

\begin{proof}
We again argue by law of iterated variance.  Observe that the conditional expectation given $X_1, \ldots, X_n$ is a (noiseless) U-statistic, with kernel of the order $\rho_n^{4s-2k+2}$.  Therefore, for $k>1$,
\begin{align*}
\var\{E(U_{n}^{(k)} | X_1, \ldots X_n) \} &= O\left(\frac{\rho_n^{4s-2k+2}}{n^{2k+1}} \right)
\\ &= O\left(\frac{\rho_n^{4s-2} }{n^5} \right)
\end{align*}
For $k=1$, we use the fact that the corresponding kernel is degenerate of order 1.  To see this, first note that the expectation is 0 since:
\begin{align*}
E[ E[\{H(A_S) - \cei \} \ | \ X_i] \cdot E[\{H(A_T) - \cei \} \ | \ X_i] = 0   
\end{align*}
 Moreover, for any $j \in K$, the conditional expectation is 0. Without loss of generality, suppose $j \in S$ and $j \neq i$; the $i=j$ calculation is analogous to the expectation one above.  Then,
 \begin{align*}
  & E[\{H(A_S) - \cei\} \cdot \{(H(A_T) - \cei \} \ | \ X_j] 
 \\ = & \  E[E[\{H(A_S) - \cei\} \cdot \{H(A_T) - \cei \} \ | \ X_i,X_j] \ | \ X_i] 
 \\ = & \  E[E[\{H(A_S) - \cei\} \ | \ X_i,X_j] \cdot E[\{H(A_T) - \cei \} \ | \ X_i] \ | \ X_i] = 0   
 \end{align*}
 Therefore, the Hayek projection is 0, and the variance satisfies:
 \begin{align*}
\var\{E(U_{n}^{(1)} | X_1, \ldots X_n) \} = O\left(\frac{\rho_n^{4s}}{n^4} \right)
 \end{align*}
 Now for the conditional expectation of the variance, we further split  the variance of a sum into variance and covariance components.  Observe that that the variance of the product $\{H(A_S) - \cei\} \cdot \{(H(A_T) - \cei \}$ conditional on $X_1, \ldots X_n$ is dominated by the variance of the term $H(A_S)H(A_T)$, which is Bernoulli(p) with $p = O( \rho_n^{2s-k+1})$.  Thus, for $k \geq 1$, the contribution of these terms is $O( \rho_n^{2s}/n^{r+1})$.  Now for the covariance, for all $k \geq 1$, by Lemma  \ref{lemma-cov-order}, the contribution of these terms is upper bounded by $O\left(\frac{\rho_n^{4s}}{n^4} \right) + O\left(\frac{\rho_n^{4s-2}}{n^5} \right)$.  The result follows.

 \end{proof}
 
 Now we state and prove a result on the order of the conditional covariance terms, which was used in the previous lemma.   
 \begin{lemma}
 \label{lemma-cov-order}
Suppose that $n \rho_n \rightarrow \infty$ and $S$ is acyclic or a simple cycle with $r \geq 2$. Then, for any $k \in \{1, \ldots, r-1 \}$, \bk
 \begin{align*}
   & \frac{1}{n^{4r}}\sum_{i,j}\sum_{\substack{S\neq T, S,T\in \mathbb{S}\{-i\}\\S'\neq T', S',T'\in \mathbb{S}\{-j\}}} E\bb{(H(A_{S\cup i})-\cei)(H(A_{T\cup i})-\cei)(H(A_{S'\cup j})-\theta^{(j)})(H(A_{T'\cup j})-\theta^{(j)})}
 \\ & = O\left(\frac{\rho_n^{4s}}{n^4} \right) + O\left(\frac{\rho_n^{4s-2}}{n^5} \right)  
\end{align*}
 \end{lemma}
 \begin{proof}

 Consider four subsets $S,T,S',T'$. There are 3 types of intersections. The first type are nodes \textit{only} in the intersection of a pair of subsets, and between six pairs. We will denote the size of  these intersections by $k_i,i\in[6]$. We will denote by $K$ the sum of $k_i's$. The second are nodes only in triple wise intersections with sizes $j_i,i\in[4]$. The total of these numbers is $J$. Finally there are nodes in the intersection between all four, whose size will be denoted by $J$. Note that for any given configuration of the sets (translating to a configuration of these numbers), there are a total of $n^{4r-(K+2L+3J)}$ nodes. This is because every node in the pairwise intersection gets counted twice. Those in three-way intersections get counted thrice and those in the four-way intersection get counted thrice. Finally note that there are a total of $K+L+J$ nodes such that each node belongs to some intersection (call this $I$). In order to find the number of edges, note that this intersection can be represented as four sets, which are $I\cap S,I\cap T$ etc. For each of these intersections, since we are allowing simple cycles or acyclic graphs,  note that we can have $|I\cup S|-1(|I\cup S|>0)$ edges. Thus the the total number of edges can be at most $|I|-1$ edges in the intersection. So the contribution of these numbers is
 $$\frac{n^{4r-(K+2L+3J)}\rho_n^{4s-(K+L+J-1)}}{n^{4r}}=\frac{\rho_n^{4s+1}}{(n\rho_n)^{K+L+J}n^{L+2J}}$$
 Now we consider some cases of configuration before summing over all $K,L,J$. We first note that all cases with $K\leq 3,L=0,J=0$ are zero. 
 For $K=4,L=0,J=0$, the only possible case that contributes to the covariance is the one with $k_i=1,i\leq 4$. This contributes $\rho_n^{4s}/n^{4}$. Similarly note that if we have $K=5$, the contribution comes from $k_1=1,k_2=1,k_3=1,k_4=2$ which has contribution $\rho_n^{4s-1}/n^5$ and so on. Thus the terms contribute $\sum_{i\geq 0}\rho_n^{4s-i}/n^{4+i}=O\bb{\rho_n^{4s}/n^4}$. 
  We also consider the case where $K= 1,L= 2,J= 1$, this will contribute $\frac{\rho_n^{4s-1}}{n^5}$.
 Now we consider the case with $K= 2,L= 0,J= 1$. These terms will contribute $ \frac{\rho_n^{4s-2}}{n^5}$. In fact for $K=2,L\geq 1,J\geq 1$, we will have a contribution of
 \begin{align*}
     \frac{\rho_n^{4s+1}}{(n\rho_n)^4n^3}=\frac{\rho_n^{4s-2}}{(n\rho_n)n^6}
 \end{align*}
 
 Now we consider the terms of the form $K\geq 3,L\geq 0, J\geq 1$. These terms contribute:
 \begin{align*}
     \frac{\rho_n^{4s+1}}{(n\rho_n)^{4}n^{2}}=\frac{\rho_n^{4s-2}}{(n\rho_n)n^{5}}
 \end{align*}
 
Consider $C_{k,\ell,j}$ the covariance contribution for sets with $K=k,L=\ell,J=j$.  Putting everything together, we have:
 \begin{align*}
     \sum_{k,\ell,j\geq 0}C_{k,\ell,j}&=
     \sum_{k< 4, \ell=j=0}C_{k,\ell,j}+\sum_{k\geq 4, \ell=j=0}C_{k,\ell,j}+\sum_{k=1,\ell=2,j=0}C_{k,\ell,j}+\sum_{k=2, \ell=0,j= 1}C_{k,\ell,j}\\
     &\qquad +\sum_{k=2, \ell\geq 1,j\geq 1}C_{k,\ell,j}+\sum_{k\geq 3, \ell\geq 0,j\geq 1}C_{k,\ell,j}\\
     &=0+O\bb{\frac{\rho_n^{4s}}{n^4}}+ O\bb{\frac{\rho_n^{4s-1}}{n^5}}+ O\bb{\frac{\rho_n^{4s-2}}{n^5}}.
 \end{align*}
 \end{proof}
 Now, we establish a CLT for the approximate count functional.  Our overall strategy is follows \citet{chen2019randomized}.  Recall from Equation~\ref{eq:def-sigmatildeAX} that $\noisesdax^2 = \var(\tilde{T}_n - \hat{T}_n \ | \ A,X)$ and $\noisesd^2 = E(\noisesdax^2)$.
 \begin{theorem}
 Suppose that $\rho_n \rightarrow 0$ and either $R$ is acyclic and $n \rho_n \rightarrow \infty$ or $R$ is a simple cycle and $n^{r-1} \rho_n^r \rightarrow \infty$.  Moreover, suppose that $N \gg 1/n^2\rho_n^s$.  Then,
 \begin{align*}
 \sup_{t \in \mathcal{R}} \left| P\left( \frac{\tilde{T}_n - \theta_n}{ \sqrt{\noisesd^2 + \sigma_n^2}} \leq t  \right) - \Phi(t)    \right| \rightarrow 0
 \end{align*}
 
 \end{theorem}
 
\begin{proof}
 Let $Z \sim N(0,1)$ be independent of $A,X$ and $t' = \sqrt{\noisesd^2 + \sigma_n^2} \cdot t$.
 We may further reduce the problem to establishing central limit theorems for the randomization and signal components separately as follows:

\begin{align*}
& \ \sup_{t \in \mathcal{R}} \left| P\left( \frac{\tilde{T}_n - \theta_n}{ \sqrt{\noisesd^2 + \sigma_n^2} } \leq t  \right) - \Phi(t)    \right|
\\ \leq & \  \sup_{t' \in \mathcal{R}} \left|P\left( \frac{\tilde{T}_n - \hat{T}_n}{\noisesdax } \leq \frac{t'-(\hat{T}_n - \theta_n)}{\noisesdax}  \right) -\Phi\left(\frac{t'-(\hat{T}_n - \theta_n)}{\noisesd} \right) \right| 
\\ + & \ \sup_{t' \in \mathcal{R}} \left| P\left(\frac{(\hat{T}_n - \theta_n) + \noisesd Z}{\sigma_n} \leq \frac{t^\prime}{\sigma_n} \right) - \Phi\left(\frac{t'}{\sqrt{\noisesd^2 + \sigma_n^2}} \right)   \right| = I +II 
\end{align*}
To show $I \rightarrow 0$, observe that:
\begin{align*}
I  \leq & \   E \left\{ \ \sup_{t' \in \mathcal{R}} \left|P\left( \frac{\tilde{T}_n - \hat{T}_n}{\noisesdax } \leq \frac{t'-(\hat{T}_n - \theta_n)}{\noisesdax} \ \biggr\rvert \  A,X  \right) -\Phi\left(\frac{t'-(\hat{T}_n - \theta_n)}{\noisesdax} \ \biggr\rvert \  A,X \right) \right| \ \right\}
 \\ + & E \left\{ \ \sup_{t' \in \mathcal{R}} \left| \Phi\left(\frac{t'-(\hat{T}_n - \theta_n)}{\noisesdax} \ \biggr\rvert \  A,X \right) - \Phi\left(\frac{t'-(\hat{T}_n - \theta_n)}{\noisesd} \ \biggr\rvert \  A,X \right)  \right| \ \right\} = I_a + I_b
\end{align*}

For $I_a$, the Berry-Esseen Theorem implies that if:
\begin{align}
\label{eq:conditional-berry-esseen-condition}
\frac{\frac{1}{n^3}\sum_{i=1}^n\E\left(|\tilde{H}_1(i) - H_1(i) |^3\ | \ A,X \right)}{\noisesdax^{3/2}} \xrightarrow{P} 0,
\end{align}
then:
\begin{align*}
\sup_{t' \in \mathcal{R}} \left|P\left( \frac{\tilde{T}_n - \hat{T}_n}{\noisesdax } \leq \frac{t'-(\hat{T}_n - \theta_n)}{\noisesdax} \ \biggr\rvert \ A, X  \right) -\Phi\left(\frac{t'-(\hat{T}_n - \theta_n)}{\noisesdax} \ \biggr\rvert \ A,X \right) \right| \xrightarrow{P} 0     
\end{align*}
By Theorem \ref{thm:clt-var-conc}, $\frac{\noisesdax^2}{\noisesd^2} \xrightarrow{P} 1$,
where $\noisesd^2 = \Theta\left( \frac{\rho_n^s}{Nn^2} \right)$.
Moreover, by Lemma \ref{lemma:third-moment-randomized}, we have that:
\begin{align*}
 \frac{1}{n^3}\sum_{i=1}^n\E\left(|\tilde{H}_1(i) - H_1(i) |^3\ | \ A,X \right)= O_P\left( \frac{\rho_n^s}{n^4N^2} \right) + O_P\left(\frac{\rho_n^{3s/2}}{N^{3/2}n^{7/2}} \right)   
\end{align*}
Therefore, under the condition $N \gg 1/n^2 \rho_n^s $, Eq \ref{eq:conditional-berry-esseen-condition} holds and thus the expectation in converges to zero since the Kolmogorov distance is bounded.  Furthermore, we have that:
\begin{align*}
\sup_{t' \in \mathcal{R}} \left|\Phi\left(\frac{t'-(\hat{T}_n - \theta_n)}{\noisesdax} \ \biggr\rvert \ A,X \right) - \Phi\left(\frac{t'-(\hat{T}_n - \theta_n)}{\noisesd} \ \biggr\rvert \ A,X \right)  \right| \xrightarrow{P} 0   
\end{align*}
Therefore by similar reasoning $I_b \rightarrow 0$.
Now for $II$, for some $Z'$ independent of $Z$, by Theorem 1 of \citet{Bickel-Chen-Levina-method-of-moments} and P{\'o}lya's theorem:
\begin{align*}
II & \leq E \left\{\sup_{t' \in \mathcal{R}} \left| P\left(\frac{(\hat{T}_n - \theta_n)}{\sigma_n} \leq \frac{t^\prime - \noisesd Z}{\sigma_n} \ \biggr\rvert \ Z=z \right) - P\left(Z^\prime \leq \frac{t^\prime - \noisesd Z}{\sigma_n} \ \biggr\rvert \ Z=z \right) \right| \right\} \rightarrow 0
\end{align*}
To conclude the proof, observe that:
\begin{align*}
P\left(Z^\prime \leq \frac{t^\prime - \noisesd Z}{\sigma_n} \right)
= P(\sigma_n Z^\prime + \noisesd Z \leq t') = \Phi\left( \frac{t'}{\sqrt{\noisesd^2 +\sigma_n^2}}  \right)
\end{align*}
\end{proof}

\bk

\bk



\begin{lemma}
\label{lemma:third-moment-randomized}
\begin{align*}
E|\tilde{H}_1(i) - H_1(i)|^3  = O\left( \frac{\rho_n^s}{n^2 N^2} \right) + O\left(\frac{\rho_n^{3s/2}}{N^{3/2}n^{3/2}} \right) 
\end{align*}
\end{lemma}

\begin{proof}
By Rosenthal's inequality, for some constant $C_1$ and $C_2$, we have 
\begin{align*}
    E\left(|\tilde{H}_1(i)-H_1(i)|^{3} \ \bigr \rvert \ A,X\right) &\leq \underbrace{\frac{C_1}{N^2} E[|H_\pi(i)-H_1(i)|^3|A,X]}_{I}\\ &+ \underbrace{\frac{C_2}{N^{3/2}}\left( E[(H_\pi(i)-H_1(i))^2|A,X]\right)^{3/2}}_{II}.
\end{align*}
Now, let $\cei = E[H(A_{S \cup i}) \ | \ X_i]$ and observe that $E|H_1(i) - \cei|^c \leq E|H(A_{S \cup i}) - \cei|^c$ for $c\geq 1$ via Jensen's inequality.  Moreover, conditionally on $X_i$, note that for $S \in  \mathbb{S}_\pi$,  $H(A_{S \cup i})$ are mutually independent since their node sets are disjoint.  Now for $I$, we use Rosenthal again to bound the third absolute moment of the Bernoulli sum:
%
\begin{align*}
    E(|H_\pi(i)-H_1(i)|^3|) &\leq \frac{Cr^3}{n^3}\frac{1}{(n-1)!} \sum_{\pi}E\left[E\left(\bigr|\sum_{S\in \mathbb{S}_\pi} H(A_{S\cup i})-\cei\bigr|^3 \biggr \rvert \ X_i\right)\right]
    \\ & \leq  \frac{C_1}{n^2} E|H(A_{S \cup i}) -\cei|^3 +  \frac{C_2}{n^{3/2}} \left\{E(H(A_{S \cup i}) -\cei)^2\right\}^{3/2}
    \\ &= O\left( \frac{\rho_n^{s}}{n^2} \right) + O\left( \frac{\rho_n^{3s/2}}{n^{3/2}} \right)
\end{align*}
Now for $II$, 
using the bound in Lemma~\ref{lem:rosenthal:V11} Eq. \ref{eq:rosenthal-var-term}, along with the fact that $\|X\|_{3/2} \leq \|X\|_{2}$, we have for some constant $C$,
\begin{align*}
E(II) \leq \frac{C \rho_n^{3/2}}{n^{3/2}N^{3/2}}.    
\end{align*}

The result follows.
\end{proof}

\begin{theorem}\label{lem:bootvarconc}
Suppose that $\rho_n \rightarrow 0$, $r \geq 2$ and either $R$ is acyclic and $n \rho_n \rightarrow \infty$ or $R$ is a simple cycle and $n^{r-1} \rho_n^{r} \rightarrow \infty$. 
\bk
Let $\kappa_n^2=\frac{1}{n^2}\sum_{i=1}^n (\tilde{H}_1(i) - \tilde{T}_n)^2$. When $N \gg 1/n^2 \rho_n^s$, we have that:
\begin{align*}
  \frac{\kappa_n^2 }{\tilde{\sigma}_n^2+\sigma_n^2/r^2} \xrightarrow{P} 1. 
\end{align*}
Furthermore, when $Nn\rho_n^s\gg 1$, $r^2(\tilde{\sigma}_n^2+\sigma_n^2/r^2)/\sigma_n^2\rightarrow 1$, whereas, when $Nn^2\rho_n^s\gg 1\gg Nn\rho_n^s$, $(\tilde{\sigma}_n^2+\sigma_n^2/r^2)/\tilde{\sigma}_n^2\rightarrow 1$. 
\end{theorem}

\begin{proof}
First, observe that the variance of the bootstrap is given by:
\begin{align}\label{eq:def-vnhat}
\bootsd^2 = \frac{c_n^2}{n^2}\sum_{i=1}^n (\tilde{H}_1(i) - \tilde{T}_n)^2  = c_n^2\bcom^2   
\end{align}

We start by lower bounding the expectation. We can decompose the expectation as:
\begin{align*}
E\left\{\frac{1}{n^2} \sum_{i=1}^n(\tilde{H}_1(i) - \tilde{T}_n)^2 \right\} &=
\frac{1}{n^2}\sum_{i=1}^n E(\tilde{H}_1(i)) - \theta_n)^2  - \frac{1}{n}E (\tilde{T}_n - \theta_n)^2
\\ &=  \frac{1-1/n}{n} \cdot \var(\tilde{H}_1(i)) - \frac{ 2}{n}\cov( \tilde{H}_1(i), \tilde{H}_1(j))  
\end{align*}
For the covariance term, by law of iterated covariance,
\begin{align*}
\cov( \tilde{H}_1(i), \tilde{H}_1(j)) = \cov( H_1(i), H_1(j)) = O\left( \frac{\rho_n^{2s}}{(n \rho_n)n }\right)     
\end{align*}

For the variance term, we have that:
\begin{align*}
\var(\tilde{H}_1(i)) &=  \var(H_1(i) - \theta_n) + E\bbb{ \var(\tilde{H}_1(i) - H_1(i)) \ | A,X \ } \\
&= \tau_n^2\bb{1+O\bb{\frac{1}{n\rho_n}}} + \Theta\left( \frac{\rho_n^s}{Nn} \right)    
\end{align*}
where the latter bound follows from:
\begin{align*}
E(\var( \tilde{H}_1(i) - H_1(i) \ | \ A,X)) &=  \frac{1}{N} \left\{ \frac{1}{(n-1)!}\sum_{\pi} E(H_\pi - \theta_i)^2 - E(H_1(i) - \theta_i)^2 \right\} 
\\ &= \frac{1}{N}\{ E( \var( H_\pi) \ | \ X_i) - E( \var( H_1(i) \ | \ X_i))\}
\\ &= \Theta\left( \frac{\rho_n^s}{Nn} \right) - O\left(\frac{\rho_n^{2s}}{Nn}\right)
\end{align*}

Therefore,
\begin{align}\label{eq:Ekn2}
E\left\{\frac{1}{n^2} \sum_{i=1}^n(\tilde{H}_1(i) - \tilde{T}_n)^2 \right\} &= \frac{\sigma_n^2}{r^2}\bb{1+O\bb{\frac{1}{n\rho_n}}} + \Theta\left(\frac{\rho_n^s}{Nn^2}\right)
\end{align}

Now we will establish concentration.  Observe that:
\begin{align*}
 \frac{1}{n} \left\{\frac{1}{n}\sum_{i=1}^n(\tilde{H}_1(i) - \tilde{T}_n)^2 \right\} &= 
 \underbrace{\frac{1}{n} \left\{ \frac{1}{n}\sum_{i=1}^n (\tilde{H}_1(i) - H_1(i))^2 - (\tilde{T}_n - \hat{T}_n)^2 \right\}}_{V_1} 
 \\ &+ \underbrace{\frac{1}{n} \left\{ \frac{1}{n} \sum_{i=1}^n (H_1(i) - \theta_n)^2 - (\hat{T}_n-\theta_n)^2 \right\} }_{V_2}
 \\ &+ \underbrace{\frac{2}{n} \left[\frac{1}{n}\sum_{i=1}^n \{\tilde{H}_1(i) - H_1(i)\} \cdot \{ H_1(i) - \theta_n\} - \{\tilde{T}_n - \hat{T}_n\} \cdot \{ \hat{T}_n - \theta_n\} \right] }_{V_3}
\end{align*}
Note that,
\begin{align}\label{eq:EV1}
    E\bbb{V_1}=\frac{1}{n}E\bbb{\var(\tilde{H}_1(i)|A,X)}\bb{1-\frac{1}{n}}=\tilde{\sigma}_n^2(1-1/n)
\end{align}
\bk
Observe that we may interpret $V_1$ and $V_2$ as sample variances and $V_3$ as a covariance term. We will first bound the variance of $V_2$ and $V_3$.

 For $V_2$, we apply results from \citet{zhang-xia-network-edgeworth}. Let:
\begin{align*}
V_2^\prime = \frac{1}{n^2}\sum_{i=1}^n \left\{ \frac{1}{{n-1 \choose r-1}}\sum_{S} h(X_{S \cup i}) - T_n  \right\}^2   
\end{align*}



\begin{align}\label{eq:varv2}
\  \  P\left( \frac{|r^2V_2 - \sigma_n^2|}{\sigma_n^2} > \epsilon/3  \right)
 \leq & \ P\left( \frac{r^2|V_2 -V_2'|}{\sigma_n^2} > \epsilon/6  \right) + P\left( \frac{|r^2V_2' - \sigma_n^2|}{\sigma_n^2} > \epsilon/6  \right)
\end{align}

By Lemma $3$c) of the above reference, the first term on the R.H.S of the above equation is $O_P(1/n\rho_n)$ and thus goes to zero.  For the second term, Lemma 3d) of the above reference implies that it is $ O_P(1/\sqrt{n})$; thus the second term goes to zero.  
 When $  1/n^2\rho_n^s \ll N \ll 1/n\rho_n^s$, we have that $\sigma_n^2/\nu_n^2 \rightarrow 0$, thus $V_2/\nu_n^2 \xrightarrow{P} 0$ and $V_1$ dominates.  \bk


 For $V_3$, observe that, $E(V_3|A,X) = 0$.  For $E\bbb{\var(V_3 |A,X)}$, rewriting the sample covariance, we have that
\begin{align}
\label{eq:v3-bound}
\var(V_3 |A,X) & \leq  \frac{2}{n^4} \sum_{i=1}^n \{(H_1(i) - \theta_n)^2 + (\hat{T}_n - \theta_n)^2\}  \cdot \var(\tilde{H}_1(i) - H_1(i) \ | \ A,X) 
\end{align}
Now define:
\begin{align*}
%
\alpha_{i,1} &= \frac{1}{n^{2r-2}} \sum_{S}\{H(A_{S \cup i}) - \theta_n\}^2 \\ 
\alpha_{i,2} &= \frac{1}{n^{2r-2}} \sum_{S,T}\{H(A_{S \cup i}) - \theta_n\}  \{H(A_{T \cup i}) - \theta_n\}  \\
\beta_{i,1} &= \frac{1}{n^{r+1}} \sum_{S'}\{H(A_{S' \cup i}) - \cei\}^2 \\
\beta_{i,2} &= \frac{1}{n^{2r-1}} \sum_{S',T' : S' \cap T' = \emptyset}\{H(A_{S' \cup i}) - \cei\} \{H(A_{T' \cup i}) - \cei\}  \\ 
\end{align*}
It follows that, for some $C>0$,
\begin{align*}
&\frac{1}{n^4N}\sum_{i=1}^n\{(H_1(i) - \theta_n)^2  \cdot \var(\tilde{H}_1(i) - H_1(i) \ | \ A,X) \\
&\leq \frac{C}{Nn^4} \sum_{i=1}^n (\alpha_{i,1}\beta_{i,1} + \alpha_{i,1}\beta_{i,2} + \alpha_{i,2}\beta_{i,1} + \alpha_{i,2}\beta_{i,2})     
\end{align*}
For $\mathbb{E}[\alpha_{i,1}\beta_{i,1}]$, the leading term is when the node sets are disjoint (other than $i$); there are $O(n^{2r-2})$ such nodes and thus the contribution is $O(\rho_n^{2s}/n^{r+1})$. For $E[\alpha_{i,1}\beta_{i,2}]$, observe that the product is nonzero only if $|S\cap S^\prime|\ge 1$ and $|S\cap T^\prime|\ge 1$. Therefore, for the dominant term, there are $O(n^{3r-5})$ terms with a total contribution $O(\rho_n^{3s-2}/n^{r+2})$.  For $E[\alpha_{i,2}\beta_{i,1}]$, the dominant term is when $S,S^\prime,$ and $T^\prime$ disjoint; in this case the contribution is $O(\rho_n^{3s}/n^2)$.  Finally for $E[\alpha_{i,2}\beta_{i,2}]$, note that $S \cup T$ must share at least two nodes with $S^\prime \cup T^\prime$; there are $O(n^{4r-6})$ such terms and thus the dominant term is $O(\rho_n^{4s-2}/n^3)$. 

Now, the second term follows similar reasoning, but we elaborate further on the bound below for completeness. Let:
\begin{align*}
\alpha_{i,1}' &= \frac{1}{n^{2r}} \sum_{S}\{H(A_{S}) - \theta_n\}^2 \\ 
\alpha_{i,2}' &= \frac{1}{n^{2r}} \sum_{S,T}\{H(A_{S}) - \theta_n\}  \{H(A_{T}) - \theta_n\} 
\end{align*}
Then, for the second term in Eq \ref{eq:v3-bound}:
\begin{align*}
\frac{1}{n^4N}\sum_{i=1}^n\{(\hat{T}_n - \theta_n)^2  \cdot \var(\tilde{H}_1(i) - H_1(i) \ | \ A,X) \leq \frac{C}{Nn^4} \sum_{i=1}^n (\alpha_{i,1}'\beta_{i,1} + \alpha_{i,1}'\beta_{i,2} + \alpha_{i,2}'\beta_{i,1} + \alpha_{i,2}'\beta_{i,2})     
\end{align*}    

By analogous reasoning, $E[\alpha_{i,1}'\beta_{i,1}] = O(\rho_n^{2s}/n^{r+2})$, $E[\alpha_{i,1}'\beta_{i,2}] = O(\rho_n^{3s-2}/n^{r+3})$, $E[\alpha_{i,2}'\beta_{i,1}] = 
O(\rho_n^{3s}/n^{3})$, $E[\alpha_{i,2}'\beta_{i,2}] = O(\rho_n^{4s-2}/n^{5})$.

Therefore,
\begin{align}
\var(V_3)=E\{\var(V_3|A,X)\} &= O\left( \frac{\rho_n^{2s}}{N n^{r+4}} \right) + O\left( \frac{\rho_n^{3s-2}}{N n^{r+5}}\right) + O\left( \frac{\rho_n^{3s}}{Nn^5} \right)+ O\left( \frac{\rho_n^{4s-2}}{Nn^6} \right)  \notag\\
&=O\bb{\frac{\rho_n^{2s}}{Nn^{r+4}}}+O\bb{\frac{\rho_n^{3s}}{Nn^5}}\label{eq:varv3}
\end{align}
\bk

Now for $V_1$, by law of iterated variance, it suffices to bound $\var( E(V_1 \ | \ A,X))$ and $E(\var(V_1 \ | \ A,X))$.  

For the former, using Theorem~\ref{thm:clt-var-conc}, we have that:
\begin{align*}
\var\bb{E(V_1 \ | \ A,X)} &= \var\bb{\frac{1-1/n}{n^2}\sum_{i=1}^n\var(\tilde{H}_1(i) \ | A,X)}  = (1-1/n)^2\var\bb{\tilde{\sigma}_{|A,X}^2}=o\bb{\frac{\rho_n^{2s}}{N^2n^4}}
\end{align*}


    Now for $E(\var(V_1 | A,X)$, let $V_{1,1} := \frac{1}{n^2} \sum_{i=1}^n\{\tilde{H}_1(i) - H_1(i))\}^2$, let $V_{1,2}:=\frac{1}{n}(\tilde{T}_n-\hat{T}_n)^2$. 



Note that,
\begin{align}\label{eq:V11-decomp}
      \var\bb{V_{1,1}|A,X}&=\frac{1}{n^4} \sum_i \var\bbb{\bb{\tilde{H}_1(i)-H_1(i)}^2|A,X}\notag\\
      &\leq \frac{1}{n^4} \sum_i E\bbb{\bb{\tilde{H}_1(i)-H_1(i)}^4|A,X}
\end{align}

Thus, using Lemma~\ref{lem:rosenthal:V11},
\begin{align*}
    E\bbb{\var\bb{\frac{1}{n^2}\sum_i \bb{\tilde{H}_1(i)-H_1(i)}^2|A,X}}\leq \frac{\rho_n^s}{n^6N^3}\bb{1+N n\rho_n^s}
\end{align*}

For $E(\var(V_{1,2})|A,X)$, note that, using Jensen's inequality,
\begin{align*}
    \var(V_{1,2}|A,X)&=\frac{1}{n^6}\var\bbb{\bb{\sum_i (\tilde{H}_1(i)-H_1(i))^2+\sum_{i\neq j}(\tilde{H}_1(i)-H_1(i))(\tilde{H}_1(j)-H_1(j))}|A,X}\\
    &\leq \frac{1}{n^6}\var\bbb{\sum_i (\tilde{H}_1(i)-H_1(i))^2|A,X}\\
    &+\frac{1}{n^6}\var\bb{\sum_{i\neq j}(\tilde{H}_1(i)-H_1(i))(\tilde{H}_1(j)-H_1(j))|A,X}
\end{align*}
We have the first part. For the second part, we see that:
\begin{align*}
    &\frac{1}{n^6}\var\bb{\sum_{i\neq j}(\tilde{H}_1(i)-H_1(i))(\tilde{H}_1(j)-H_1(j))|A,X}\\
    &=\frac{1}{n^6}\sum_{i\neq j}\var\bb{(\tilde{H}_1(i)-H_1(i))(\tilde{H}_1(j)-H_1(j))|A,X}\\
    &=\frac{1}{n^6}\sum_{i\neq j}E\bb{(\tilde{H}_1(i)-H_1(i))^2|A,X}E\bb{(\tilde{H}_1(j)-H_1(j))^2|A,X}\\
    &\leq \frac{1}{n^6} \bb{\sum_i E\bb{(\tilde{H}_1(i)-H_1(i))^2|A,X}}^2\\
    &\leq \frac{1}{n^6} \bb{\sum_i \var(\tilde{H}_1(i)|A,X)}^2
    \\&=\frac{1}{n^2}\var(\tilde{\sigma}_{|A,X}^2)+\frac{1}{n^6} \bb{\sum_i E(\tilde{H}_1(i)-H_1(i))^2}^2\\
    &\stackrel{(i)}{=}\frac{1}{n^2}o_P\bb{\frac{\rho_n^{2s}}{N^2n^4}}+\frac{1}{n^4}\bb{\frac{\rho_n^s}{nN}}^2=O\bb{\frac{\rho_n^{2s}}{n^6N^2}},
\end{align*}

where the second term in step $(i)$ follows from Lemma~\ref{lem:E-perm-to-subset}.

By Equation~\ref{eq:Ekn2}, we have the expectation of $\kappa_n^2$ is given by
$\Theta\bb{\frac{\rho_n^{s}}{Nn^2}(1+N n\rho_n^s)}$. 


For $V_1$,  we have, for some constant $C_1$:
\begin{align*}
    \var\bb{V_1} = \frac{\rho_n^s}{n^6N^3}\bb{C_1(1+N n\rho_n^s)+o(Nn^2\rho_n^{s})}.
\end{align*}
\bk

Finally  we have:
\begin{align}
    &P\bb{\frac{|V_1+V_2+V_3-(EV_1+\sigma_n^2/r^2)|}{\sigma_n^2/r^2+\tilde{\sigma}_n^2}
    \geq \epsilon}\notag\\
    &\leq P\bb{\frac{|V_1-EV_1|}{\sigma_n^2/r^2+\tilde{\sigma}_n^2}\geq \frac{\epsilon}{3}}+P\bb{\frac{|V_2-\sigma_n^2/r^2|}{\sigma_n^2/r^2+\tilde{\sigma}_n^2}}+P\bb{\frac{|V_3|}{\sigma_n^2/r^2+\tilde{\sigma}_n^2}\geq \frac{\epsilon}{3}}\notag\\
    &\leq \frac{9\var(V_1)}{(\sigma_n^2/r^2+\tilde{\sigma}_n^2)^2\epsilon^2}+P\bb{\frac{|V_2-\sigma_n^2/r^2|}{\sigma_n^2/r^2}\geq \frac{\epsilon}{3}}+\frac{9\var(V_3)}{(\sigma_n^2/r^2+\tilde{\sigma}_n^2)^2\epsilon^2}\label{eq:cheby}
\end{align}

When $Nn\rho_n^s\gg1$, $\sigma_n^2/r^2+\tilde{\sigma}_n^2$ is dominated by $\frac{\rho_n^{2s}}{n}$. When $Nn\rho_n^s\ll1\ll Nn^2\rho_n^s$, it is dominated by $\frac{\rho_n^{s}}{n^2N}$.  It is easy to check that when $Nn\rho_n^s\gg 1$, or $Nn\rho_n^s\ll 1 \ll Nn^2\rho_n^s$, we have $\var(V_1)/(E \bcom^2)^2\rightarrow 0$, which establishes consistency of the first term on the RHS of Eq~\ref{eq:cheby}. The second term was shown to go to zero in Eq~\ref{eq:varv2}. For the third term, using Eq~\ref{eq:varv3}, we see: 
\begin{align*}
    \frac{\var(V_3)}{(\sigma_n^2/r^2+\tilde{\sigma}_n^2)^2}=\begin{cases}
    O\bb{\frac{1}{Nn^{r+2}\rho_n^{2s}}}+O\bb{\frac{1}{N\rho_n^sn^3}}&\mbox{When $N\gg \frac{1}{n\rho_n^s}$}\\
    O\bb{\frac{N}{n^r}}+O\bb{\frac{N\rho_n^s}{n}}&\mbox{When $\frac{1}{n^2\rho_n^s}\ll N\ll \frac{1}{n\rho_n^s}$}\\
    \end{cases}
\end{align*}
For simple cycles and acyclic subgraphs, one can verify that the above goes to zero.

Thus, we have:
\begin{align*}
    \frac{V_1+V_2+V_3-(\sigma_n^2/r^2+\tilde{\sigma}_n^2)}{\sigma_n^2/r^2+\tilde{\sigma}_n^2}&=\frac{V_1+V_2+V_3-(\sigma_n^2/r^2+EV_1/(1-1/n))}{\sigma_n^2/r^2+\tilde{\sigma}_n^2}\\
    &=\frac{V_1+V_2+V_3-(\sigma_n^2/r^2+EV_1)}{\sigma_n^2/r^2+\tilde{\sigma}_n^2}-\frac{1}{n-1}\frac{\tilde{\sigma}_n^2(1-1/n)}{\sigma_n^2/r^2+\tilde{\sigma}_n^2}\\
    &=\frac{V_1+V_2+V_3-(\sigma_n^2/r^2+EV_1)}{\sigma_n^2/r^2+\tilde{\sigma}_n^2}-O\bb{\frac{1}{n}}\\
    &\stackrel{P}{\rightarrow} 1,
\end{align*}
where the last line follows from Eq~\ref{eq:cheby}.


 \end{proof}



\begin{lemma}\label{lem:rosenthal:V11}
Suppose that $R$ is acyclic or a simple cycle, $r \geq 2$, and $n \rho_n \rightarrow \infty$.  Then, \bk 
\begin{align*}
    E\bbb{\bb{\tilde{H}_1(i)-H_1(i)}^4}
    \leq \frac{\rho_n^s}{n^3 \bk N^3}\bb{1+N n\rho_n^s}
\end{align*}
\end{lemma}
\begin{proof}

Note that,
\begin{align}
     &E\bbb{\bb{\tilde{H}_1(i)-H_1(i)}^4|A,X}\notag\\
    &\leq \frac{1}{N^4}\bb{C_1\sum_{j=1}^N E\bbb{\bb{H_{\pi_j}(i)-H_1(i)}^4|A,X}+C_2\bb{\sum_{j=1}^N E\bbb{\bb{H_{\pi_j}(i)-H_1(i)}^2|A,X}}^2}\label{eq:varh1}
\end{align}

Since, we now take an expectation of the RHS of above, we bound the first part with Rosenthal's inequality. For a given permutation $\pi$,
\begin{align}\label{eq:varh1-fourth}
&n^4 E\bbb{\bb{H_{\pi}(i)-H_1(i)}^4}\notag\\
&=n^4 E E\bbb{\bb{H_{\pi}(i)-H_1(i)}^4|X_i}\notag\\
&\leq C_1 E \sum_{S\in \mathcal{S}_\pi}  E\bbb{\bb{H(A_{S\cup i})-H_1(i)}^4|X_i}+C_2 E\bb{\sum_{S\in \mathcal{S}_\pi}  E\bbb{\bb{H(A_{S\cup i})-H_1(i)}^2|X_i}}^2\notag\\
&\leq C_1' n\rho_n^s+ C_2 E\bb{\sum_{S\in \mathcal{S}_\pi}  E\bbb{\bb{H(A_{S\cup i})-H_1(i)}^2|X_i}}^2\notag\\
&\leq C_1' n\rho_n^s+C_2'n^2\rho_n^{2s}
\end{align}

For the the last step Eq~\ref{eq:varh1-fourth}, we used
\begin{align*}
E\bb{\sum_{S\in \mathcal{S}_\pi}  E\bbb{\bb{H(A_{S\cup i})-H_1(i)}^2|X_i}}^2\leq n \sum_{S\in \mathcal{S}_\pi} E \bbb{E\bb{H(A_{S\cup i})-H_1(i)}^2|X_i}^2\leq n^2\rho_n^{2s}
\end{align*}
As for the second term in Eq~\ref{eq:varh1}, we have that:
\begin{align}
\begin{split}
\label{eq:rosenthal-var-term}
E\bb{\sum_{j=1}^N E\bbb{\bb{H_{\pi}(i)-H_1(i)}^2|A,X}}^2&\leq C_1 N\sum_{j=1}^N E\bb{E\bbb{\bb{H_{\pi_j}(i)-H_1(i)}^2|A,X}}^2\\
&\leq C_1' N^2 (\rho_n^{2s}/n^2+\rho_n^{4s-2}/n^2)\\
&= O(N^2\rho_n^{2s}/n^2)
\end{split}
\end{align}

To obtain the last step, observe that
\begin{align*}
    &E\bb{E\bbb{\bb{H_{\pi}(i)-H_1(i)}^2|A,X}}^2= E\bb{\frac{1}{(n-1)!}\sum_\pi \bb{H_{\pi}(i)-H_1(i)}^2}^2\\ 
    & = E\bb{\frac{1}{(n-1)!}\sum_\pi (H_{\pi}(i)-\cei)^2-(H_1(i)-\cei)^2}^2\\
    &\leq 2E\bb{\frac{1}{(n-1)!}\sum_\pi \bb{H_{\pi}(i)-\cei}^2}^2+2E\bbb{\bb{H_1(i)-\cei}^4}\\
    &\stackrel{(i)}{\leq} 4E\bb{\frac{1}{(n-1)!}\sum_\pi \bb{H_{\pi}(i)-\cei}^2}^2\\
    &\leq  4E\bb{\frac{1}{n^2(n-1)!}\sum_{\pi}\bb{\sum_{S\in\mathcal{S}_\pi(i)} \bb{H(A_{S\cup i}-\cei}^2+\sum_{\stackrel{S\cap T=\phi}{S,T\in\mathcal{S}_\pi(i)}} (H(A_{S\cup i})-\cei)(H(A_{T\cup i})-\cei)}}^2\\
    &\leq 4E\bb{\frac{C H_1(i) + {\cei}^2}{n}+C'H'_1(i)}^2 = O\bb{\frac{\rho_n^{2s}}{n^2}}+O\bb{\frac{\rho_n^{4s-2}}{n^2}}
\end{align*}

Note that:
\begin{align*}
    &E(H'_1(i)^2|X_i)\\&=\frac{C}{n^{4(r-1)}}\sum_{\substack{S\cap T=\phi\\S'\cap T'=\phi}}E\bbb{H(A_{S\cup i})-\cei)H(A_{T\cup i})-\cei)H(A_{S'\cup i})-\cei)H(A_{T'\cup i})-\cei)|X_i}\\
    &=O\bb{\frac{\rho_n^{4s-2}}{n^2}} + O\left( \frac{\rho_n^{2s}}{n^{2(r-1)}} \right)
\end{align*}
The last line is true because the terms with the largest contribution are those with $|S\cap S'|=1$ and $|T\cap T'|=1$.  The latter term corresponds to the variance terms (i.e. $S=S'$, $T=T'$), which is $O\left( \frac{\rho_n^{2s}}{n^2} \right)$ for $r \geq 2$. 

Step (i) follows from Jensen's inequality:
\begin{align*}
    \bb{H_1(i)-\cei}^2\leq \frac{1}{(n-1)!}\sum_\pi\bb{H_{\pi}(i)-\cei}^2
\end{align*}

Putting everything together,
\begin{align*}
   &E \bbb{\left(\tilde{H}_1(i)-H_1(i)\right)^4}
   \leq \frac{\rho_n^s}{N^3n^3}+\frac{\rho_n^{2s}}{n^2N^3}+\frac{\rho_n^{2s}}{N^2n^2}=O\bb{\frac{\rho_n^s}{N^3n^3}(1+N n\rho_n^s)},
\end{align*}
concluding the proof. 
\end{proof}

\begin{theorem}

Suppose that $\rho_n \rightarrow 0$, $r \geq 2$ and either $R$ is acyclic and $n \rho_n \rightarrow \infty$ or $R$ is a simple cycle and $n^{r-1} \rho_n^{r} \rightarrow \infty$.
When $N\gg 1/n^2\rho_n^s$, we have 
\begin{align*}
\sup_{t \in \mathcal{R}} |P( \tilde{T}_{n,L}^* - \tilde{T}_n   \leq t) - P( \tilde{T}_n - \theta_n   \leq t)     | \xrightarrow{P} 0.    
\end{align*}

\end{theorem}
\begin{proof}
Now, we verify the Berry-Esseen condition.  Observe that, for some $C>0$, the third central absolute moment of the bootstrap conditioned on the data $\psi_i$ may be upper-bounded as:
\begin{align*}
\psi_i \leq \frac{CE|\xi_i - 1|^3}{n^3} \left( |\tilde{H}_1(i) - H_1(i)|^3 +  |H_1(i) - \theta_n|^3 + |\tilde{T}_n - \hat{T}_n|^3 + |\hat{T}_n - \theta_n|^3   \right)       
\end{align*}
By Jensen's inequality, the latter two terms are again lower order. Therefore, by Lemma \ref{lemma:third-moment-randomized} and Lemma \ref{lem:proveconv}:
\begin{align*}
\Psi_n := \sum_{i=1}^n \psi_i = \underbrace{O_P\left( \frac{\rho_n^s}{n^4 N^2} \right) + O_P\left( \frac{\rho_n^{3s/2}}{n^{7/2} N^{3/2}} \right)}_{\text{contribution of randomization}} + \underbrace{O_P\left( \frac{\rho_n^{3s}}{n^2} \right)}_{\text{contribution of signal}}   
\end{align*}
Moreover, by definition of $\bootsd^2$ in Equation~\ref{eq:def-vnhat}, and 
$\nu_n^2:=c_n^2(\sigma_n^2/r^2+\tilde{\sigma}_n^2)$
from Theorem~\ref{lem:bootvarconc}, we have that when $N \gg 1/(n^2 \rho_n^s)$,
\begin{align}
\frac{\bootsd^2}{\nu_n^2} \xrightarrow{P} 1,     
\end{align}
where by Equation~\ref{eq:Ekn2}
,%
\begin{align*}
 \nu_n^2 = \underbrace{\Theta\left( \frac{\rho_n^s}{Nn^2} \right)}_{\text{contribution of randomization}} +  \underbrace{\Theta\left( \frac{\rho_n^{2s}}{n} \right)}_{\text{contribution of signal}}    
\end{align*}
Observe that when $N \gg 1/n \rho_n^s$, the signal dominates for both $\Psi_n$ and $\nu_n$.  When $1/n^2 \rho_n^s \ll N \ll 1/n \rho_n^s$, the randomization dominates and when $N= \Theta(1/n \rho_n^s)$ the contributions are the same order; in all cases:
$\Psi_n/\bootsd^{3/2} = o_P(1)$. Thus the desired CLT holds for the bootstrap.   
\end{proof}

\begin{lemma}\label{lem:E-perm-to-subset}
\begin{align*}
\E\bbb{\mathrm{var}\{\tdh(i)\mid A,X\}}=\Theta\left(\frac{r\rho_n^{s}}{Nn}\right).
\end{align*}
\begin{proof}\normalfont
\begin{align*}
 \E\bbb{\var\{\tdh(i)\mid A,X)\}}&= \E\bbb{\var\left\{\frac{\sum_jH_{\pi_j}(i)}{N} \mid A,X\right\}} =\E\left(\frac{1}{N^2}\sum_j\var\{H_{\pi_j}(i)\mid A,X\}\right)\\
 &=\E\left(\frac{1}{N^2}\sum_j\E\{(H_{\pi_j}(i)-\theta)^2\mid A,X\}\right)\\
 &=\E\left(\frac{1}{N^2}\left[\sum_j\sum_{S\in \mathbb{S}_{\pi_j}}\frac{E\{H(A_{S\cup i}-\theta)^2\mid A,X\}}{(\frac{n-1}{r-1})^2}\right]\right)\\
 &=\frac{\var\{H(A_{S\cup i})\}}{N\frac{n-1}{r-1}}
\end{align*}


We further have,
\begin{align*}
   \var\{\hsxi\}=\var[E\{\hsxi\mid X\}]+E[\var\{\hsxi\mid X\}]=\Theta(\rho_n^s).
\end{align*}
Thus, we arrive at
\begin{align*}
\E\bbb{\var\{\tdh(i)\mid A,X)\}}=\Theta\left(\frac{r\rho_n^{s}}{Nn}\right).    
\end{align*}

\end{proof}
\end{lemma}

\section{Proof of Proposition~\ref{prop:ew}}
\label{sec:suppprop3}
\begin{proof}\normalfont
 In what follows, we prove Proposition~\ref{prop:ew} holds for Edgeworth expansion of a standardized count functional.  Our argument here is closely related to  \cite{zhang-xia-network-edgeworth}, thus we do not present the complete proof here. They have showed in Theorem 3.1 in the above reference, that under same conditions, Edgeworth Expansion for studentized $\hat{T}_n$, denote as $\tilde{G}_n(u)$ here, has the same property as Proposition~ \ref{prop:ew}. We have first derived our Edgeworth Expansion formula in eq~\ref{eq:edgeworth-gn} for standardized $\hat{T}_n$ instead of studentized $\hat{T}_n$ and we state the form of the characteristic function of $G_n(u)$ below:

\begin{proposition}\label{lm:chf_g}
We have:
\begin{align*}
    \psi_{G(n)}(t)&:= \int e^{itu}dG_n(u)\\
    &= e^{-\frac{t^2}{2}}\left(1- it^3\frac{1}{6n^{1/2}\tau_n^3}\left[E\{g_1^3(X_1)\}+3(r-1)E\{g_1(X_1)g_1(X_2)g_2(X_1,X_2)\}\right]\right).
\end{align*}
\end{proposition}

Our standardized $\hat{T}_n$, denote as $\tilde{T}_n$ can be decomposed into
\begin{equation*}
    \tilde{T}_n:=\frac{\hat{T}_n-\theta_n}{\sigma_n} = \frac{T_n-\theta_n}{\sigma_n}+ \frac{\tnhat - T_n}{\sigma_n}= T_{n,1}+T_{n,2}  + O_P\left(\frac{1}{n} \right)+ R_n,
\end{equation*}
where
\begin{align*}
    &T_{n,1}= \frac{1}{n^{1/2}\tau_n}\sum_{i=1}^n g_1(X_i),
    &T_{n,2}= \frac{r-1}{n^{1/2}(n-1)\tau_n}\sum_{i<j}g_2(X_i,X_j), \ 
    &R_{n} =  \frac{\tnhat - T_n}{\sigma_n}.
\end{align*}

 We will begin by bounding $R_n$. Similar to the theory for U-statistics, the behavior is largely determined by a linear term.

Let:
\begin{equation*}\label{eq:linear-bern}
    R_{n,1} = \text{Linear part of } \frac{\tnhat - T_n}{\sigma_n}.
\end{equation*}
where the linear part has the form:
\begin{align*}
R_{n,1} = \frac{1}{{n \choose 2 }}\sum_{i<j} c_{ij} \left\{A_{ij} - \E(A_{ij} \mid  X_i, X_j) \right\}
\end{align*}
for $c_{ij} = c_{ij}(X_i,X_j,\rho_n) \asymp \rho_n^{-1}n^{-1/2}$ defined in Section 7 of the above reference.   Theorem 3.1(b) of the above authors establishes that:     

\begin{equation*}
     R_{n}- R_{n,1}= O_P( \mathcal{M}(n,\rho_n,R)),
\end{equation*}
Under the assumed sparsity conditions, given $\mathbf{X}$, the distribution of $R_{n,1}$ permits the following (uniform) approximation by a Gaussian-distributed variable $Z_n$:
\begin{equation*}
    \sup_{u} \bigg{|}F_{R_{n,1}|X}(u)-F_{Z_n} \bigg{|}=O_P\left(\frac{1}{\rho_n^{1/2}n}\right),
\end{equation*}
where $Z_n\sim N(0,\frac{\sigma_w^2}{n\rho_n})$ and $\sigma_w^2$ is defined as the variance of Eq~\ref{eq:linear-bern}. Note that $\sigma_w \asymp 1$ when $n \xrightarrow[]{} \infty$.

Now to prove our theorem, we will show the three equations below. 
\begin{equation}\label{eq:decomp1}
    \sup_{u} \bigg{|}F_{\tnt}(u)-F_{T_{n,1}+T_{n,2}+R_n} \bigg{|}=O\left( \mathcal{M}(n,\rho_n,R)\right),
\end{equation}
\begin{equation}\label{eq:decomp2}
    \sup_{u} \bigg{|}F_{T_{n,1}+T_{n,2}+R_n}(u)-F_{T_{n,1}+T_{n,2}+ Z_n} \bigg{|}=O\left(\frac{1}{\rho_n^{1/2}n}\right),
\end{equation}
\begin{equation}\label{eq:decomp3}
    \sup_{u} \bigg{|}F_{T_{n,1}+T_{n,2}+ Z_n}-G_n(x) \bigg{|}=O\left(\frac{1}{n}\right),
\end{equation}

We prove Eq~\ref{eq:decomp3} using Esseen’s smoothing lemma from Section XVI.3 in \cite{feller-vol-2},
\begin{equation}\label{eq:essensmoothing}
\begin{split}
     &\sup_{u} \bigg{|}F_{T_{n,1}+T_{n,2}+ Z_n}(u)-G_n(x) \bigg{|} \\&\leq c_1\int_{-\gamma}^{\gamma}\frac{1}{t} \bigg{|}\psi_{F_{T_{n,1}+T_{n,2}+ Z_n}}(u) - \psi_{G_n}(t)  \bigg{|}dt + c_2 \sup_u\frac{G_n'(u)}{\gamma},
\end{split}
\end{equation}
where $\psi$ is the characteristic function. $\gamma$ is set to $n$. We omit the proof here as it is not hard to check by breaking the integral into $|t| \in (0,n^{\epsilon})$,$(n^{\epsilon},n^{1/2})$ and $(n^{1/2},n)$. Using similar arguments as Lemma 8.3 of \cite{zhang-xia-network-edgeworth}, we have Eq~\ref{eq:essensmoothing} and thus Eq~\ref{eq:decomp3} hold for our characteristic function in Proposition~\ref{lm:chf_g}.  It is also not hard to check that, using similar arguments of the above reference, under Assumption~\ref{ass:sparse}, Eq~\ref{eq:decomp1} and  Eq~\ref{eq:decomp2} hold given Eq~\ref{eq:decomp3}.
\end{proof}
\bk

 \section{Edgeworth expansion for weighted bootstrap - proofs of Theorem~\ref{thm:mbq-firstorder} and Corollary~\ref{thm:mbm-firstorder}}
 \label{sec:suppthm23}
\bk




Using Eq~(10), we express our quadratic bootstrap statistic as:
\begin{align}\label{eq:cfbqform}
    \cfbq=\frac{\sum_i (\w_i-1)\gi}{n^{1/2}\hat{\tau}_n} +  &\frac{(r-1)\sum_{1\leq i <j \leq n}(\w_i\w_j-\w_i-\w_j+1) \tilde{g}_2(i,j)}{n^{1/2}(n-1)\hat{\tau}_n}  
\end{align}
We will first prove Theorem~\ref{thm:mbq-firstorder}. However in order to prove it we state a slightly different version of Theorem 3.1 in~\cite{wang-jing-weighted-bootstrap-u-statistics}\bk. The main difference is that one condition in the original lemma is not fulfilled in our case.  In particular, Bernoulli noise with $\rho_n \rightarrow 0$ blows up some terms that are needed to bound the error associated with the Edgeworth expansion.  However, a thorough examination reveals that the argument carries through with some modifications. \bk  

Let
\begin{subequations}
\begin{align}\label{eq:alldefswangjing}
    \kn&=\frac{1}{n^{3/2}B_n^2}\sum_{1\leq i<j\leq n} b_{ni} b_{nj}  d_{nij}\E\{\y_1\y_2\psi(\y_1,\y_2)\}\\
    \lnone(x)&=\sum_{j=1}^n \left\{\E\Phi(x- b_{nj}\y_j/B_n)-\Phi(x)\right\}-\frac{1}{2}\Phi''(x)\\
    \lntwo(x)&=-\kn\Phi'''(x)\\
      E_{2n}(x)&= \Phi(x) + \lnone(x) + \lntwo(x),
\end{align}
\end{subequations}

\begin{lemma}\label{lem:wangjing}
Consider the following expression.
\begin{align}\label{eq:wangjingdecomp}
    V_n=\frac{1}{B_n}\sum_j  b_{nj} \y_j+\frac{1}{n^{3/2}} \sum_{i < j}   d_{nij} \psi(\y_i,\y_j),
\end{align}
where $B_n^2=\sum_j b_{nj}^2$. Let $\beta:=\E(|\y_1|^3)$ and $\ro=\E \{\psi^2(\y_1,\y_2)\}$, and let $\E(\y_1)=0$, $\E(\y_1^2)=1$ and  $\kappa(X_1)>0$ . Furthermore, let $\E\{\psi(\y_1,\y_2)\mid \y_t\}=0$ for all $1\leq t\leq n$.
For some constants $\ell_1$, $\ell_2$, $\ell_3$ the sequence $b_{n,i}$ satisfies 
\begin{equation}\label{eq:wb-conds}
    \begin{split}
    &\frac{1}{n}\sum_{i=1}^n b_{n,i} ^2 \geq l_1 >0, \ \ \frac{1}{n}\sum_{i=1}^n |b_{n,i}| ^3 \leq l_2 \leq \infty, \\
    \end{split}
\end{equation}

Furthermore, define $\alpha_{i}:=\frac{1}{n}\sum_{j\neq i} d_{nij}^2$.
and for sufficiently large $k$, define:
\begin{align}
    \label{eq:ddef}
    l_{4,n}=\frac{1}{n}\sum_{i=1}^n\alpha_{i}, \ \ \ \ s_n^2=\frac{1}{n}\sum_i\alpha_i^2-(l_{4,n})^2, \ \ \ \ l_{5,n} = l_{4,n} + k s_n
\end{align}

 If $\beta,\kappa(\y_1)$ and $\ro$ are bounded, then,
\begin{align*}
    \sup_x|P(V_n\leq x)-E_{2n}(x)|= O\left(\frac{ l_{5,n}\log n}{n^{2/3}}\right),
\end{align*}

\end{lemma}

  Intuitively, arguments for establishing rates of convergence for the Edgeworth expansions require comparing the characteristic function of the random variable of interest with the Fourier transform of the Edgeworth expansion.  To this end, the respective integrals are broken up into several pieces.  The bounds required in  (\ref{eq:wb-conds}) are used to estimate the error of the Edgeworth expansion in some of these steps, but appear as constants and are suppressed in the Big-O notation.  

On the other hand, as previously mentioned, it turns out that certain terms that appear as constants in \cite{wang-jing-weighted-bootstrap-u-statistics} blow up when perturbed by sparse network noise and appear in the rate.  In particular, the term $l_{5,n}$ arises from needing to bound $\frac{1}{m}\sum_{i=1}^m \alpha_i$ for all $m \leq M$ for some $M$ large enough.  

Since the data is fixed,  we may view $\alpha_1, \ldots, \alpha_n$ as constants. We therefore have the liberty of choosing a ``good set" in which $\alpha_i$ are well-behaved.  Without loss of generality, we may label these elements $\{\alpha_1, \ldots, \alpha_M\}$; the corresponding multiplier random variables are still independent. Even when there is no randomness, it turns out that a large proportion of $\{\alpha_1, \ldots \alpha_n\}$ must be within $k$ sample standard deviations of the sample mean $l_{4,n}$ for $k$ large enough. This observation, which we believe is novel in the bootstrap setting, allows us to establish a tight bound for $\frac{1}{m}\sum_{i=1}^m \alpha_i$ for all $m \leq M$.  We state this lemma below.   
%
\begin{lemma}
\label{lemma:deterministic-deviation-identity}
Let $x_1 ,\ldots, x_n$ be constants in $\mathbb{R}$ and let $\bar{x}_n = \frac{1}{n}\sum_{i=1}^n x_i$ and $s_n^2 = \frac{1}{n}\sum_{i=1}^n (x_i - \bar{x})^2$  Define the set:
\begin{align*}
    \Gamma_k =  \left\{x_i \geq \bar{x}_n + k s_n  \right\}
\end{align*}
Then,
\begin{align*}
|\Gamma_k| \leq \frac{n}{k^2} 
\end{align*}
\end{lemma}
\begin{proof}\normalfont
Observe that:
\begin{align*}
s_n^2 &\geq \frac{1}{n}\sum_{i \in \Gamma_k} (x_i-\bar{x}_n)^2
\\ &\geq \frac{1}{n}\sum_{i \in \Gamma_k} k^2 s_n^2  \implies  |\Gamma_k | \leq \frac{n}{k^2}  
\end{align*}
\end{proof}

\begin{remark}
Our lemma is closely related to concentration of sums sampled without replacement from a finite population. In fact, it implies the without-replacement Chebychev inequality; see, for example, Corollary 1.2 of \citet{serfling-concentration-without-replacement}.
\end{remark}

We will show that $\cfbq$ can be written as Eq~\ref{eq:wangjingdecomp}, with carefully chosen $\{ b_{ni}\}$ and $\{ d_{nij}\}$'s. We now present some accompanying Lemmas to show that Eq~\ref{eq:wb-conds} is satisfied with probability tending to 1. Proofs of Lemmas~\ref{lem:wangjing},~\ref{lem:proveconv},
and~\ref{lem:thirdabsmoment} are provided in following subsections.

We present some useful results shown in \cite{zhang-xia-network-edgeworth} which we will use later in proofs of our theorems. 
\begin{lemma}\label{lem:moments-z/x}
Let $\hat{\tau}_n^2=\sum_i \gi^2/n$.  
We have, 
\begin{enumerate}
	\item For acyclic graphs, if $n\rho_n\rightarrow \infty$, and for cyclic graphs, if  $n\rho_n^r\rightarrow\infty$, from \cite{zhang-xia-network-edgeworth} Lemma 3.1 and its proof, we have:
	\begin{align}
	\frac{\hat{\tau}_n^2}{\tau_n^2}&=1+O_P\left(\frac{1}{n\rho_n}\right)+O_P\left(\frac{1}{\sqrt{n}}\right)\\
	E(|\gi-g_1(X_i)|/\rho_n^s)^2&=O(1/n\rho_n)
	\end{align}
	\item Under Assumption~\ref{ass:sparse}, we have
\begin{align}\label{eq:zxempiricalg13}
\left|\frac{\sum_j \hat{g}_1(i)^3}{n}-\E\{g_1(X_1)^3\}\right|=O_P\left(\rho_n^{3s-0.5}n^{-1/2}\right),
\end{align}
\begin{align}\label{eq:zxempiricalg1g1g2}
\left|\frac{\sum_{i<j} \gi\gj\ggij}{{n\choose 2}}-\E\{g_1(X_1)g_1(X_2)g_2(X_1,X_2)\}\right|=O_P\left(\rho_n^{3s-0.5}n^{-1/2}\right),
\end{align}
and
\begin{align}\label{eq:tauhaterr}
|\hat{\tau}_n^3-\tau_n^3|=O_P(\rho_n^{3s}/n^{1/2}).
\end{align}
\end{enumerate}

\end{lemma}

 \begin{lemma}\label{lem:proveconv-mean}
 Under the sparsity assumptions in Assumption~\ref{ass:sparse}, for large enough $C$, 
 \begin{align*}
 P\left(\frac{1}{n^2}\sum_i \sum_{j\neq i}\tilde{g}_2(i,j)^2\geq C\rho_n^{2s-1}\right)\rightarrow 1
 \end{align*}
 \end{lemma}

\begin{lemma}\label{lem:proveconv}
Under the sparsity conditions in Assumption~\ref{ass:sparse} and for some arbitrary $\epsilon>0$,
\begin{align*}
    P\left(\frac{\sum_i |\gi/\rho_n^s|^3}{n}\leq c\right)&\rightarrow 1\\
    P\left(\frac{\sum_i |\gi/\rho_n^s|^2}{n}\geq c'\right)&\rightarrow 1,
\end{align*}
for positive constants $c,c'$ not depending on $n$.
\end{lemma}
\begin{lemma}\label{lem:thirdabsmoment}
Let $\w_1$ be generated from the Gaussian product distribution. We have $\E|\w_1-1|^3<\infty$. 
\end{lemma}
Now we are ready to provide the proof.
\begin{proof}\normalfont[ of Theorem~\ref{thm:mbq-firstorder}]
It is easy to see from Eq~\ref{eq:cfbqform} that $\cfbq$ can be expressed as:

\begin{equation*}
\begin{split}
     \cfbq& = \frac{\sum_i b_{n,i}\y_i }{B_n} + \frac{1}{n^{1/2}(n-1)}\sum_{1\leq i <j \leq n} \psi(\y_i,\y_j)d_{n,ij},
\end{split}
\end{equation*}
where we have:
\begin{subequations}\label{eq:wangjingbd}
\begin{align}
&\y_i=\w_i-1 ,\\
    &b_{n,i}= \frac{\gi}{\rho_n^s}, \\
     &B_n^2= \sum_{i=1}^n b_{n,i}^2,\\
    &d_{n,ij}= \frac{r-1}{\hat{\tau}_n}\tilde{g}_2(i,j) \times \frac{n}{n-1},\\
    &\psi(\w_i,\w_j)=\w_i\w_j-\w_i-\w_j+1.
\end{align}
\end{subequations}
Note that since $\hat{\tau}_n^2=\sum_i \gi^2/n$. Thus we use $B_n^2=n\hat{\tau}_n^2/\rho_n^{2s}$. Thus, $B_n^2=\sum_i b_{n,i}^2$. Furthermore, Lemma~\ref{lem:thirdabsmoment} shows that our $\w_i-1$ random variables have finite $\E\{|\xi_i-1|^3\}$.

Lemma~\ref{lem:proveconv} shows that the conditions in Eq~\ref{eq:wb-conds} are satisfied on a high probability set of $A,X$. 

Using Lemma~\ref{lem:proveconv}, we see that the first two conditions in Eq~\ref{eq:wb-conds} are satisfied with probability tending to one under Assumption~\ref{ass:sparse}. Since $B_n^2/n=\sum_i b_{n,i}^2/n$ converges to a positive constant (see Lemma~\ref{lem:moments-z/x})
, the first condition holds. 
 Now, we need to bound $\ell_{4,n}$ and $s_n$ as defined Eq~\ref{eq:ddef}. First, 
 let $\beta_{n,i}:=\sum_{j\neq i} \tilde{g}_2(i,j)^2/n$ and $\bar{\beta}_n=\sum_i\beta_{n,i}/n$. Also let $\gamma_n=\sum_i \beta_{n,i}^2/n-\bar{\beta}_n^2$. 
Then $\ell_{4,n}=C\beta_{n,i}/\hat{\tau}_n^2$.
Note that using Lemma~\ref{lem:proveconv-mean}, we have, with probability tending to one, $\ell_{4,n}\leq C\rho_n^{-1}$. From  Lemma~\ref{lem:moments-z/x}, we have $\hat{\tau}_n$ is $\Theta(\rho_n^s)$.
 Furthermore, let $\hat{G}_2(i,j):=\hhij -h_2(X_i,X_j)$. 

We have
\begin{align*}
    \hat{G}_2(i,j)^2=\underbrace{\hat{G}_2(i,j)^2-\E\{\hat{G}_2(i,j)^2\mid X\}}_{\delta_{ij}}+\underbrace{\E\{\hat{G}_2(i,j)^2\mid X\}}_{O(\rho_n^{2s-1})}
\end{align*}
We now will establish the $O(\rho_n^{2s-1})$ bound stated above for the second term. 
Let $\mathcal{S}^{ij}_r$ denote all subsets of size $r$ not containing $i,j$.
\begin{align*}
   \E\{\hat{G}_2(i,j)^2\mid X\}= \frac{\sum_{S,T\in \mathcal{S}^{ij}_r} \E\{H(A_{ij\cup S})H(A_{ij\cup T})\mid X\}}{{n-2\choose r-2}^2}
\end{align*}
In the above sum the terms with $|S\cap T|=0$ dominate, and for each of them the conditional expectation is bounded a.s. by $O(\rho_n^{2s-1})$ because of the boundedness of the graphon.
Now note that:
\begin{align*}
   \tilde{g}_2(i,j)^2&\leq 3\left[\{\hhij -h_2(X_i,X_j)\}^2+\{h_2(X_i,X_j)-\theta_n\}^2+(\tnhat-\theta_n)^2\right]\\
   &\leq 3\{\hhij -h_2(X_i,X_j)\}^2+O(\rho_n^{2s-1})\\
\beta_{n,i}&\leq \frac{1}{n}\sum_{j\neq i}\delta_{ij}+O(\rho^{2s-1})\\
\gamma_n^2&\leq \frac{1}{n}\sum_i \beta_{n,i}^2\leq \frac{1}{n}\sum_i \left\{\frac{1}{n}\sum_{j\neq i}\delta_{ij}+O(\rho^{2s-1})\right\}^2\\
&\leq O(\rho^{4s-2})+\underbrace{\frac{1}{n}\sum_i \frac{1}{n}\sum_{j\neq i}\delta_{ij}^2}_{A}
\end{align*}

Now note that, $\E(\delta_{ij})=\E\{\E(\delta_{ij}\mid X)\}=0$. Thus, for all $i$,
\begin{align*}
\E(A)=\frac{1}{n}\sum_i\E \left\{\frac{1}{n}\sum_{j\neq i}\delta_{ij}\right\}^2=\frac{1}{n}\sum_i\var\left(\frac{1}{n}\sum_{j\neq i}\delta_{ij}\right)=O(\rho_n^{4s-3}/n)
\end{align*}

Thus, we have, for a large enough $C$,
\begin{align*}
    P\bb{\gamma_n^2\geq C \rho_n^{4s-2}}\leq P\bb{A\geq C'\rho_n^{4s-2}}\leq O\bb{\frac{\E(A)}{\rho_n^{4s-2}}}=O\bb{\frac{1}{n\rho_n}}
\end{align*}

Therefore, since $s_n^2=\frac{(r-1)^2n^2}{(n-1)^2}\gamma_n/\hat{\tau}_n^2$, we have with probability tending to one,
$l_{4,n}+ks_n=O(\rho_n^{-1})$.

\bk

Since the first two conditions in eq~\ref{eq:wb-conds} are satisfied with probability tending to one, from \cite{wang-jing-weighted-bootstrap-u-statistics}  Theorem 3.1, we have,
\begin{equation*}
   \sup_u  \bigg{|}L_{1n}(u)+\frac{E(
   \xi_i-1)^3}{6B_n^3}\sum_{i=1}^nb_{n,i}^3\Phi'''(u) \bigg{|} = o_P(n^{-1/2}), 
\end{equation*}

Now we see that,
using the definitions of $L_{1n}$, $L_{2n}$ in Eq~\ref{eq:alldefswangjing}, plugging in definitions of $ b_{ni}$ and $ d_{nij}$'s from Eq~\ref{eq:wangjingbd}, and using the fact that $\E[\y_i\y_j\psi(\y_i,\y_j)]=\E[(\w_i-1)^2(\w_j-1)^2]=1$,
$$\sup_u |E_{2n}(u)-\gnhat(u)|=o_P(n^{-1/2}).$$

Therefore, putting all the pieces together we see that

\begin{equation}\label{eq:tnasbserror}
    \sup_{u}\bigg{|}P^*\left(\frac{\cfbq-\tnhat}{\hat{\sigma}_n}\leq u\right) -\gnhat(u)\bigg{|} = o_p\left(n^{-1/2}\right)+ O_P\left(\frac{\log n}{n^{2/3}\rho_n}\right)
\end{equation}

\end{proof}

Now we are ready to finish the proof of Corollary~\ref{thm:mbm-firstorder}.
\begin{proof}\normalfont[of Corollary~\ref{thm:mbm-firstorder}]
Here we take care of the error term in the Hoeffding projection in Eq~\ref{eq:boot-hoeff}. 
Set $X=\frac{\cfbm-\tnhat}{\shat}$, $Y=\cfbq$. From Eq~\ref{eq:boot-hoeff}, we see that $X=Y+R_n$, where $R_n =O_P(\mathcal{\delta}(n,\rho_n,R))$. Using Eq~\ref{eq:tnasbserror}, we see that on a high probability set,
\begin{align*}
    F_{Y}(u+a)-F_{Y}(u)&\leq |F_Y(u+a)-\gnhat(u+a)|+|\gnhat(u+a)-\gnhat(u)|+|\gnhat(u)-F_Y(u)| \\
    &\leq Ca+O\left(\frac{\log n}{n^{2/3}\rho_n}\right)
\end{align*}
Therefore, using Lemma 8.2 in~\cite{zhang-xia-network-edgeworth},
$$\sup_u\left|P^*\left(\frac{\cfbm-\tnhat}{\shat}\leq u\right)-\gnhat(u)\right|=o_P(n^{-1/2}) +  O_P\left(\frac{\log n}{n^{2/3}\rho_n}\right).$$

\end{proof}

\begin{proof}\normalfont[ of Lemma~\ref{lem:convEW}]

If we can establish Eq~\ref{eq:zxempiricalg1g1g2} and Eq~\ref{eq:tauhaterr} from Lemma~\ref{lem:moments-z/x} \bk for our empirical moments, we will get the desired result.
Note that our empirical moments involve the first term as well as a slight variation of the second term, which is given below.
\begin{align*}
   \widehat{E}_n\{ g_1(i)g_1(j)g_2(i,j) \}&=\frac{\sum_{i<j}\gi\gj\tilde{g}_2(i,j)}{{n\choose 2}}
\end{align*}
We will show that this follows from Eq~\ref{eq:zxempiricalg1g1g2}.
\begin{align*}
    &\frac{\sum_{i<j} \gi\gj\ggij}{{n\choose 2}}\\
    &=\frac{\sum_{i<j} \gi\gj(\tilde{g}_2(i,j)-\gi-\gj)}{{n\choose 2}}\\
    &=\frac{\sum_{i<j} \gi\gj\tilde{g}_2(i,j)}{{n\choose 2}}-\frac{\sum_{i\neq j}\gi\gj(\gi+\gj)}{2{n\choose 2}}\\
    &=\frac{\sum_{i<j} \gi\gj\tilde{g}_2(i,j)}{{n\choose 2}}-\frac{\sum_{i\neq j}\gi^2\gj}{{n\choose 2}}\\
    &=\frac{\sum_{i<j} \gi\gj\tilde{g}_2(i,j)}{{n\choose 2}}-\frac{(\sum_{i}\gi^2)(\sum_j\gj)-\sum_i \gi^3}{{n\choose 2}}\\
    &\stackrel{(i)}{=}\frac{\sum_{i<j} \gi\gj\tilde{g}_2(i,j)}{{n\choose 2}}+\frac{\sum_i \gi^3}{{n\choose 2}}\\
    &\stackrel{(ii)}{=}\frac{\sum_{i<j} \gi\gj\tilde{g}_2(i,j)}{{n\choose 2}}+O_P\left(\frac{\rho_n^{3s}}{n}\right)
\end{align*}
(i) uses the fact that $\sum_i\gi=0$. (ii) uses the fact that $\E\{g_1(X_1)^3\}=O(\rho_n^{3s})$ along with Eq~\ref{eq:zxempiricalg13}.
Hence from Eq~\ref{eq:zxempiricalg1g1g2} we have:
\begin{align*}
    &\left|\frac{\sum_{i<j} \gi\gj\tilde{g}_2(i,j)}{{n\choose 2}}-\E\{g_1(X_1)g_1(X_2)g_2(X_1,X_2)\}\right|\\
    &=\max\left\{O_P\left(\frac{\rho_n^{3s}}{n}\right), O_P\left(\rho_n^{3s-\frac{1}{2}}n^{-1/2}\right) 
    \right\}\\
    &=O_P\left(\rho_n^{3s-\frac{1}{2}}n^{-1/2}\right).
\end{align*}
This, along with Eqs~\ref{eq:zxempiricalg13} and~\ref{eq:tauhaterr} yields the result.
\end{proof}

\subsection{Proof of Lemma~\ref{lem:proveconv-mean}}
\begin{proof}\normalfont
Recall the definition of $\hhij$ from Eq~\ref{eq:hhij}.
\begin{align*}
    \tilde{g}_2(i,j)&=\hhij-\tnhat=\{\hhij -h_2(X_i,X_j)\}+\{h_2(X_i,X_j\}-\theta_n)-  (\tnhat-\theta_n)\notag\\
    \tilde{g}_2(i,j)^2 
    &\leq 3\left[\{\hhij -h_2(X_i,X_j)\}^2+\{h_2(X_i,X_j)-\theta_n\}^2+(\tnhat-\theta_n)^2\right].
\end{align*}
Since $\var(\tnhat)=O(\rho_n^{2s}/n)$ and the second term is bounded a.s. due to our boundedness assumption. We will just prove that
$\sum_{j\neq i} \{\hhij -h_2(X_i,X_j)\}^2/(n-1)\rho_n^{2s}$ is bounded with high probability. It is not hard to check 
that
$$\E\{(\hhij -h_2(X_i,X_j))^2/\rho_n^{2s}\}=O(1/\rho_n)$$
Therefore,
\begin{align*}
    \sum_{j\neq i}\E\{\ggij^2/(n\rho_n^{2s})&=O(1/\rho_n)\}
\end{align*}
Furthermore, let $\hat{G}_2(i,j):=\hhij -h_2(X_i,X_j)$. 
We have
\begin{align*}
    \hat{G}_2(i,j)^2=\underbrace{\hat{G}_2(i,j)^2-\E\{\hat{G}_2(i,j)^2\mid X\}}_{\delta_{ij}}+\underbrace{\E\{\hat{G}_2(i,j)^2\mid X\}}_{O(\rho_n^{2s-1})}
\end{align*}
 We now will establish the $O(\rho_n^{2s-1})$ bound stated above for the second term.  
Let $\mathcal{S}^{ij}_r$ denote all subsets of size $r$ not containing $i,j$.
\begin{align*}
   \E\{\hat{G}_2(i,j)^2\mid X\}= \frac{\sum_{S,T\in \mathcal{S}^{ij}_r} \E\{H(A_{ij\cup S})H(A_{ij\cup T})\mid X\}}{{n-2\choose r-2}^2}
\end{align*}
In the above sum the terms with $|S\cap T|=0$ dominate, and for each of them the conditional expectation is bounded a.s. by $O(\rho_n^{2s-1})$ because of the boundedness of the graphon.

We will analyze $\sum_i\sum_{j\neq i}\delta_{ij}$. Note that $\E(\delta_{ij}\mid X)=0$.
\begin{align}\label{eq:vardel2}
    \var\left(\frac{1}{n^2}\sum_i\sum_j \delta_{ij}\mid X\right)=\frac{\sum_i\sum_j \var(\delta_{ij}\mid X)+\sum_{i,k,k\neq i}\sum_{j,\ell,j\neq \ell} \cov(\delta_{ik},\delta_{j\ell}\mid X)}{n^4}
\end{align}

\begin{align*}
    \delta_{ij}=\frac{1}{{n-2\choose r-2}^2}\sum_{S,T\in \mathcal{S}^{ij}_r} \underbrace{H(A_{ij\cup S})H(A_{ij\cup T})-\E\{H(A_{ij\cup S})H(A_{ij\cup T})\mid X\}}_{H_{ij}'(S,T)}
\end{align*}
For variance, we have:
\begin{align*}
\var(\delta_{ij})&=\E\{\var(\delta_{ij}\mid X)\}\\
&=\frac{ \sum\limits_{S_1\neq T_1, S_2\neq T_2\in \mathcal{S}^{ij}_r}\E\{\cov(H_{ij}'(S_1,T_1),H_{ij}'(S_2,T_2)\mid X)\}}{{n-2\choose r-2}^4}\\
\end{align*}

The dominant term in the above sum is the one with $S_1, S_2, T_1, T_2$ all disjoint. Consider any other term in the above sum where any pair of the subsets have $p$ nodes, $d$ edges in common and the rest are disjoint. In this case there are $2(r-2-p)+2(r-2)+p=4(r-2)-p$ choices of nodes and the number of edges are lower bounded by $4(s-1)+1-d=4s-3-d$ (since all pairs have $\{i,j\}$ in common). When $p\geq 1$, for acyclic graphs, $d\leq p-1$ and for general subgraphs with a cycle, $d\leq {p\choose 2}$. Thus, for $p\geq 0$, we have:
\begin{align*}
\frac{O\left(n^{4(r-2)-p}\rho_n^{4s-3-d}\right)}{{n-2\choose r-2}^4}=O(\rho_n^{4s-3})\times O\left(\frac{1}{n^p \rho_n^d}\right)
\end{align*}
Note that for acyclic graphs, it is easy to see that under our sparsity conditions the above is dominated by $p=0$. For general cyclic graphs, since $\rho_n= \omega(n^{-1/r})$, note that, since $p\leq r$,
\begin{align*}
n^p\rho_n^d\geq n^p\rho_n^{p(p-1)/2}\geq n^{p\left(1-\frac{p-1}{2r}\right)}\geq n^{\frac{p(r+1)}{2r}} \rightarrow \infty
\end{align*}
So, $\var(\delta_{ij})=O(\rho_n^{4s-3})$.

For covariance, for $i\neq j\neq k \neq \ell$, we have:
\begin{align*}
\cov(\delta_{ik},\delta_{j\ell})&=\E\{\cov(\delta_{ik},\delta_{j\ell}\mid X)\}\\
&=\frac{\sum\limits_{S_1\neq T_1 \in  \mathcal{S}^{ik}_r,S_2\neq T_2\in \mathcal{S}^{j\ell}_r} \E\{\cov (H_{ij}'(S_1,T_1),H_{ik}'(S_2,T_2)\mid X)\}}{{n-2\choose r-2}^4}
\end{align*}
Consider any two pairs of subsets with $p$ nodes and $d$ edges in common. First note that $p\geq 2$ in order to have a nonzero covariance. In this case there will be $4(r-2)-p$ choices for nodes, and $(2s-A_{ik})+(2s-A_{j\ell})-d\geq 4s-3-d$ edges.
\begin{align*}
=\frac{O\left(n^{4(r-2)-p}\rho_n^{4s-2-d}\right)}{{n-2\choose r-2}^4}=O\left(\frac{\rho_n^{4s-3}}{n^2}\right)\times O\left(\frac{1}{n^{p-2}\rho_n^{d-1}}\right).
\end{align*}
Note that, for acyclic graphs $d\leq p-1$ and hence the above is maximized at $p=2,d=1$ as long as $n\rho_n\rightarrow \infty$. 

For general cyclic subgraphs, $d\leq {p\choose 2}$.  Furthermore, since $p+2\leq r$, and $\rho_n=\omega(n^{-1/r})$, we have, for $p> 2$:
\begin{align*}
n^{p-2}\rho_n^{d-1}&=n^{p-2-\frac{1}{r}\left(\frac{p(p-1)}{2}-1\right)}\\
&=n^{p-2-\frac{(p-2)(p+1)}{2r}}=n^{(p-2)\left(1-\frac{p+1}{2r}\right)}\geq n^{(p-2)\frac{r+1}{2r}}\rightarrow \infty
\end{align*}
Thus under the conditions of Assumption~\ref{ass:sparse}, we have:
\begin{align*}
\cov(\delta_{ik},\delta_{j\ell})=O(\rho_n^{4s-3}/n^2)
\end{align*}
\bk
Step (i) is true, because conditioned on $X$, there needs to be at lease two nodes $u_1$, $u_2$ in common between $\{i,k\cup S_1\cup T_1\}$ and $\{j,\ell \cup S_2\cup T_2\}$ to have a nonzero covariance. This leads to only $4(r-2)-2$ choices, which dominates the sum. This along with Eq~\ref{eq:vardel2} gives us:
$$\var\left(\frac{1}{n^2}\sum_i\sum_{j \neq i} \delta_{ij}\right)=\E\left\{\var\left(\sum_i\sum_{j \neq i} \delta_{ij}/n^2\mid X\right)\right\}=O(\rho_n^{4s-3}/n^2).$$

Thus we have for large enough $C$, we have

\begin{align*}
   P\left( \frac{1}{n^2}\sum_i\sum_{j\neq i}\tilde{g}_2(i,j)^2\geq C\rho_n^{2s-1}\right)
   &\leq P\left( \bigg{|}\sum_i\sum_{j\neq i}\delta_{ij}/n^2+O(\rho_n^{2s-1})\bigg{|}\geq C\rho_n^{2s-1}\right)\\
    &\leq P\left( \sum_i\bigg{|}\sum_{j\neq i}\delta_{ij}/n^2\bigg{|}\geq C'\rho_n^{2s-1}\right)\\
    &\leq C''\frac{\rho_n^{4s-3}/n^2}{\rho_n^{4s-2}}=O\left(\frac{1}{n^2\rho_n}\right).
 \end{align*}

\end{proof}
\bk

\subsection{Proof of Lemma~\ref{lem:proveconv}}
\begin{proof}\normalfont
Let $\Delta_i:=|\gi-g_1(X_i)|/\rho_n^s$. We have:
\begin{align}\label{eq:b12}
    &\dfrac{\sum_{i}|\gi/\rho_n^s|^3}{n}\notag\\
    &\leq \dfrac{\sum_i \Delta_i^3}{n} + 3\dfrac{\sum_i |g_1(X_i)/\rho_n^s|\Delta_i^2}{n}+3\dfrac{\sum_i |g_1(X_i)/\rho_n^s|^2\Delta_i}{n}+\dfrac{\sum_i |g_1(X_i)/\rho_n^s|^3}{n}\notag\\
    &=B_1+B_2+B_3+B_4
\end{align}
First note that using the boundedness condition on the graphon, $|g_1(X_i)/\rho_n^s|$ is bounded. Hence $B_4\leq c$ a.s.
Using Lemma~\ref{lem:moments-z/x}\bk, we know that $E(\Delta_i)^2=O(1/n\rho_n)$.
Since $\sum_i \Delta_i\leq n^{1/2}{\sum_j \Delta_j^2}$, for the second term we have, for some $C > 0$ \bk:
\begin{align}\label{eq:b2}
    P(B_2\geq \epsilon)\leq \frac{n^{1/2} E\sum_i\Delta_i^2}{n\epsilon^2}\leq \frac{ C \bk }{n^{1/2}\rho_n\epsilon^2}
\end{align}
Furthermore,
\begin{align}\label{eq:b3}
    P(B_3\geq \epsilon)\leq \frac{ E\sum_i\Delta_i^2}{n\epsilon^2}\leq \frac{ C \bk}{n\rho_n\epsilon^2}.
\end{align}
By repeated application of Cauchy-Schwarz inequality, we have $(\sum_i x_i^3)^2\leq \sum_i x_i^2\sum_i x_i^4\leq (\sum_i x_i^2)^3$, we also have:
\begin{align*}
    P(B_1\geq \epsilon)\leq \frac{ E\sum_i\Delta_i^3}{n\epsilon^2}\leq \frac{(\sum_i E\Delta_i^2)^{3/2}}{n\epsilon^2}
    \leq \frac{C \bk }{n\rho_n^{3/2}\epsilon^2}
\end{align*}
Therefore, using the sparsity conditions in Assumption~\ref{ass:sparse}, we see that the first equation in the lemma statement is proved.

For the second, we use:
\begin{align*}
    \dfrac{\sum_{i}|\gi/\rho_n^s|^2}{n}&\geq \dfrac{\sum_i |g_1(X_i)/\rho_n^s|^2}{n} +\dfrac{\sum_i \Delta_i^2}{n} - 2\dfrac{\sum_i |g_1(X_i)/\rho_n^s|\Delta_i}{n}\\
    &=C_1+\alpha B_2-\beta B_3,
\end{align*}
where $\alpha,\beta$ are positive constants, and $B_2,B_3$ were defined in Eq~\ref{eq:b12}.
Using Assumption 2 part 1, \bk we see $C_1>0$, a.s. Also, now for a small enough constant $\epsilon$, using Eqs~\ref{eq:b2} and~\ref{eq:b3}, we see that the second equation in the lemma statement is proven.

\end{proof}

\subsection{Proof of Lemma~\ref{lem:wangjing}}
\begin{proof}\normalfont
Define the following quantities.
\begin{align*}
    \gjt&=\E\{\exp(it b_{nj}\y_j/B_n)\}
    \end{align*}
    Also define $\phi_{1,n}$ and $\phi_{2,n}$ as:
    \begin{align*}
    \phi_{1,n}(t)&=e^{-t^2/2}\left[1+\sum_j\{\gjt-1\}+\frac{t^2}{2}\right]\\
    \phi_{2,n}(t)&=-t^2\kn \et.
    \end{align*}
    Finally define,
    \begin{align*}
    S_n&=\frac{1}{B_n}\sum_j  b_{nj} \y_j,\qquad
    \del_{n,m}=\frac{1}{n^{3/2}}\sum_{i=1}^{m-1}\sum_{j=i+1}^n d_{nij}\psi(\y_i,\y_j)\\
\end{align*}
As in the original proof, we define:
\begin{subequations}
\label{eq:ltophi}
\begin{align}
    &\int_{-\infty}^\infty e^{itx} d\{\Phi(x)+\lnone\}dx=\phi_{1,n}(t)\\
    &\int_{-\infty}^\infty e^{itx} d \lntwo dx=it\phi_{2,n}(t)\\
    &\int_{-\infty}^\infty e^{itx}E_{2n}(x)= \phi_{1,n}(t)+it\phi_{2,n}(t)
\end{align}
\end{subequations}

Now, for some $c>0$ to be chosen later, from Esseen's smoothing lemma~\cite{petrov2012sums} and Eq~\ref{eq:ltophi} we have:
\begin{align}\label{eq:esseen}
    &\sup_x\bigg{|}P(V_n\leq x)-E_{2n}(x)\bigg{|}\notag\\
    &\leq \int_{|t|\leq n^{1-c}}|t|^{-1}|\E (e^{itV_n})-\phi_{1,n}(t)-it\phi_{2,n}(t)|dt+Cn^{c-1}\sup_x\left|\dfrac{d E_{2,n}(x)}{dx}\right|\notag\\
    &\leq \int_{|t|\leq n^{1-c}}|t|^{-1}|\E (e^{itV_n})-\phi_{1,n}(t)-it\phi_{2,n}(t)|dt+\frac{C_1(|\kn|+\beta)}{n^{1-c}}
\end{align}

The last line is true due to the following argument. Note that, for some $v_j$ in the $| b_{nj}\y_j/B_n|$ ball in the neighborhood of $x$, for $j\in\{1,\dots, n\}$,
\begin{align*}
    \frac{d\lnone(x)}{dx}&=\sum_{j=1}^n \left[\E\{\phi(x- b_{nj}\y_j/B_n)\}-\phi(x)\right]-\frac{1}{2}\Phi'''(x)\\
    &=\sum_{j=1}^n \E\left\{- b_{nj}\y_j/B_n\phi'(x)+ b_{nj}^2\y_j^2/2B_n^2\phi''(x)- b_{nj}^3\y_j^3/6B_n^3\phi'''(v_j)\right\}-\frac{1}{2}\Phi'''(x).
    \end{align*}
    Thus, we have:
    \begin{align*}
    \sup_x\left|\frac{d\lnone(x)}{dx}\right|&\leq \sum_{j=1}^n c_1\left\{ b_{nj}^2/B_n^2 |\phi''(x)|+ c_2| b_{nj}/B_n|^3\E(|\y_j^3|)|\phi'''(v_j)|\right\}+\frac{1}{2}|\Phi'''(x)| \\
    &\leq C+\E(|X_1|^3) \left(\sum_j | b_{nj}/B_n|^3\right)+C'\\
    &\leq C+\beta/n^{1/2}\leq C\beta\qquad\text{Since $\beta\geq 1$}
\end{align*}

Also note that, for any $\epsilon >0$, for $n$ large enough,
\begin{align*}
    \int_{|t|>n^\epsilon} |\phi_{1,n}(t)/t|dt&=O(1/n^{1-c})\\
    \int_{|t|>n^\epsilon} |\phi_{2,n}(t)/t|dt&=O(|\kn|/n^{1-c})
\end{align*}

Thus the main idea is that $E(e^{itV_n})$ behaves like $E(itS_n)+itE(itS_n\del_{n,n})$.
\begin{align*}
    &\int_{|t|\leq n^{1-c}}|t|^{-1}|\E (e^{itV_n})-\phi_{1,n}(t)-\phi_{2,n}(t)|dt\leq \sum_{j=1}^4I_{j,n}
\end{align*}

Going back to Eq~\ref{eq:esseen}, we break up the first part of the RHS into four parts, and the remainder gets absorbed into $O(|\kn+\beta|/n^{1-c})$ term in Eq~\ref{eq:esseen}. 

\begin{align*}
    |I_{1,n}|&=\int_{|t|<n^\epsilon}|t|^{-1}\left|\E (e^{it V_n})-E(itS_n)-itE(itS_n\del_{n,n})\right|dt\\
    |I_{2,n}|&=\int_{|t|<n^\epsilon}|t|^{-1}\left|\E (e^{it S_n})-\phi_{1,n}(t)\right|dt\\
    |I_{3,n}|&=\int_{|t|<n^\epsilon}\left|\E (\del_{n,n} e^{itS_n}) -\phi_{2,n}(t)\right|dt\\
    |I_{4,n}|&=\int_{n^\epsilon \leq |t|<n^{1-c}}|t|^{-1}\left|\E (e^{itV_n})\right|dt\\
\end{align*}
First we will bound some terms which will be used frequently. Since $ab\leq (a^2+b^2)/2$.
\begin{align}
    |\kn|&\leq \frac{C}{n^{3/2}B_n^2}\sum_{1\leq i<j\leq n} ( b_{ni}^2 b_{nj}^2+ d_{nij}^2)(1+\ro)\notag\\
    & \leq \frac{C(1+\ro)}{n^{3/2}B_n^2} \left\{\left(\sum_j b_{nj}^2\right)^2+\sum_{i<j} d_{nij}^2\right\}\notag\\
    &\leq \frac{C(1+\ro)}{n^{3/2}B_n^2}\left(B_n^4+l_{4,n} n^2\right)\notag\\
    &\leq \frac{C'(1+\ro)l_{4,n}}{n^{1/2}}\label{eq:knbound}
\end{align}

As for $\del_{n,n}$, we have:
\begin{align*}
    E\del_{n,n}^2=\frac{\ro}{n^3}\sum_{i=1}^{n-1}\sum_{j=i+1}^n  d_{nij}^2=\frac{\ro l_{4,n}}{n}
\end{align*}
\bk
Furthermore we will use:
\begin{align}\label{eq:rem}
    R(z):=e^{iz}-1-iz\qquad |R(z)|\leq |z|^\alpha \ \ \text{for all } \alpha\in [1,2]
\end{align}
We will first bound $I_{1,n}$. Using Taylor expansion, for some $|\eta|\leq 1$,

\begin{align*}
    |\ione|&\leq \int_{|t|<n^\epsilon}|t|^{-1}t^2/2  |\E (\del_{n,n}^2 e^{itS_n}e^{it\del_{n,n}\eta})|dt\notag\\
    &\leq 1/2\int_{|t|<n^\epsilon}|t|  \E(\del_{n,n}^2)dt
    \leq C\frac{(1+\ro)l_{4,n}}{n^{1-2\epsilon}}
\end{align*}

Next we bound $\itwo$. 
Using a similar argument in the proof of the original version of this theorem, we have:
\begin{align*}
    |\itwo|&\leq \frac{C_1}{B_n^4}\sum_j  b_{nj}^4+C_2\left(\frac{1}{B_n^3}\sum_j | b_{nj}|^3E(|X_1|^3)\right)^2\\
    &\leq C_1n^{-2/3}+C_2\ro^2/n
\end{align*}
Now we do $\ithr$. Denote $Z_j= b_{nj}\y_j/B_n$ and $\psi_{ij}= d_{nij}\psi(\y_i,\y_j)$.
First note that
\begin{align}\label{eq:ltwomain}
    \E\{\psi_{ij}e^{it(Z_i+Z_j)}\}=-t^2\ell_{ij}+\theta_{1ij}(t),
\end{align}
where we have:
\begin{align}\label{eq:lij}
    \ell_{ij}&=\E\left(|\psi_{ij}Z_iZ_j |\right)\leq  | b_{ni} b_{nj} d_{nij}|/B_n^2|\E\{\y_i\y_j\psi(\y_i,\y_j)\}|\leq \ro^{1/2}( b_{ni}^2 b_{nj}^2+ d_{nij}^2)/B_n^2.
\end{align}
Using Eq~\ref{eq:rem} and the fact that $\E\{\psi(\y_i,\y_j)\}=0$ and $\E\{\psi(\y_i,\y_j)\mid \y_i\}=0$,
\begin{align*}
    \tij&=\E(\psi_{ij}[it \{Z_i R(tZ_j)+Z_j R(tZ_i))+R(tZ_i)R(tZ_j)\}])\notag\\
    &\leq C|t|^{2.5}\E\left(|\psi_{ij}Z_iZ_j^{1.5}|+|\psi_{ij}Z_jZ_i^{1.5}|\right)\notag\\
    &\leq C|t|^{2.5}\E\{|\y_1\y_2^{1.5}\psi(\y_i,\y_j)|\}\left(| d_{nij} b_{ni} b_{nj}^{1.5}/B_n^{2.5}|+| d_{nij} b_{ni}^{1.5} b_{nj}/B_n^{2.5}|\right)\notag\\
    &\leq C|t|^{2.5}(\ro\beta)^{1/2}\left( d_{nij}^2+ b_{ni}^2| b_{nj}|^3+| b_{ni}|^3 b_{nj}^2\right)n^{-5/4}
\end{align*}
Using Eq~\ref{eq:ltwomain}, and setting $\prod_{k\neq i,j}\gamma_k(t)=e^{-t^2/2}+\ttij$ we see:
\begin{align*}
    E(\del_{n,n}e^{itS_n})&=n^{-3/2}\sum_{i<j}\E(\psi_{ij}e^{it S_n})
    =n^{-3/2}\sum_{i<j}\E\{\psi_{ij}e^{it(Z_i+Z_j)}\}\prod_{k\neq i,j}\gamma_k(t)\\
    &=n^{-3/2}\sum_{i<j}\{-t^2\lij+\tij(t)\}\left(e^{-t^2/2}+\ttij\right)\\
    &=n^{-3/2}\sum_{i<j}\{-t^2\lij e^{-t^2/2}+\tttij(t)\}=\underbrace{-\kn t^2 e^{-t^2/2}}_{\phi_{2,n}(t)}+\underbrace{n^{-3/2}\sum_{i<j} \tttij}_{\rnf},
\end{align*}

where using Lemma A.4 in~\cite{wang-jing-weighted-bootstrap-u-statistics}, for $|t|<n^\epsilon << n^{1/6}$
\begin{align*}
    |\ttij|&\leq \frac{C}{n^{1/2}}\left(\beta+\frac{ b_{nj}^2+ b_{ni}^2}{n^{1/2}}\right)(t^2+t^4)e^{-t^2/8}
\end{align*}
Furthermore, using Lemma A.4 and $\sum_i | b_{ni}|^3\leq \ell_2 n$
\begin{align*}
    |\tttij|&\leq t^2|\lij\ttij|+|\tij|\prod_{k\neq i,j}\gamma_k(t)\leq t^2|\lij\ttij|+4|\tij|e^{-t^2/8}\notag\\
    |\lij\ttij|&\leq \frac{C}{n^{1/2}}|\lij|\left(\beta+\frac{ b_{nj}^2+ b_{ni}^2}{n^{1/2}}\right)(t^2+t^4)e^{-t^2/8}
    \end{align*}
    Summing the above expression over $i<j$, we also have,
    \begin{align}\label{eq:tttij}
    &\sum_{i<j}|\lij\ttij|\notag\\
    &\leq C\lambda^{1/2}\left(\beta \frac{B_n^4+l_{4,n} n^2}{n^{3/2}}+\frac{\sum_{i<j}| d_{nij}|(| b_{ni}^3 b_{nj}|+| b_{ni} b_{nj}^3|)}{n^2}\right)(t^2+t^4)e^{-t^2/8}\notag\\
    & \leq C\lambda^{1/2}\left[\beta \frac{B_n^4+l_{4,n} n^2}{n^{3/2}}+\frac{c}{n^2}
    \underbrace{\left\{\left(\sum_{i<j} d_{nij}^2\right)\left(\sum_{i<j} b_{ni}^6 b_{nj}^2\right)\right\}^{1/2}}_{A}\right](t^2+t^4)e^{-t^2/8}
\end{align}
To bound (A) we see:
\begin{align*}
    (A)&\leq 
    \left\{(n^2l_{4,n})(n\ell_2)\sum_{i<j} b_{ni}^3 b_{nj}^2\right\}^{1/2}\leq n^{3/2}(\ell_2l_{4,n})^{1/2}\left\{\left(\sum_{i} b_{ni}^3\right)\left(\sum_j b_{nj}^2\right)\right\}^{1/2}\\
&\leq c' n^{5/2}l_{4,n}^{1/2}\ell_2
\end{align*}
Plugging this back in Eq~\ref{eq:tttij}, and assuming WLOG $l_{4,n}\geq 1$,
\begin{align*}
    \sum_{i<j}|\lij\ttij|&\leq C'\lambda^{1/2}\left(\beta \frac{B_n^4+l_{4,n} n^2}{n^{3/2}}+\frac{1}{n^2} n^{5/2}\ell_2l_{4,n}^{1/2}\right)(t^2+t^4)e^{-t^2/8}\notag\\
    &\leq C'l_{4,n}\ro^{1/2}\beta n^{1/2}(t^2+t^4)e^{-t^2/8}
\end{align*}
Finally, we also have:
\begin{align*}
    \sum_{i<j}|\tij|&\leq C|t|^{2.5}(\ro\beta)^{1/2}
    \left(l_{4,n} n^2+2\ell_2 n B_n^2 \right)n^{-5/4}\leq C|t|^{2.5}(\ro\beta)^{1/2}l_{4,n}n^{1/4}
\end{align*}
Finally we have, since $t^4\leq |t|+|t|^6$, and $|t|^{2.5}\leq |t|+|t|^6$,
\begin{align*}
    \rnf&\leq n^{-3/2}\sum_{i<j}|\tttij|\leq n^{-3/2}\left(\sum_{i<j}t^2|\lij\ttij|+4\sum_{i<j}|\tij|e^{-t^2/8}e^{-t^2/8}\right)\\
    &\leq \left(t^2l_{4,n}\ro^{1/2}(t^2+t^4)e^{-t^2/8}n^{-1}+|t|^{2.5}(\ro\beta)^{1/2}l_{4,n}n^{-3/4}e^{-t^2/8}\right)\\
    &\leq l_{4,n}\left(\ro^{1/2}\beta n^{-1}+(\ro\beta)^{1/2}n^{-3/4}\right)(|t|+|t|^6)e^{-t^2/8}\\
    &\leq C'l_{4,n}\left(\beta^2 n^{-1}+(\ro+\beta)n^{-3/4}\right)(|t|+|t|^6)e^{-t^2/8}\\
\end{align*}
Finally, for $\ithr$, we have:
\begin{align*}
    |\ithr|&\leq \int_{|t|\leq n^\epsilon} |t|^{-1}\rnf dt\\
    &\leq  C'l_{4,n}\left(\beta^2 n^{-1}+(\ro+\beta)n^{-3/4}\right)\int_{|t|\leq n^\epsilon} (1+|t|^5)e^{-t^2/8} dt\\
    &\leq C'1'l_{4,n} \left(\beta^2 n^{-1}+(\ro+\beta)n^{-3/4}\right)
\end{align*}

Now we will bound $\ifour$. 


Define $\Omega:=\{k:\min(1/2,\ell_2/\ell_1^{3/2})\leq n^{1/2}b_{n,k}/B_n\leq 2\ell_2/\ell_1^{3/2}\}$. Using Lemma A.5 in~\cite{wang-jing-weighted-bootstrap-u-statistics}\bk, we see that $|\Omega|\geq c_0 n$, for some $c_0\in (0,1)$.

Now, let $\Gamma:=\{i \ | \ \alpha_i\geq \bar{\alpha}+ks_n\}$. Applying Lemma \ref{lemma:deterministic-deviation-identity} and setting $k=\sqrt{2/c_0}$, we see that $|\Gamma^c|\geq n (1-c_0/2)$. Therefore, $|\Gamma^c \cap \Omega|\geq n c_0/2$. Let $k_0=\lfloor c_0/2\rfloor.$

WLOG assume $b_{n,1}\dots b_{n,k_0n}\in \Omega\cap\Gamma^c$ and $\ell_2/\ell_1^{3/2}\geq 1/2$. Now for $m\in [2,k_0 n]$, we have:

\begin{align*}
    S_m=\frac{1}{B_n}\sum_{k=1}^m b_{nk} \y_k\qquad S_m^{i,j}:=\frac{1}{B_n}\sum_{k\neq i,j}b_{nk} X_k
\end{align*}

For $1,\dots, m \leq k_0n$, we have:
\begin{align*}
    \frac{1}{mn}\sum_{i=1}\sum_{j=1,j\neq i} d_{nij}^2=\frac{1}{m}\sum_{i=1}^m\alpha_i\leq l_{4,n}+ks_n=:\ell_{5,n}
\end{align*}

As for $\del_{n,m}$, we have:
\begin{align}\label{eq:delnm}
    E(\del_{n,m}^2)=\frac{\ro}{n^3}\sum_{i=1}^{m-1}\sum_{j=i+1}^n  d_{nij}^2\leq \ro l_{5,n}\frac{m}{n^2}
\end{align}

Now we use the decomposition in~\cite{bickel1986edgeworth} (17)-(22).
\begin{align*}
    \E(e^{it V_n})&=\E\{e^{it (V_n-\del_{n,m})}e^{it \del_{n,m}}\}\\
    &=\E\{e^{it (V_n-\del_{n,m})}(1+it \del_{n,m})\}+R_{n,5}\\
    &=\E\{e^{it (V_n-\del_{n,m})}(1+it \del_{n,m})\}+Ct^2 \lambda  \lfn m/n^2\\
    &=\E\{e^{it (V_n-\del_{n,m})}\}+\frac{it}{n^{3/2}}\sum_{i=1}^m\sum_{j=i+1}^n 
    \underbrace{\E\{e^{it (V_n-\del_{n,m})}\psi_{ij}\}}_{D_{ij}}+Ct^2\ro \lfn \frac{m}{n^2}
\end{align*}
where the last line is obtained using Eqs~\ref{eq:rem} and~\ref{eq:delnm}, as follows:
\begin{align*}
    R_{n,5}\leq |\E\{e^{it (V_n-\del_{n,m})}t^2\del_{n,m}^2\}|\leq Ct^2 \lambda   l_{5,n} \bk m/n^2
\end{align*}
Note that $V_n-\del_{n,m}$ can be written as $S_{m-1}+Y_{m,n}$, where $Y_{m,n}$ does not depend on $\y_1,\dots,\y_{m-1}$. 
 So we will write:
 \begin{align*}
     \left|\sum_{i=1}^{m-1}\sum_{j=i+1}^n(D_{ij}) \right| &= \left|\sum_{i=1}^{m-1}\sum_{j=i+1}^n\E\{e^{it (S_{m-1}^{ij}+\psi_{ij}+Y_{m,n})}\psi_{ij}\} \right|\\
     &=\left|\sum_{i=1}^{m-1}\sum_{j=i+1}^n\E\{e^{it S_{m-1}^{ij}}\}\E\{e^{(\psi_{ij}+Y_{m,n})}\psi_{ij}\} \right| \\
     &\leq \sup_{i<j}|\E(e^{it S_{m-1}^{ij}})|\sum_{i=1}^{m-1}\sum_{j=i+1}^n\E(|\psi_{ij}|)\\
     &\leq \sup_{i<j}|\E(e^{it S_{m-1}^{ij}})|\sum_{i=1}^{m-1}\sum_{j=i+1}^n| d_{nij} | \E\{|\psi(\y_i,
     \y_j)|\}\\
     &\leq \ro^{1/2}\sup_{i<j}|\E(e^{it S_{m-1}^{ij}})|\sqrt{mn\sum_{i=1}^{m-1}\sum_{j=i+1}^n d_{nij}^2}\\
     &\leq \sqrt{\lambda  \lfn}\sup_{i<j}|\E(e^{it S_{m-1}^{ij}})| mn\\
 \end{align*}
 Plugging it back, we have:
 \begin{align}\label{eq:vncharac}
   |\E(e^{it V_n})|\leq   |\E(e^{it S_{m-1}})|+\frac{|t|}{n^{3/2}}\sqrt{\lambda  \lfn}\sup_{i<j}|\E(e^{it S_{m-1}^{ij}})| mn+Ct^2\lambda  \lfn \frac{m}{n^2}
 \end{align}
 Now, we have for $|t|\leq 1/4n^{1/2}/\E(|\y_1|^3)$
 \begin{align*}
     |\E(e^{itS_{m}})|\leq e^{-c_0 mt^2/n}\qquad |\E(e^{itS_{m}^{ij}})|\leq e^{-c_0(m-2)t^2/n}
 \end{align*}
 Taking $m=[6n\log n/c_0 t^2]+1$ (for a large enough $\epsilon$, this is still smaller than $k_0n$), from Eq~\ref{eq:vncharac} we have:
 \begin{align*}
     &\int_{n^\epsilon \leq |t|<1/4n^{1/2}/\E(|\y_1|^3)}|t|^{-1}\left|\E (e^{itV_n})\right|dt\\
     &\leq \int_{n^\epsilon \leq |t|<1/4n^{1/2}/\E(|\y_1|^3)} \left(\frac{e^{-c_0mt^2/n}}{|t|}+\frac{m}{n^{1/2}}\sqrt{\lambda  \lfn}e^{-c (m-2)t^2/n}+C|t|\lambda \lfn \frac{m}{n^2}\right)dt\\
     &\leq C'\lambda \lfn \frac{\log^2 n}{n}
 \end{align*}
 Now we will deal with the range $1/4n^{1/2}/\E(|\y_1|^3)\leq |t|\leq n^{1-c}$. Since $\kappa(\y_1)>0$, and hence for large enough $n$, 
 \begin{align*}
 |\gamma_k(t)|&\leq 1- \kappa(\y_1) \\
     |\E(e^{itS_{m}})|&\leq e^{-m\kappa(\y_1)}\\
     \E(e^{itS_{m}^{ij}})|&\leq e^{-(m-2)\kappa(\y_1)} 
 \end{align*}
 Using this in conjunction with Eq~\ref{eq:vncharac}, and setting $m=[4\log n/\kappa(\y_1)]+2$,
 \begin{align*}
      &\int_{1/4n^{1/2}/\E(|\y_1|^3)\leq |t|\leq n^{1-c}}|t|^{-1}\left|\E (e^{itV_n})\right|dt\\
      &\int_{1/4n^{1/2}/\E(|\y_1|^3)\leq |t|\leq n^{1-c}}\left(\frac{e^{-\kappa(\y_1)m}}{|t|}+\frac{m}{n^{1/2}}\sqrt{\lambda  \lfn}e^{-\kappa(\y_1)(m-2)}+C|t|\lambda \lfn \frac{m}{n^2}\right)\\
      &\leq C'\frac{\rho_n \lfn}{\kappa(\y_1)}\frac{\log n}{n^2}n^{2(1-c)}=C'\frac{\lambda  \lfn}{\kappa(\y_1)}\frac{\log n}{n^{2c}}
 \end{align*}
 
 Thus, using the bounds on $\ione,\itwo,\ithr$ and $\ifour$ along with Eq~\ref{eq:esseen}, Eq~\ref{eq:knbound} we get:
 \begin{align*}
     &\sup_x|P(V_n\leq x)-E_{2n}(x)|\\
     &\leq \sum_{i=1}^4 I_{n,i}+C'\frac{\beta+(1+\ro)l_{4,n}/n^{1/2}}{n^{1-c}}\\
     &\leq C\left(\frac{(1+\ro)l_{4,n}}{n^{1-2\epsilon}}+n^{-2/3}+l_{4,n}(\ro+\beta+\beta^2)n^{-3/4}+\frac{\lambda  \lfn}{\kappa(\y_1)}\frac{\log n}{n^{2c}}\right)+C'\frac{\beta+(1+\ro)l_{4,n}/n^{1/2}}{n^{1-c}}\\
     &\leq (l_{4,n}+ks_n)\frac{\log n}{n^{2/3}}
 \end{align*}
 The last line assumes $\beta,\ro$ and $\kappa(\y_1)$ are all bounded.
\end{proof}

\subsection{Proof of Lemma~\ref{lem:thirdabsmoment}}
\begin{proof}\normalfont
Let $X\sim N(1,c_1^2)$ and $Y\sim N(1,c_2^2)$ be two independent random variables. We have $\w_1=XY$.
\begin{align*}
    \E(|XY-1|^3) &\leq  \E(|XY|^3)+1+3\E(X^2|Y|)+3\E(|X|Y^2)\\
    &=\E(|X|^3)\E(|Y|^3)+3\E(X^2)\E(|Y|)+3\E(|X|)\E(Y^2)\\
    &<\infty
\end{align*}
The last step is true because both $\E(|X|^3)$ and $\E(|Y|^3)$ are bounded for bounded $c_1,c_2$.
\end{proof}

\section{Detailed results for smooth functions of counts}\label{sec:supp:proof-smf-ew-and-smf-boot}
 In this section, we establish Edgeworth expansions for smooth functions of counts for the bootstrap and show that they are close to Edgeworth expansions of the conditional expectation of the count statistic, which is a U-statistic.  To our knowledge, Edgeworth expansions for smooth functions are not explicitly stated in the literature even for U-statistics. It turns out that the non-negligible terms arising from a Taylor approximation of the smooth functional are of a form where a flexible Edgeworth expansion result of \citet{jing2010unified} may be invoked. Edgeworth expansions for smooth functionals are also considered in \citet{hall-bootstrap-edgeworth}, but the argument provided there requires multivariate Edgeworth expansions and depends heavily on the properties of cumulants of independent random variables, complicating extensions even to U-statistics.   

Since the Edgeworth expansion of the  conditional expectation requires a non-lattice condition, it is assumed below.  However, it is likely that this condition can be removed if one derives an Edgeworth expnasion for the count functional directly and uses a proof strategy similar to \cite{zhang-xia-network-edgeworth} that exploits the smoothing nature of Bernoulli noise.  
\bk

\subsection{Edgeworth expansion for smooth functions of counts}\label{sec:supp:sub-proof-smf-ew}
 In what follows, let $f: \mathbb{R}^d \mapsto \mathbb{R}$ denote the smooth function of interest, $\bu$ denote a $d$-dimensional vector of conditional expectations corresponding to scaled count functionals  $\{\frac{\tnhat^{(1)}}{\rho_n^{s_1}},\ldots, \frac{\tnhat^{(d)}}{\rho_n^{s_d}}\}$ given $X$, and $\bmu=E(\bu)$.  In this section, we consider Edgeworth expansions for the statistic:
$$S_n=n^{1/2}(f(\bu)-f(\bmu))/\sigf,$$ 
where $\sigf^2$ is the asymptotic variance of $S_n$.  The standard Delta Method involves a first-order Taylor expansion, resulting in a Normal approximation with rate $O(1/\sqrt{n})$ when the gradient is not equal to $\boldsymbol{0}$ at $\bmu$.  To attain higher-order correctness, we need to consider a second-order expansion.  Recall the derivatives of interest $a_i, 1\leq i\leq d$ and $a_{ij},1\leq i,j\leq d$ defined in Eq~\ref{eq:ai}.
Furthermore, define the following analog the moments of the linear component of the U-statistic:
\begin{align*}
& \lm_{i_1,\ldots,i_j}=E\left\{\left(\frac{r_{i_1}g_1^{(i_1)}(X_l)}{\rho_n^{s_{i_1}}}\right)\ldots\left(\frac{r_{i_d}g_1^{(i_d)}(X_l)}{\rho_n^{s_{i_d}}}\right)\right\}.
\end{align*}

In the proposition below, we state the form of the Edgeworth for an appropriately smooth function $f$.

\bk
\begin{proposition}\label{prop:smooth-pop}

Suppose that $\sigf>0$, the function $f$ has three continuous derivatives in a neighbourhood of $\bmu$, and $\sum_{i=1}^d a_i g_1^{(i)}(X_l)$ is non-lattice. Then, 
\bk 
\begin{align*}
    P(S_n \leq x) = \Phi(x) + n^{-1/2}p_1(x)\phi(x) + o\left(\frac{1}{n^{1/2}}\right),
\end{align*}
\begin{align*}
    p_1(x)=-\{A_1\sigf^{-1}+\frac{1}{6}A_2\sigf^{-3}(x^2-1)\},
\end{align*}
where $\sigf^2$, $A_1$ and $A_2$ are given by:
\begin{align*}
\sigf^2&=\sum_{i=1}^d\sum_{j=1}^da_ia_j\lm_{ij},\\
A_1&=\frac{1}{2}\sum_{i=1}^d\sum_{j=1}^da_{ij}\lm_{ij},\\
    A_2&=\sum_{i=1}^d\sum_{j=1}^d\sum_{k=1}^d a_i a_ja_k \lm_{ijk}+3\sum_{i=1}^d\sum_{j=1}^d\sum_{k=1}^d\sum_{t=1}^{d}a_ia_ja_{kt}\lm_{ik}\lm_{jt} \\ &+ 3\sum_{i=1}^d\sum_{j=1}^d\sum_{k=1}^d a_ia_ja_k E\left(\frac{r_ig_1^{(i)}(X_{i_1})}{\rho_n^{s_i}}\frac{r_jg_1^{(j)}(X_{i_2})}{\rho_n^{s_j}}\frac{r_k(r_k-1)g_2^{(k)}(X_{i_1},X_{i_2})}{\rho_n^{s_k}}\right) .
\end{align*}
\end{proposition}

Before stating the proof in detail we add an auxiliary lemma, which is needed to the bound the contribution of a remainder term.

\begin{lemma}\label{lem:shift}
Let $S_n=V_n+c/\sqrt{n}$. Let $P(V_n\leq x)=\Phi(x)+p_1(x)\phi(x)/\sqrt{n}+o(1/\sqrt{n})$. Then we have:
\begin{align}\label{eq:taylorcdf}
    P(S_n\leq x)=\Phi(x)+(p_1(x)+c)\phi(x)/\sqrt{n}+o(1/\sqrt{n})
\end{align}
\end{lemma}
\begin{proof}\normalfont
We have:
\begin{align*}
    P(S_n\leq x)&=P(V_n\leq x-c/\sqrt{n})\\
   &= \Phi(x-c/\sqrt{n})+\frac{p_1(x-c/\sqrt{n})}{\sqrt{n}}\phi(x-c/\sqrt{n})+o(1/\sqrt{n})
\end{align*}
Note that $\sup_x|\phi(x)|\leq C$ for some universal constant $C$. So we have:
\begin{align*}
    \Phi(x-c/\sqrt{n})=\Phi(x)-c/\sqrt{n}\phi(x)+O(1/n),
\end{align*}
and 
\begin{align*}
    |\phi(x-c/\sqrt{n})-\phi(x)|=O(1/n).
\end{align*}
    Using the above two equations with Eq~\ref{eq:taylorcdf}, we attain the stated result.
\end{proof}

Now we present the result in Proposition~\ref{prop:smooth-pop}.
\begin{proof}\normalfont
A second-order Taylor expansion yields:
\begin{equation}\label{eq:delta-method}
    n^{1/2}(f(\bu)-f(\bmu))=n^{1/2}<\bu-\bmu,\nabla f(\bmu)> + \frac{1}{2} n^{1/2}(\bu-\bmu)^TH(\bmu)(\bu-\bmu) + O_P\left(\frac{1}{n}\right).
\end{equation}
Furthermore, by a multivariate Hoeffding Decompostion for $\bu-\bmu$:
\begin{align}\label{eq:smf-hd-decomp}
   \bu -\bmu &=\frac{1}{n}\left\{\begin{array}{c}
      \frac{r_1}{\rho_n^{s_1}}\sum_{i=1}^n g_1^{(1)}(X_i)\\ 
           ...\\
    \frac{r_d}{\rho_n^{s_d}}\sum_{i=1}^n g_1^{(d)}(X_i)\\
    \end{array}\right\} 
    +\frac{1}{n^2}\left\{\begin{array}{c}
         \frac{r_1(r_1-1)}{\rho_n^{s_1}}\sum_{i<j} g_2^{(1)}(X_i,X_j) \\
         ...\\
          \frac{r_d(r_d-1)}{\rho_n^{s_d}}\sum_{i<j} g_2^{(d)}(X_i,X_j)
    \end{array}\right\}+ O_P\left(\frac{1}{n^{3/2}}\right) \nonumber \\
    &= \frac{1}{n}\bu_L + \frac{1}{n^2}\bu_Q + O_P\left(\frac{1}{n^{3/2}}\right),
\end{align}
where
\begin{align*}
     n^{-1/2}||\bu_L||=O_P(1), \qquad
     \frac{1}{n^{3/2}}||\bu_Q||=O_P\left(n^{-1/2}\right).
\end{align*}
Now for the first term in Eq~\ref{eq:delta-method}, we have
\begin{align*}
    n^{1/2}<\bu-\bmu, \nabla f(\bmu)> = n^{-1/2}<\bu_L, \nabla f(\bmu)>+ \frac{1}{n^{3/2}}<\bu_Q, \nabla f(\bmu)> + O_P\left(\frac{1}{n}\right).
\end{align*}
The second term in Eq~\ref{eq:delta-method} is 
\begin{align*}
    &n^{1/2}(\bu-\bmu)H(\bmu)(\bu-\bmu)\\&=n^{1/2}\left\{\frac{1}{n}\bu_L +\frac{1}{n^2}\bu_Q + O_P\left(\frac{1}{n^{3/2}}\right)\right\}^T H(\bmu) \left\{\frac{1}{n}\bu_L +\frac{1}{n^2}\bu_Q + O_P\left(\frac{1}{n^{3/2}}\right)\right\}\\
    & = n^{1/2}\left\{\frac{1}{n^2}\bu_L^T H(\bmu) \bu_L + \frac{2}{n^3}\bu_L^T H(\bmu)\bu_Q + O_P\left(\frac{1}{n^{3/2}}\right)\right\} \\
    & = \frac{1}{n^{3/2}}\bu_L^T  H(\bmu) \bu_L + O_P\left(\frac{1}{n}\right)
\end{align*}

Now  Eq~\ref{eq:delta-method} may be expressed as: 
\begin{equation*}
    n^{1/2}(f(\bu)-f(\bmu))=n^{-1/2}<\bu_L,\nabla f(\bmu)> + \frac{1}{n^{3/2}}\left\{<\bu_Q,\nabla f(\bmu)>+\frac{1}{2}\bu_L^T H(\bmu)\bu_L\right\}+O_P\left(\frac{1}{n}\right).
\end{equation*}
We have,
\begin{align}\label{eq:smf-q-decompSn}
    S_n= \frac{A_1}{\sqrt{n}\sigf} + n^{-1/2}\alpha(X_l) +  n^{-3/2}\sum_{l<m} \beta(X_l,X_m) +O_P\left(\frac{1}{n}\right) , 
\end{align}
where 
\begin{align*}
    &\alpha(X_l)= \frac{1}{\sigf}\sum_{i=1}^da_ig_1^{(i)}(X_l)\frac{r_l}{\rho_n^{s_l}},\\
    &\beta(X_l,X_m)= \frac{1}{\sigf}\left\{\sum_{i=1}^d a_i\frac{r_i(r_i-1)}{\rho_n^{s_i}}g_2^{(i)}(X_l,X_m) + \sum_{i, j} a_{ij}\frac{r_ir_j}{\rho_n^{s_i+s_j}}g_1^{(X_i)}(l)g_1^{(j)}(X_m)\right\}
\end{align*}

Applying Theorem 2.1 of \cite{jing2010unified},  under the conditions of proposition~\ref{prop:smooth-pop}, we have
\begin{align*}
    \sup_x\left|P\left(S_n -\frac{A_1}{\sqrt{n}\sigf} \leq x\right)- E_n(x)\right| =  o(n^{-1/2}),
\end{align*}
where 
\begin{align*}
    E_n(x)= \Phi(x) - \frac{(x^2-1)\phi(x)}{6\sqrt{n}}\{E\alpha(X_l)^3+
3E\alpha(X_l)\alpha(X_m)\beta(X_l,X_m)\}.
\end{align*}

Using  Lemma~\ref{lem:shift} \bk and definition of $A_2$, we can  simply $E_n(x)$, yielding: 
\begin{align*}
    &P(S_n\leq x)= \Phi(x) + n^{1/2}\phi(x)p_1(x) + o(n^{-1/2}),\\
    &p_1(x)= -\left\{A_1\sigf^{-1}+\frac{1}{6}A_2\sigf^{-3}(x^2-1)\right\}.
\end{align*}
\end{proof}

\subsection{Proposed bootstrap for smooth functions of counts}\label{sec:supp:sub-proof-smf-bt}
In this section, we consider Edgeworth expansions of smooth functions for the bootstrap.  Recall from Section~\ref{sec:smf} that $\bu^*$ denotes a d-dimensional vector of bootstrapped counted functionals generated by either the multiplier bootstrap $\hat{T}_{n,M}^*$ or the the quadratic bootstrap $\hat{T}_{n,Q}^*$; in the latter case, one may ignore an additional $O_P(n^{-3/2})$ term that arises from approximating a U-statistic by the first two terms of the Hoeffding decomposition.   Now recall the bootstrap analogue $S_n^*$ from Eq~\ref{eq-snstar}, the gradients of the smooth function evaluated at the empirical counts from Eq~\ref{eq-ahat}.

Let $P^*$ denote the bootstrap measure conditioned on $A$ and $X$, with randomness arising from the multiplier weights $\xi$.  Furthermore, let $\hat{P}_n$ denote the the empirical measure $\hat{P}_n = \frac{1}{n}\sum_{i=1}^n \delta_{X_i}$.  It will turn out these two measures are closely related. With a slight abuse of notation, the expectation operator  corresponding to $\hat{P}_n$ will be denoted by $\hat{E}_n f(X) = \frac{1}{n}\sum_{i=1}^n f(X_i)$.

We define the following empirical analogues of the moments of interest:

\begin{align*}
    \tilde{\lm}_{i_1,\ldots,i_j}&=E^*\left\{\left(r_{i_1}\frac{\hat{g}_1^{(i_1)}(l)V_l}{\rho_n^{s_{i_1}}}\right)\ldots \left(r_{i_d}\frac{\hat{g}_1^{(i_d)}(l)V_l}{\rho_n^{s_{i_d}}}\right)\right\}=\frac{1}{n}\sum_{l=1}^n\left(r_{i_1}\frac{\hat{g}_1^{(i_1)}(l)}{\rho_n^{s_{i_1}}}\right)\ldots \left(r_{i_d}\frac{\hat{g}_1^{(i_d)}(l)}{\rho_n^{s_{i_d}}}\right)\\
     &= \hat{E}_n\left\{\left(r_{i_1}\frac{\hat{g}_1^{(i_1)}(l)}{\rho_n^{s_{i_1}}}\right)\ldots \left(r_{i_d}\frac{\hat{g}_1^{(i_d)}(l)}{\rho_n^{s_{i_d}}}\right)\right\}.
\end{align*}

 Now recall that the empirical analogue of the asymptotic variance from Eq~\ref{eq:sigmahatfemp}.
We now prove Theorem~\ref{thm:boot_smooth}, which establishes an Edgeworth expansion for $P^*(S_n^*\leq x)$.

\subsection{Proof of Theorem~\ref{thm:boot_smooth}}
\begin{proof}\normalfont
We will start by establishing Eq \ref{eq:smf-bt-ew}. Let $V_l$ be $\xi_l - 1$ and let $V$ denote the vector:
\begin{align*}
V=(\xi_1-1, \ \ldots, \ \xi_n-1)^T.
\end{align*}

Given $A$ and $X$,  we have
\begin{align*}
    \bu^{*}-\bhatu&=\frac{1}{n}\left\{\begin{array}{c}
         \frac{r_1}{\rho_n^{s_1}}\sum_{l=1}^n \hat{g}_1^{(1)}(l)V_l  \\
         ...\\
           \frac{r_d}{\rho_n^{s_d}}\sum_{l=1}^n \hat{g}_1^{(d)}(l)V_l \\
    \end{array}\right\}
    +\frac{1}{n^2}\left\{\begin{array}{c}
         \frac{r_1(r_1-1)}{\rho_n^{s_1}} \sum_{l<m} \tilde{g}_2^{(1)}(l,m) V_{l}V_{m} \\
         ..\\
          \frac{r_d(r_d-1)}{\rho_n^{s_d}} \sum_{l<m} \tilde{g}_2^{(d)}(l,m)V_{l}V_{m}
    \end{array}\right\}+ O_P\left(\frac{1}{n^{3/2}}\right)\\
    &= \frac{1}{n}\bu_L^* + \frac{1}{n^2}\bu_Q^* + O_P\left(\frac{1}{n^{3/2}}\right).
\end{align*}
Using a second-order Taylor expansion analogous to Eq\ref{eq:delta-method}, we have:
\begin{equation}\label{eq:delta-method-boot}
    n^{1/2}(f(\bu^*)-f(\bhatu))=n^{-1/2}<\bu^*_L,\nabla f(\bhatu)> \hspace{-1mm}+ \frac{1}{n^{3/2}}\left\{<\bu^*_Q,\nabla f(\bhatu)>+\frac{1}{2}\bu_L^{*T} H(\bhatu)\bu_L^*\right\}\hspace{-1mm}+\hspace{-1mm}O_P\left(\frac{1}{n}\right).
\end{equation}


We also have, by definition, 
\begin{align*}
    E^*\{ \hat{g}_1^{(i)}(l)\hat{g}_1^{(j)}(m)\tilde{g}_2^{(k)}(l,m)V_{l}V_{m}\}=\hat{E}\{ \hat{g}_1^{(i)}(l)\hat{g}_1^{(j)}(m)\tilde{g}_2^{(k)}(l,m)\}.
\end{align*}

Then, by Eq~\ref{eq:delta-method-boot} and definition of $\tdsigf$ and $\tilde{A}_1$, we have,
\begin{align*}
    S_n^*= \frac{n^{1/2}(f(\bu^*)-f(\bhatu))}{\tdsigf}= \frac{\tilde{A}_1}{\sqrt{n}\tdsigf} + \frac{1}{B_n}\sum_{l=1}^nb_{n,l}V_l+\frac{1}{n^{3/2}}\sum_{l<m}d_{n,lm}\psi(V_l,V_m)+O_P\left(\frac{1}{n}\right),
\end{align*}
where 
\begin{subequations}
\label{eq:smf-boot-bnl-B-dnlt}
\begin{align}
    &b_{n,l}=\frac{1}{\tdsigf}\sum_{i=1}^d\hat{a}_i\hat{g}_1^{(i)}(l)\frac{r_i}{\rho_n^{s_i}} ,\\
    &B_n^2= \sum_{l=1}^nb_{n,l}^2 = n, \\
    &d_{n,lm}=\frac{1}{\tdsigf}\left\{\sum_{i=1}^d\hat{a}_i\frac{r_i(r_i-1)}{\rho_n^{s_i}}\tilde{g}_2^{(i)}(l,m) + \sum_{i, j} \hat{a}_{ij}\frac{r_ir_j}{\rho_n^{s_i+s_j}}\hat{g}_1^{(i)}(l)\hat{g}_1^{(j)}(m)\right\} ,\\
   &\psi(V_l,V_m)=V_lV_m.
\end{align}
\end{subequations}

Using Lemma~\ref{lem:wangjing} and similar arguments therein if 
\begin{equation}\label{eq:conditions}
     \frac{1}{n}\sum_{l=1}^n b_{n,l} ^2 \geq l_1 >0, \ \ \frac{1}{n}\sum_{l=1}^n |b_{n,l}| ^3 \leq l_2 \leq \infty,
\end{equation}

then 
\begin{align}\label{eq:smf-bt-weight}
    \sup_x|P^*\left(S_n^* - \frac{\tilde{A}_1}{\sqrt{n}\tdsigf} \leq x \right)-\tilde{G}(x)|=  O\left(\frac{ l_{5,n}\log n}{n^{2/3}}\right),
\end{align}

where
and  $\alpha_{l}:=\frac{1}{n}\sum_{m\neq l} d_{n,lm}^2$.
and for sufficiently large $k$:
\begin{align*}
    l_{4,n}=\frac{1}{n}\sum_{l=1}^n\alpha_{l}, \ \ \ \ s_n^2=\frac{1}{n}\sum_l\alpha_l^2-(l_{4,n})^2, \ \ \ \ l_{5,n} = l_{4,n} + k s_n
\end{align*}
and 
\begin{align*}
    &\tilde{G}_n(x)=\Phi(x)+ \tilde{L}_{1,n}(x) +\tilde{L}_{2,n}(x),\\
    & \tilde{L}_{1,n}=\frac{EV_l^3}{6B_n^3}\sum_{l=1}^nb_{n,l}^3(x^2-1)\phi(x),
    \\
    &\tilde{L}_{2,n} = \frac{1}{n^{3/2}B_n^2}\sum_{l<m}b_{n,l}b_{n,m}d_{n,lm}E(V_lV_m\psi(V_l,V_m))(x^2-1)\phi(x).
\end{align*}

Since $\sum_i b_{n,i}^2/n=1$, the first condition in Eq\ref{eq:conditions} is satisfied. Let $c_j:=\frac{r_j}{\rho_n^{s_j}\hatsigf}$. For the second condition, note that since $|.|^3$ is convex,
\begin{align*}
    \frac{1}{n}\sum_i |b_{n,i}|^3\leq \frac{d^2}{\hatsigf^3} \sum_{j=1}^d c_j^3\frac{1}{n}\sum_i |\hat{a}_i\hat{g}_1(i)|^3
\end{align*} 
Since the function $f$ has three gradients in the neighborhood of $\mu$, Lemma~\ref{lem:proveconv} shows that the above is bounded, thereby satisfying the second condition in Eq~\ref{eq:conditions}.

Simplifying $\tilde{L}_{1,n}$ and $\tilde{L}_{2,n}$ using Eq~\ref{eq:smf-boot-bnl-B-dnlt},  we have
\begin{align*}
    \tilde{G}_n(x) = \Phi(x) + n^{-1/2}\phi(x)\frac{1}{6}\tilde{A}_2\tdsigf^{-3}(x^2-1).
\end{align*}

Now we bound the remainder term by bounding $l_{5,n}$. We write $\alpha_l$ as

\begin{align}\label{eq:smf-bt-alphal}
     \alpha_l=\frac{1}{n\tdsigf^2}\sum_{m\neq l}\left\{\underbrace{\sum_{i=1}^d\hat{a}_i\frac{r_i(r_i-1)}{\rho_n^{s_i}}\tilde{g}_2^{(i)}(l,m)}_{Y_{1,lm}} + \underbrace{\sum_{i, j} \hat{a}_{ij}\frac{r_ir_j}{\rho_n^{s_i+s_j}}\hat{g}_1^{(i)}(l)\hat{g}_1^{(j)}(m)}_{Y_{2,lm}}\right\}^2.
\end{align}


Expanding $(Y_{1,lm}+Y_{2,lm})^2$ in Eq~\ref{eq:smf-bt-alphal}, it is straightforward that, by Lemma~\ref{lem:proveconv-mean}, $l_{4,n}$ is $O_P(\rho_n^{-1})$.

Now we bound $s_n$. Since $\alpha_l\geq 0$ and $(Y_{1,lm}+Y_{2,lm})^2 \leq 2(Y_{1,lm}^2+Y_{2,lm}^2)$, we write 
\begin{align*}
    s_n^2&\leq \frac{1}{n}\sum_{l=1}^n\alpha_l^2
    \leq  \frac{4}{n}\sum_{l=1}^n \left(\frac{1}{n\tdsigf^2}\sum_{m\neq l}Y_{1,lm}^2+ \frac{1}{n\tdsigf^2}\sum_{m\neq l}Y_{2,lm}^2\right)^2\\
    &\leq 8\underbrace{\times\frac{1}{n}\sum_{l=1}^n\left(\frac{1}{n\tdsigf^2}\sum_{m\neq l}Y_{1,lm}^2\right)^2}_{Z_1} + 8\times\underbrace{\frac{1}{n}\sum_{l=1}^n\left(\frac{1}{n\tdsigf^2}\sum_{m\neq l}Y_{2,lm}^2\right)^2}_{Z_2}.
\end{align*}

To estimate $Z_1$, we use:
\begin{align*}
    Z_1=\frac{1}{n}\sum_{l=1}^n\left(\frac{1}{n\tdsigf^2}\sum_{m\neq l}\sum_{i=1}^d\sum_{j=1}^d \hat{a}_i\hat{a}_j\frac{r_ir_j(r_i-1)(r_j-1)}{\rho_n^{s_i+s_j}}\tilde{g}_2^{(i)}(l,m)\tilde{g}_2^{(j)}(l,m)\right)^2.
\end{align*}

Using the fact that
\begin{align*}
    \frac{1}{\rho_n^{s_i+s_j}}\tilde{g}_2^{(i)}(l,m)\tilde{g}_2^{(j)}(l,m) \leq \frac{1}{2}\left(\frac{1}{\rho_n^{2s_i}}\tilde{g}_2^{(i)}(l,m)^2+\frac{1}{\rho_n^{2s_j}}\tilde{g}_2^{(j)}(l,m)^2\right),
\end{align*}
by the same arguments in the proof of Theorem~\ref{thm:mbq-firstorder}, it is easy to check that $Z_1= O_P(\rho_n^{-1})$. 

\newcommand{\ghi}{\hat{g}_1^{(i)}(\ell)}
\newcommand{\ghj}{\hat{g}_1^{(j)}(m)}
\newcommand{\ghk}{\hat{g}_1^{(k)}(\ell)}
\newcommand{\ght}{\hat{g}_1^{(t)}(m)}

For $Z_2$, let $c_{ijkt}=a_{ij}a_{kt} r_ir_jr_kr_t/\rho_n^{s_i+s_j+s_k+s_t}$ and $\hat{c}_{ijkt}= \hat{a}_{ij} \hat{a}_{kt} r_ir_jr_kr_t/\rho_n^{s_i+s_j+s_k+s_t}$.  
Consider the estimate:
\begin{align*}
Z_2 \leq &   \underbrace{\frac{2}{n\tdsigf^4}\sum_{l=1}^n \left( \frac{1}{n}\sum_{m\neq \ell}\sum_{i,j,k,t} (\hat{c}_{ijkt} - c_{ijkt}) \ghi\ghj\ghk\ght\right)^2}_{Z_2^{(1)}} 
\\ &+ \underbrace{\frac{2}{n\tdsigf^4}\sum_{l=1}^n \left( \frac{1}{n}\sum_{m\neq \ell}\sum_{i,j,k,t}  c_{ijkt} \ghi\ghj\ghk\ght\right)^2}_{Z_2^{(2)}}       
\end{align*}
We will start by establishing the order of the $Z_2^{(2)}$ term.  Observe that:
\begin{align*}
    &\frac{1}{n}E\sum_\ell\left(\frac{1}{n}\sum_{m\neq \ell}\sum_{i,j,k,t}c_{ijkt} \ghi\ghj\ghk\ght\right)^2\\
    &\leq  \frac{d^4}{n^2}\sum_\ell\sum_{m\neq \ell}\sum_{i,j,k,t}c_{ijkt}^2 E\left[\ghi^2\ghj^2\ghk^2\ght^2\right]\\
    &\leq \frac{d^4}{n^2}\sum_\ell\sum_{m\neq \ell}\sum_{i,j,k,t}c_{ijkt}^2 \left(E\left[\ghi^4\ghj^4\right]E\left[\ghk^4\ght^4\right]\right)^{1/2}\\
    &\leq \frac{d^4}{n^2}\sum_\ell\sum_{m\neq \ell}\sum_{i,j,k,t}c_{ijkt}^2 \left(E\left[\ghi^8\right]E\left[\ghj^8\right]E\left[\ghk^8\right]E\left[\ght^8\right]\right)^{1/4}\\
\end{align*}



 Due to Lemma~\ref{lemma:bounding-eighth-power} and  Eq~\ref{eq:smf-bt-ai}, since $d$ is finite, we see that the above is $O(1)$.  To complete our bound for $Z_2$, observe that:
  \begin{align*}
& P(Z_2^{(2)} > M)
\\ \leq &  P\left( \max_{i,j,k,t} \  (\hat{c}_{ijkt} - c_{ijkt})^2 \  \frac{2}{n\tdsigf^4}\sum_\ell\left(\frac{1}{n}\sum_{m\neq \ell}\sum_{i,j,k,t} \left|\ghi\ghj\ghk\ght \right| \right)^2               > M \right) 
  \end{align*}
 Since $\max_{i,j,k,t} \  (\hat{c}_{ijkt} - c_{ijkt})^2$ is lower-order and the second term in the product inside the probability statement may be viewed as a variant of  $Z_2^{(1)}$ with $c_{ijkt} =1$, we can conclude $Z_2 = O(1)$. Combining $Z_1$ and $Z_2$, we have,
  with probability tending to one,
  $s_n^2 \leq C\rho_n^{-1}$ and 
$l_{5,n}=l_{4,n}+s_n \leq  C'\rho_n^{-1}$ for some universal positive constants $C$ and $C'$. \bk    

Thus, from Eq~\ref{eq:smf-bt-weight} and Lemma~\ref{lem:shift}, we have Eq~\ref{eq:smf-bt-ew}. 

\bk 

\end{proof}

We now state and prove Lemma \ref{lemma:bounding-eighth-power}, which we had used in the proof of the above theorem. 

\begin{lemma}
\label{lemma:bounding-eighth-power}
Under the sparsity conditions in Assumption~\ref{ass:sparse}, $$E(\hat{g}_1(l)^8) = O(\rho_n^{8s})$$.  
\end{lemma}
\begin{proof}\normalfont We decompose $\hat{g}_1(l)$ into 
\begin{align*}
   \hat{g}_1(l)=\hat{H}_1(l) - h_1(l) + g_1(l) - (\hat{T}_n-\theta). 
\end{align*}
Then for some constant $C$,
\begin{align}\label{eq:smf-bt-Eg1hat8}
    \hat{g}_1(l)^8 \leq C\{(\hat{H}_1(l) - h_1(l))^8 +  g_1(l)^8 + (\hat{T}_n-\theta)^8\}.
\end{align}
$g_1(l)^8$ is  $O(\rho_n^{8s})$. Now for $ (\hat{T}_n-\theta)^8$,
\begin{align*}
    (\hat{T}_n-\theta)^8 \leq C\{(\hat{T}_n-T_n)^8 + (T_n-\theta)^8\},
\end{align*}
where $(T_n-\theta)^8= \Theta(\rho_n^{8s})$ by boundness of graphon and we investigate $E\{(\hat{T}_n-T_n)^8\}$. Let $\mathcal{S}_r$ denote all $r$-node subsets from node $\{1,\ldots, n\}$,

\begin{align*}
    E\{(\hat{T}_n-T_n)^8\}&=\frac{\sum\limits_{S_1,\ldots,S_8\in \mathcal{S}_r} E[\{\hat{H}(S_1)-h(S_1)\}\ldots\{\hat{H}(S_8)-h(S_8)\}]}{{n \choose r}^8}.
\end{align*}

Consider any term in the above sum where each of the four pairs of the subsets have $p_i,i=1,\dots,4$ nodes, $d_i,i=1,\dots,4$ edges in common. In this case there are $8r-\sum_i  p_i$ choices of nodes and the number of edges are at least $8s-\sum_i d_i$. First note that $p_i\geq 2$, to have non-zero contribution. For acyclic graphs, $d_i\leq p_i-1$ and for general subgraphs with a cycle, $d_i\leq {p_i\choose 2}$. Thus, for $p_i\geq 2$, we have:
\begin{align*}
\frac{O\left(n^{8r-\sum_i p_i}\rho_n^{8s-\sum_i d_i}\right)}{{n\choose r}^8}=O(\rho_n^{8s})\times O\left(\frac{1}{n^{\sum_ip_i} \rho_n^{\sum_i d_i}}\right).
\end{align*}
For acyclic graphs, it is easy to see that under our sparsity conditions the above is dominated by $p=2$. For general cyclic graphs, since $\rho_n=\omega(n^{-1/r})$ and $p\leq r$,
\begin{align*}
n^{p_i}\rho_n^{d_i}\geq n^{p_i\left(1-\frac{p_i-1}{2r}\right)}\geq n^{\frac{4p_i(r+1)}{r}} \rightarrow \infty.
\end{align*}
So, $ E\{(\hat{T}_n-T_n)^8\}=O(\rho_n^{8s})$.

To finish bounding $E[\hat{g}_1(l)^8]$, we look into the first term of Eq~\ref{eq:smf-bt-Eg1hat8}. Let $\mathcal{S}_r^{l}$ denote all $r-1$ node subsets from node $\{1,\ldots, n\}$ excluding node $l$,
\begin{align*}
    E\{(\hat{H}_1(l)-h_1(l))^8\}&= \frac{\sum\limits_{S_1,\ldots,S_8\in \mathcal{S}_r^{l}}E[\{\hat{H}(l\cup S_1)-h(l \cup S_1)\}\ldots\{\hat{H}(l \cup S_8)-h(l \cup S_8)\}]}{{n-1 \choose r-1}^8}.
\end{align*}
Similarly, consider any term in the above sum where each of the four pairs of the subsets have $p_i,i\leq 4$ nodes (besides node $l$), $d_i,i\leq 4$ edges in common. In this case there are $4(2r-2)-\sum_i p_i$ choices of nodes and the number of edges are $8s-\sum_i d_i$. (since each subset already share node $l$). When $p_i\geq 1$, each pair share node $l$ and another $p_i$ nodes, then for acyclic graphs, $d_i\leq p_i$, and for general subgraphs with a cycle, $d_i\leq {p_i+1 \choose 2}$. Thus, for $p_i \geq 1$, $d_i \geq 0$, we have
\begin{align*}
\frac{O\left(n^{4(2r-2)-\sum_i p_i}\rho_n^{8s-\sum_i d_i}\right)}{{n-1\choose r-1}^8}=O(\rho_n^{8s})\times O\left(\frac{1}{n^{\sum_i p_i} \rho_n^{\sum_i d_i}}\right),
\end{align*}
where as we showed above for acyclic graphs, under our sparsity conditions the above is dominated by $p=1$. For general cyclic graphs, since $\rho_n\gg n^{1/r}$ and $p_i\leq r$, $
n^{p_i}\rho_n^{d_i} \rightarrow \infty$. Thus, $E\{(\hat{H}_1(l)-h_1(l))^8\}$ is also $O(\rho_n^{8s})$.
\bk

Thus, combining all terms in Eq~\ref{eq:smf-bt-Eg1hat8}, $E[\hat{g}_1(l)^8]$ is $O(\rho_n^{8s})$.
\end{proof}

\bk

\subsection{Comparing bootstrap Edgeworth expansion with the U-statistic Edgeworth expansion}\label{sec:supp:sub-proof-smf-bt-close-to-ew}

Finally, we show that the bootstrap Edgeworth expansion is close to that of the conditional expectation, which was established in Proposition~\ref{prop:smooth-pop}.

\begin{proposition} \label{prop:smf-bt-ew-close}
Suppose that $\sigf>0$, the function $f$ has three continuous derivatives in a neighbourhood of $\bmu$, and $\sum_{i=1}^d a_i g_1^{(i)}(X_l)$ is non-lattice.  Furthermore, suppose that the weights $\w_1,\dots,\w_n$ are generated from a non-lattice distribution such that $\E(\xi_1) = 1$, $\E\{(\xi_1-1)^2\} = 1$,  $\E\{(\xi_1-1)^3\} = 1$. 
 Then we have:
 \begin{align}\label{eq:smf-bt-ew-close}
       P^*(S_n^*\leq x) 
    =P(S_n \leq x)+ o_P\left(n^{-1/2}\right )+O_P\left(\frac{\log n}{n^{2/3}\rho_n}\right).
 \end{align}
\bk
\end{proposition}

\begin{proof}\normalfont
	Now we show that $\tdsigf$, $\tilde{A}_1$ and $\tilde{A}_2$ converge to $\sigf$, $A_1$ and $A_2$. We first show $\tilde\lm_{ij}$ and $\tilde\lm_{ijk}$ converge to $\lm_{ij}$ and $\lm_{ijk}$.
	\begin{align*}
	\tilde\lm_{ij}&=r_ir_j\hat{E}\left\{\frac{\hat{g}_1^{(i)}(l)\hat{g}_1^{(j)}(l)}{\rho_n^{s_i}\rho_n^{s_j}}\right\},\\
	\lm_{ij}&=r_ir_j\E\left\{\frac{g_1^{(i)}(X_l)g_1^{(j)}(X_l)}{\rho_n^{s_i}\rho_n^{s_j}}\right\}.
	\end{align*}
	Using the fact that $\E(V_l)=0$,$\E(V_l)^2=1$, $\E\{g_1^{(i)}\}=0$, and an analogous argument as in the proof of Lemma 3.1d) in~\cite{zhang-xia-network-edgeworth}, we have: 
	\begin{align*}
	\tilde\lm_{ij}-\lm_{ij} = O_P\left(n^{-1/2}\rho_n^{-1/2}\right).
	\end{align*}
	Similarly, expanding $\tilde\lm_{ijk}$ and $\lm_{ijk}$, using the fact that $\E(V_l^3)=1$, $\E\{g_1^{(i)}\}=0$, 
	\begin{align*}
	\tilde\lm_{ijk}-\lm_{ijk}= O_P\left(n^{-1/2}\rho_n^{-1/2}\right).
	\end{align*}
	
	Using the same argument in the proof of Lemma~\ref{lem:convEW}, we have
	\begin{align*}
	\hat{E}\{ \hat{g}_1^{(i)}(l)\hat{g}_1^{(j)}(m)\tilde{g}_2^{(k)}(l,m)\} - E\{g_1^{(i)}(X_{l})g_1^{(j)}(X_{m})g_2^{(k)}(X_{l},X_{m})\}= O_p(n^{-1/2}\rho_n^{-1/2}).
	\end{align*}
	Furthermore, under the assumption that $f$ has three continuous derivatives in the neighbourhood of $\bmu$, we know that 
	\begin{align}\label{eq:smf-bt-ai}
	\hat{a}_i= a_i + O_P\left(n^{-1/2}\rho_n^{-1/2}\right), \ \hat{a}_{ij}= a_{ij} + O_P\left(n^{-1/2}\rho_n^{-1/2}\right).
	\end{align}
	Thus, together with Eq~\ref{eq:smf-bt-ai}, we have,
	\begin{align*}
	\tdsigf^{2}-\sigf^2= O_P\left(n^{-1/2}\rho_n^{-1/2}\right), \ \tilde A_1-A_1=O_P\left(n^{-1/2}\rho_n^{-1/2}\right), \ \tilde A_2-A_2= O_P\left(n^{-1/2}\rho_n^{-1/2}\right).
	\end{align*}
	Finally we have,
	\begin{align*}
	\tilde p_1(x)=-\{\tilde{A}_1\tdsigf^{-1}+\frac{1}{6}\tilde{A}_2\tdsigf^{-3}(x^2-1)\}=p_1(x)+O_P\left(n^{-1/2}\rho_n^{-1/2}\right).
	\end{align*}
	Therefore, under the same condition of Proposition~\ref{prop:smooth-pop}, from Eq~\ref{eq:smf-bt-ew}, we have, 
	\begin{align*}
	P^*(S_n^*\leq x)&= \Phi(x) + n^{-1/2}\tilde{p}_1(x)\phi(x) + O_P\left(\frac{\log n}{n^{2/3}\rho_n}\right)\\ &
	=P(S_n \leq x) 
	+o_P\left(n^{-1/2}\right)+O_P\left(\frac{\log n}{n^{2/3}\rho_n}\right).
	\end{align*}
	\end{proof}
\bk

\section{Detailed results of confidence interval bias correction for smooth functions of counts }\label{sec:supp:proof-smf-ew-stu-and-ci}
\newcommand{\crm}{\frac{r_m}{\rho_n^{s_m}}}
\newcommand{\gom}{g_1^{(m)}(X_i)}
\subsection{Edgeworth expansion for studentized smooth function of counts}
\label{sec:supp:sub-proof-smf-ew-stu}


In order to write $\sigf^2$ as a function of $\bmu$ and $\hatsigf^2$ as function of $\bu$, we have to expand the vector of $\bu$ by including terms such that the variance can be written as a function of the expectation. For example, for simple mean, one needs to add $(x_1,x_2)=(x,x^2)$ for data point $x$, since the variance is then $x_2-x_1^2$. For i.i.d random variables, this is simple, but for U statistics, the dependence makes this more nuanced. We expand the vector of $\bu$ into $\cbu$.  Given $X$, the uncentered $\cbu$ is
\begin{equation}\label{eq:longu}
\begin{split}
   \cbu=\bigg{\{}\underbrace{\frac{\tnhat^{(1)}}{\rho_n^{s_1}},\ldots, \frac{\tnhat^{(d)}}{\rho_n^{s_d}}}_{d \text{ terms}},\underbrace{\frac{r_1r_2\sum_{i=1}^n\hat{h}_1^{(1)}(X_i)\hat{h}_1^{(2)}(X_i)}{n\rho_n^{s_1}\rho_n^{s_2}},\ldots,\frac{r_{d-1}r_d\sum_{i=1}^n\hat{h}_1^{(d-1)}(X_i)\hat{h}_1^{(d)}(X_i)}{n\rho_n^{s_1}\rho_n^{s_2}}}_{{d \choose 2} \text{ terms}},\\
   \underbrace{\frac{r_1^2\sum_{i=1}^n\hat{h}_1^{(1)}(X_i)^2}{n\rho_n^{2s_1}},\ldots,\frac{r_d^2\sum_{i=1}^n\hat{h}_1^{(d)}(X_i)^2}{n\rho_n^{2s_d}}}_{d \text{ terms}}\bigg{\}}, 
\end{split}
\end{equation}
where 
\begin{align*}
    \hat{h}_1(X_i)=\frac{1}{{n-1 \choose r-1}}\sum_{1\leq i_1<\ldots< i_r\leq n, i_1,\ldots,i_r \neq i } \hat{h}(X_i,X_{i_1},\ldots,X_{i_r}).
\end{align*}

Denote $\cbmu=E\cbu$, and $\pbu=\cbu-\cbmu$. Define $h(\bmu)={\sigf}^2$,  $h(\cbu)=\hatsigf^2$ and $c_i=\nabla h(\cbmu)^{(i)}$.

\begin{proposition}\label{prop:smooth-pop-q}
Define $S'_n=n^{1/2}(f(\bu)-f(\bmu))/\hatsigf$.
Under the condition that the function $f$ has three continuous derivatives in a neighbourhood of $\bmu$, and $\sum_{i=1}^d a_i g_1^{(i)}(X_1)$, is non-lattice, we have:
\begin{align*}
    P(S'_n \leq x) = \Phi(x) + n^{-1/2}q_1(x)\phi(x) + o\left(\frac{1}{n^{1/2}}\right),
\end{align*}
\begin{align*}
    q_1(x)=-\{B_1+\frac{1}{6}B_2(x^2-1)\},
\end{align*}
where $B_1$ and $B_2$ are
\begin{align*}
B_1&=  A_1\sigf^{-1}-\frac{1}{2}\sigf^{-3}n\sum_{i=1}^{d'}\sum_{j=1}^{d'}a_ic_j\E\{\pbu^{(i)}\pbu^{(j)}\} ,\\
B_2&= 6B_1-6A_1 + \frac{A_2}{\sigf^3}.
\end{align*}
$A_1$ and $A_2$ are defined in Proposition ~\ref{prop:smooth-pop}. The regularity conditions are to ensure the remainders in the stated order uniformly in $x$.
\end{proposition}

\begin{proof}\normalfont

Now, we define $A(\cbu)=A(\bu)=f(\bu)-f(\bmu)$, $B(\cbu)=(f(\bu)-f(\bmu))/h(\cbu)$. Then by Taylor Expansion we have,
\begin{align*}
    B(\cbu)&=A(\cbu)/h(\cbu)^{1/2}=A(\cbu)*h(\cbu)^{-1/2}\\
    &=A(\cbu)\hspace{-1mm}\left\{h(\cbmu)^{-1/2}+(\cbu-\cbmu)^T\nabla (h(\cbmu)^{-1/2})+(\cbu-\cbmu)^T\frac{H(h(\cbmu)^{-1/2})}{2}(\cbu-\cbmu)+o_P\left(\frac{1}{n}\right)\hspace{-1mm}\right\}\\
    &= A(\cbu)/h(\cbmu)^{-1/2} - \frac{1}{2}A(\cbu)h(\cbmu)^{-3/2}(\nabla h(\cbmu))^T(\cbu-\cbmu) + O_P\left(\frac{1}{n^{3/2}}\right)\\
    &= A(\cbu)/\sigf -\frac{1}{2}(\cbu-\cbmu)^T \sigf^{-3}\underbrace{\nabla f(\cbmu)(\nabla h(\cbmu))^T}_{D}(\cbu-\cbmu)+O_P\left(\frac{1}{n^{3/2}}\right)\\
    &= A(\bu)/\sigf - \frac{1}{2} (\cbu-\cbmu)^T\sigf^{-3}\underbrace{\left[\begin{array}{ccc}
        a_1c_1 &... &a_1c_d  \\
         ... &... &...\\
         a_dc_1 &... &a_dc_d
    \end{array}\right]}_{D}(\cbu-\cbmu)+O_P\left(\frac{1}{n^{3/2}}\right),
\end{align*}
where 
\begin{align*}
    &(\cbu-\cbmu)^{T}D(\cbu-\cbmu) = \sum_{i=1}^{d'} \sum_{i=1}^{d'}a_ic_j\pbu^{(i)}\pbu^{(j)},\\
    & a_{d+1}=a_{d+2}=\ldots=a_{d'}=0.
\end{align*}

We have $S_n'=n^{1/2}\frac{f(\bu)-f(\bmu)}{h(\cbu)}=n^{1/2}B(\cbu)$. Thus we can write $S_n'$ into
\begin{align}\label{eq:smf-q-snprime}
    S_n'=n^{1/2}A(\bu)/\sigf -\frac{1}{2}\sigf^{-3}n^{1/2}\sum_{i=1}^{d'} \sum_{i=1}^{d'}a_ic_j\pbu^{(i)}\pbu^{(j)} + O_P\left(\frac{1}{n}\right).
\end{align}

Since $a_i=0$ for $i>d$, we only discuss here $\pbu^{(i)}\pbu^{(j)}$ for $i \leq d, j \leq d$ and $\pbu^{(i)}\pbu^{(j)}$ for $i \leq d, d < j \leq d'$. We first prove that 
\begin{align}\label{eq:smf-q-uiuj}
    \pbu^{(i)}\pbu^{(j)} =\E\{\pbu^{(i)}\pbu^{(j)}\} + n^{-2} \sum_{l<m}\gamma(X_l,X_m) +O_P\left(\frac{1}{n^{3/2}}\right)
\end{align}
holds for both cases, where $\gamma$ is some symmetric function of $X_l$ and $X_m$.

For $i \leq d, j \leq d$,  since $ \pbu^{(i)} =\frac{\bu_L^{(i)}}{n} + \frac{\bu_Q^{(i)}}{n^2} + O_P\left(\frac{1}{n^{3/2}}\right)$, $ \pbu^{(j)} =\frac{\bu_L^{(j)}}{n} + \frac{\bu_Q^{(j)}}{n^2} + O_P\left(\frac{1}{n^{3/2}}\right)$,
\begin{align*}
     \pbu^{(i)}\pbu^{(j)}&= \frac{\bu_L^{(i)}\bu_L^{(j)}}{n^2} + O_P\left(\frac{1}{n^{3/2}}\right)\\& = \frac{r_ir_j}{n^2\rho_n^{s_i}\rho_n^{s_j}}\sum_{l=1}^ng_1^{(i)}(X_l)g_1^{(j)}(X_l) + \frac{2r_ir_j}{n^2\rho_n^{s_i}\rho_n^{s_j}}\sum_{l<m}g_1^{(i)}(X_l)g_1^{(j)}(X_m) + O_P\left(\frac{1}{n^{3/2}}\right).
\end{align*}
Thus,
\begin{align*}
     E\pbu^{(i)}\pbu^{(j)}=\frac{r_ir_j}{n^2\rho_n^{s_i}\rho_n^{s_j}}\sum_{l=1}^ng_1^{(i)}(X_l)g_1^{(j)}(X_l),
\end{align*}
and Eq~\ref{eq:smf-q-uiuj} follows.

For $i \leq d,  d < j \leq d'$, $ \pbu^{(i)} =\frac{\bu_L^{(i)}}{n} + \frac{\bu_Q^{(i)}}{n^2} + O_P\left(\frac{1}{n^{3/2}}\right)$, while $$ \pbu^{(j)} =\frac{r_kr_t}{n\rho_n^{s_k}\rho_n{s_t}}\sum_{l=1}^n\hat{h}_1^{(k)}(X_l)\hat{h}_1^{(t)}(X_l)- E\left\{\frac{r_kr_t}{\rho_n^{s_k}\rho_n^{s_t}n}\sum_{l=1}^n\hat{h}_1^{(t)}(X_l)\hat{h}_1^{(k)}(X_l)\right\},$$ for some $k, t \in \{1,\ldots,d\}$. Denote $\E\{\hat{h}^{k}(X_l)\}=\theta^{(k)}$, Hoeffding decomposition of $\hat{h}^{k}(X_l)$ yields, 
\begin{align}\label{eq:smf-q-h1-hoef}
    \frac{\hat{h}^{(k)}_1(X_l)-\theta_n^{(k)}}{\rho_n^{s_k}} &= h^{(k)}_1(X_l)-\theta_n^{(k)} + \frac{r-1}{n-1} \sum_{s \neq l,1\leq s \leq n} \{g^{(k)}_2(X_l,X_s) + g_1^{(k)}(X_s)\}+  O_P\left(\frac{1}{n}\right)\bk \nonumber \\
    &= g^{(k)}_1(X_1) +\frac{r-1}{n-1} \sum_{s \neq l,1\leq s \leq n} \{g^{(k)}_2(X_l,X_s) + g_1^{(k)}(X_s)\}  +  O_P\left(\frac{1}{n}\right)\bk.
\end{align}

Denote $\bm{U}^{(i)}= \frac{\sum_{l=1}^n g_1^{(i)}(X_l)}{n}$, then
\begin{align*}
     \pbu^{(i)}\pbu^{(j)}& = \frac{r_kr_ir_t}{\rho_n^{s_k+s_i+s_t}}\bigg{\{}\theta^{(k)}\bm{U}^{(i)}\bm{U}^{(t)} + \theta^{(t)}\bm{U}^{(i)}\bm{U}^{(k)}  
     +\frac{1}{n^2}\sum_{l=1}^ng_1^{(i)}(X_l)g_1^{(k)}(X_l)g_1^{(t)}(X_l) \\ & + \frac{2}{n^2}\sum_{l<m}g_1^{(i)}(X_l)g_1^{(k)}(X_m)g_1^{(t)}(X_m) +
     \frac{2}{n^2}\sum_{l<m}g_1^{(i)}(X_l)g_1^{(k)}(X_m)g_2^{(t)}(X_l,X_m)   \\
     &+
     \frac{2}{n^2}\sum_{l<m}g_1^{(i)}(X_l)g_1^{(t)}(X_m)g_2^{(k)}(X_l,X_m)  \bigg{\}} +
     O_P\left(\frac{1}{n^{3/2}}\right).
\end{align*}


Taking Expectation, Eq~\ref{eq:smf-q-uiuj} easily follows.

Now that Eq~\ref{eq:smf-q-uiuj} holds, using Eq~\ref{eq:smf-q-decompSn} and Eq~\ref{eq:smf-q-snprime}, we have
\begin{align}\label{eq:smf-q-snprime-decomp}
    S_n'&=  \frac{A_1}{\sqrt{n}\sigf} -\frac{1}{2}\sigf^{-3}n\sum_{i=1}^{d'}\sum_{j=1}^{d'}a_ic_j\E\{\pbu^{(i)}\pbu^{(j)}\} \nonumber \\
    &+  n^{-1/2}\sum_{l=1}^n\alpha(X_l) +  n^{-3/2}\sum_{l<m} \{\beta(X_l,X_m)+\gamma(X_l,X_m)\} +O_P\left(\frac{1}{n}\right) 
\end{align}

Therefore, using Lemma~\ref{lem:shift}, we know that \begin{equation}\label{eq:smf-q-b1}
    B_1=  A_1\sigf^{-1}-\frac{1}{2}\sigf^{-3}n\sum_{i=1}^{d'}\sum_{j=1}^{d'}a_ic_j\E\{\pbu^{(i)}\pbu^{(j)}\}.
\end{equation}

From Eq~\ref{eq:smf-q-snprime-decomp}, we also have that Theorem 2.1 of \cite{jing2010unified} applies to $S_n'$ under the same conditions of Proposition~\ref{prop:smooth-pop}. For the simplicity of calculation, we note that $B_2$ can be estimated using the identity $p_1(0)=q_1(0)$ and the forms of $A_1,A_2,$ and $B_1$, which gives us $B_2= 6B_1-6A_1 + \frac{A_2}{\sigf^3}$.

Thus, under same conditions of Proposition~\ref{prop:smooth-pop}, we have
\begin{align*}
    P(S_n' \leq x) = \Phi(x) + n^{-1/2}q_1(x)\phi(x) + o\left(n^{-1/2}\right),
\end{align*}
where $B_1$ and $B_2$ are defined above, $q_1(x)$ is as 
\begin{align*}
    q_1(x)=-\{B_1+\frac{1}{6}B_2(x^2-1)\}.
\end{align*}

\end{proof}

\subsection{Estimating confidence interval corretion for smooth function of counts}\label{sec:supp:sub-ci-correction}


In order correct the confidence intervals arising from the standardized bootstrap, we need to estimate $p_1(x)$ and $q_1(x)$. This requires the calculation of  $\sigf^2$, $A_1$ and $A_2$ are straightforward. In this section, we show how to compute $\hat{q}_1(x)$ for transitivity. 

While we only show in detail the calculations of transitivity ($d=2$), they can be easily used as building blocks to extend to other smooth functions of counts with $d \geq 2$.

In the case of transitivity, the original $\bu$ used for estimation of $p_1(x)$ is of length $d=2$. Recall that for estimating $q_1(x)$ we need to expand this vector so that the variance is a function of this vector. This expanded vector (see Eq~\ref{eq:longu}) is of length $d'=5$. 
Denote $\mu_{ij}=n\times E\pbu^{(i)}\pbu^{(j)}$. We also have for $T$ and $V$, $r_1=r_2=r=3$, and $s_1=3$ and $s_2=2$.

To estimate $B_1$ and $B_2$, we first use the fact that $c_k$ for k in $1\leq k \leq d'$ follows  \cite{hall-bootstrap-edgeworth} Section 3.10.6 as follows:
\begin{align}\label{eq:smf-c}
    c_k = 2\sum_{i=1}^d\sum_{j=1}^da_{ik}a_j\mu_{ij} - 2a_k\sum_{i=1}^d a_i\bmu^{(i)}+ \sum_{i=1,j=1,(k)}^{d,d} a_i a_j,
\end{align}
where $\sum^{d,d}_{i=1,j=1,(k)}$ denotes the pair $(i,j)$ in $\cbu^{(1),\ldots,(d)}$ that $\cbu^{(i)}\cbu^{(j)}=\cbu^{(k)}$. For example, in transitivity, $\cbu^{(3)}=\cbu^{(1)}\cbu^{(2)}$.

Now we simplify $E\pbu^{(i)}\pbu^{(j)}$ for $1\leq i \leq d$, $1\leq j \leq d'$ for the purpose of estimating $B_1$ and $c_k$ in Eq~\ref{eq:smf-q-b1} and Eq~\ref{eq:smf-c}. We do not consider the case where $i>d$ since $a_i$ for $i>d$ is $0$.   By the definition of $\bu'$, using Hoeffding Decompostion of $\hat{h}^{(i)}_1(X_1)$ ($i \in \{1,2\})$ showed in Eq~\ref{eq:smf-q-h1-hoef}, 
simple algebra yields, 
\begin{align*}
    &\E\{\pbu^{(i)}\pbu^{(j)}\}=
    \frac{r^2}{n\rho_n^{s_1}\rho_n^{s_2}} \E\{g_1^{(i)}(X_1)g_1^{(j)}(X_1)\} +  O\left(\frac{1}{n^{3/2}}\right), \bk \ 1\leq i, j \leq 2,
\end{align*}
\vspace{-3mm}
\begin{align*}
   \E\{\pbu^{(1)}\pbu^{(4)}\}=&\frac{r^3}{n\rho_n^{s_1^3}} \bigg{[}\E\{g_1^{(1)}(X_1)g_1^{(1)}(X_1)^2\}\\
   &+ 2(r-1)\E\{g_1^{(1)}(X_1)g_1^{(1)}(X_2)g_2^{(1)}(X_1,X_2)\} \bigg{]} + \frac{2r^4}{n\rho_n^{2s_1}}\bmu_1E\{g_1^{(1)}(X_1)g_1^{(1)}X_1\} \\
   &+  O\left(\frac{1}{n^{3/2}}\right) \bk, 
\end{align*}
\begin{align*}
   \E\{\pbu^{(2)}\pbu^{(5)}\}=&\frac{r^3}{n\rho_n^{3s_2}} \bigg{[}\E\{g_1^{(2)}(X_1)g_1^{(2)}(X_1)^2\}\\
   &+ 2(r-1)\E\{g_1^{(2)}(X_1)g_1^{(2)}(X_2)g_2^{(2)}(X_1,X_2)\} \bigg{]} + \frac{2r^4}{n\rho_n^{2s_2}}\bmu_2E\{g_1^{(1)}(X_1)g_1^{(2)}X_1\} \\
   &+  O\left(\frac{1}{n^{3/2}}\right)\bk,
\end{align*}

\begin{align*}
   \E\{\pbu^{(1)}\pbu^{(5)}\}=&\frac{r^3}{n\rho_n^{s_1}\rho_n^{2s_2}} \bigg{[}\E\{g_1^{(1)}(X_1)g_1^{(2)}(X_1)^2\}\\
   &+ 2(r-1)\E\{g_1^{(1)}(X_1)g_1^{(2)}(X_2)g_2^{(2)}(X_1,X_2)\} \bigg{]} + \frac{2r^4}{n\rho_n^{s_1}\rho_n^{s_2}}\bmu_2E\{g_1^{(1)}(X_1)g_1^{(2)}X_1\} \\
   &+  O\left(\frac{1}{n^{3/2}}\right)\bk, 
\end{align*}
\begin{align*}
   \E\{\pbu^{(2)}\pbu^{(4)}\}=&\frac{r^3}{n\rho_n^{2s_1}\rho_n^{s_2}} \bigg{[}\E\{g_1^{(2)}(X_1)g_1^{(1)}(X_1)^2\}\\
   &+ 2(r-1)\E\{g_1^{(2)}(X_1)g_1^{(1)}(X_2)g_2^{(1)}(X_1,X_2)\} \bigg{]} + \frac{2r^4}{n\rho_n^{s_1}\rho_n^{s_2}}\bmu_1E\{g_1^{(1)}(X_1)g_1^{(2)}X_1\} \\
   &+ O\left(\frac{1}{n^{3/2}}\right)\bk, 
\end{align*}

For the other two cases of $i=1,2$ and $j=3$, applying the same technique, we have
\begin{align*}
   \E\{\pbu^{(1)}\pbu^{(3)}\}=&\frac{r^3}{n\rho_n^{2s_1}\rho_n^{s_2}} \bigg{[}\E\{g_1^{(1)}(X_1)^2g_1^{(2)}(X_1)\}\\
   &+ (r-1)\E\{g_1^{(1)}(X_1)g_1^{(1)}(X_2)g_2^{(2)}(X_1,X_2)\} \\ &+ (r-1)\E\{g_1^{(1)}(X_1)g_1^{(2)}(X_2)g_2^{(1)}(X_1,X_2)  \bigg{]} + \frac{r^4}{n\rho_n^{s_1}\rho_n^{s_2}}\bmu_1E\{g_1^{(1)}(X_1)g_1^{(2)}X_1\} \\
   &+ \frac{r^4}{n\rho_n^{2s_1}}\bmu_2E\{g_1^{(1)}(X_1)^2\} +  O\left(\frac{1}{n^{3/2}}\right)\bk, 
\end{align*}

\begin{align*}
   \E\{\pbu^{(2)}\pbu^{(3)}\}=&\frac{r^3}{n\rho_n^{s_1}\rho_n^{2s_2}} \bigg{[}\E\{g_1^{(2)}(X_1)^2g_1^{(1)}(X_1)\}\\
   &+ (r-1)\E\{g_1^{(2)}(X_1)g_1^{(1)}(X_2)g_2^{(2)}(X_1,X_2)\} \\ &+ (r-1)\E\{g_1^{(2)}(X_1)g_1^{(2)}(X_2)g_2^{(1)}(X_1,X_2)  \bigg{]} + \frac{r^4}{n\rho_n^{s_1}\rho_n^{s_2}}\bmu_2E\{g_1^{(1)}(X_1)g_1^{(2)}X_1\} \\
   &+ \frac{r^4}{n\rho_n^{2s_1}}\bmu_1E\{g_1^{(2)}(X_1)^2\} +  O\left(\frac{1}{n^{3/2}}\right)\bk.
\end{align*}

Now we can estimate $B_1$ from Eq~\ref{eq:smf-q-b1} and $c_{i},1\leq i\leq 5$ from Eq~\ref{eq:smf-c} by estimating $E\{\pbu^{(i)}\pbu^{(j)}\}$ above using $\hat{g}^{(i)}_1(X_1)$ and $\hat{g}^{(i)}_2(X_1,X_2)$, for $i\in\{1,2\}$.  Using the fact of $p_1(0)=q_1(0)$, we can estimate $B_2$ by 
\begin{align*}
    \hat{B}_2= 6\hat{B}_1 -6\hat{A}_1\hatsigf^{-1} +  \hat{A}_2 \hatsigf^{-3}.
\end{align*}
Then we have the estimated $\hat{q}_1(x)$,
\begin{align*}
    \hat{q}_1(x)=-\{\hat{B}_1+\frac{1}{6}\hat{B}_2(x^2-1)\}.
\end{align*}

Now we show the studentized edgeworth expansion of some statistics $f(T,V)$ using same $\cbu$ as transitivity, including $T$, $3T+5V$, $TV$, $3T/V$(transitivity) and $T^2V^2$. The $q_1(x)$ of the Edgeworth expansion of the studentized version of these statistics $f(T,V)$ share the same $E\{\pbu^{(i)}\pbu^{(j)}\}$ ($i \in \{1,2\}$, j in \{1,\ldots,5\}. The the only difference lies in evaluating different derivatives $\bold{a}$ of $f$ and thus having different $\hat{c}_k$, $\hat{B}_1$ and $\hat{B}_2$.  

Recall that $\text{err}(F,G)$ is defined in Section~\ref{sec:exp} as the maximum of $|F(x)-G(x)|$ over the range $[-3,3]$, over a grid size $0.1$.  In the following two tables,we show this distance between the true CDF and our empirical edgeworth expansion and the normal approximation for five different smooth functions. Tables~\ref{tb:smf-p-superror} and~\ref{tb:smf-q-superror} show  these for the standardized and studentized statistics. The empirical edgeworth expansion is estimated using a random graph with $n=160$, $\rho_n=1$, generated from two graphons \sbm and \smg with the same parameters as in Section~\ref{sec:exp}. The true CDF is estimated by $10^6$ size $160$ graphs generated by the same graphons with same model parameters.

\begin{table}[H]
	\centering
	{\setlength{\extrarowheight}{4pt}
	\begin{tabular}{|l|l|l|l|l|}
		\hline
		& \multicolumn{2}{l|}{SBM}                            & \multicolumn{2}{l|}{SM-G}                           \\ \hline Studentized
		& \small{$\text{err}(\hat{F}(S_n),F)$} & \small{$\text{err}(\Phi,F)$} & \small{$\text{err}(\hat{F}(S_n),F)$} & \small{$\text{err}(\Phi,F)$} \\ \hline
		T                                        & \textbf{0.002}                           & 0.018             & \textbf{0.004}                           & 0.030             \\ \hline
		3T+5V                                    & \textbf{0.003}                           & 0.011             & \textbf{0.002}                           & 0.018             \\ \hline
		TV                                       & \textbf{0.006}                           & 0.023             & \textbf{0.016}                           & 0.042             \\ \hline
		3T/V                                     & \textbf{0.005}                           & 0.027             & \textbf{0.006}                           & 0.051             \\ \hline
		$T^2V^2$ & \textbf{0.036}                           & 0.078             & \textbf{0.092}                           & 0.142            \\ \hline
	\end{tabular}}
	\caption{Standardized EW Sup CDF error compared to N(0,1)}
	\label{tb:smf-p-superror}
\end{table}
Table~\ref{tb:smf-p-superror} shows the standardized sup error $\sup_x |\hat{F}(S_n\leq x) - F(x)|$, where $S_n$ is the standardized statistic, $\hat{F}(S_n\leq x)=\Phi(x) + n^{-1/2}\hat{p}_1(x)\phi(x)$ and $F(x)$ is the true distribution of the standardized statistic.  In Table~\ref{tb:smf-q-superror}, we show $\sup_x |\hat{F}(S_n\leq x) - F'(x)|$, where $\hat{F}(S_n'\leq x)=\Phi(x) + n^{-1/2}\hat{q}_1(x)\phi(x)$ and $F'(x)$ indicates the true CDF of the studentized statistic. 

\begin{table}[H]
	\centering
		{\setlength{\extrarowheight}{4pt}
\begin{tabular}{|l|l|l|l|l|}
\hline
                                         & \multicolumn{2}{l|}{SBM}                            & \multicolumn{2}{l|}{SM-G}                           \\ \hline Studentized
                                         & \small{$\text{err}(\hat{F}(S_n'),F')$} & \small{$\text{err}(\Phi,F')$} & \small{$\text{err}(\hat{F}(S'_n),F')$} & \small{$\text{err}(\Phi,F')$} \\ \hline
T                                        & \textbf{0.004}                           & 0.021             & \textbf{0.008}                           & 0.043             \\ \hline
3T+5V                                    & \textbf{0.002}                           & 0.012             & \textbf{0.005}                           & 0.026             \\ \hline
TV                                       & \textbf{0.006}                           & 0.029             & \textbf{0.015}                           & 0.054             \\ \hline
3T/V                                     & \textbf{0.012}                           & 0.031             & \textbf{0.007}                           & 0.052             \\ \hline
$T^2V^2$ & \textbf{0.022}                           & 0.058             & \textbf{0.045}                           & 0.106             \\ \hline
\end{tabular}}
\caption{Studentized EW Sup CDF error compared to N(0,1)}
\label{tb:smf-q-superror}
\end{table}

We see that for both graphons the empirical edgeworth expansion has much lower error than the Gaussian approximation. Also, the linear combinations of the statistics typically have lower error than those which need the estimation of first and second derivatives.



\bk

\section{Additional experiments}
\label{sec:addlexp}
In this section we provide additional experiment results that were left out from the main text for space concerns.

\subsection{Simulation study for \smg and additional results for \sbm}
Here we first present experimental results on the smooth graphon model from~\cite{zhang2017smooth} (\smg)  with  $w(u,v)=(u ^2+v^2)/3\times cos(1/(u^2+v^2))+0.15$.  This graphon is continuous and high rank in contrast to the first graphon, which is piece-wise constant and low rank.  
\begin{figure}[h]
\vspace{-5mm}
\centering
\begin{tikzpicture}[zoomboxarray,
zoomboxarray inner gap=.5cm, 
zoomboxarray columns=1, 
zoomboxarray rows=1]
    \node [image node] { \includegraphics[width=0.35\textwidth]{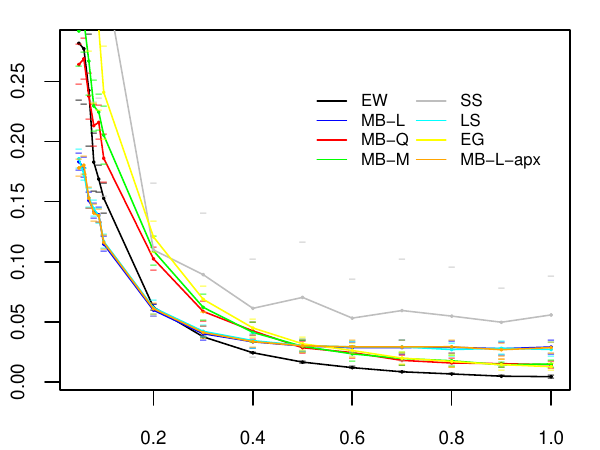} };
    \zoombox[magnification=2]{.15,.7}
\end{tikzpicture}\\
\begin{tikzpicture}[zoomboxarray,
zoomboxarray inner gap=.5cm, 
zoomboxarray columns=1, 
zoomboxarray rows=1]
    \node [image node] { \includegraphics[width=0.35\textwidth]{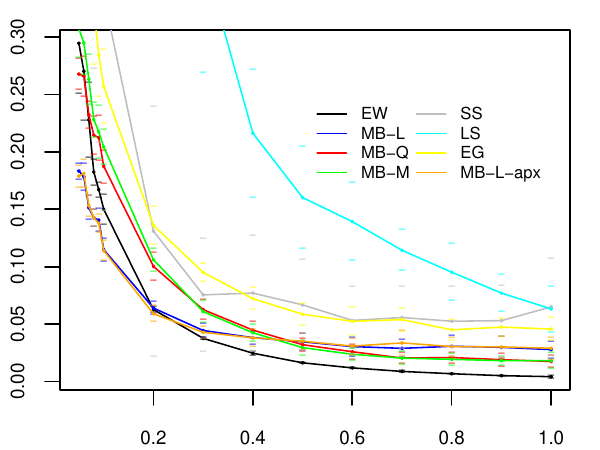} };
    \zoombox[magnification=2]{.15,.7}
\end{tikzpicture}
\vspace{-3mm}
\caption{We plot \KL{$F_n$}{$F_n^*$} \bk for triangle density for all methods on the $Y$ axis, where $F_n^*(t)$ corresponds to the appropriate resampling distribution. We vary the sparsity parameter $\rho_n$ on the $X$ axis. The networks are simulated from \smg. The first row are centered at bootstrap mean and normalized by variance estimation from each method $\hat{\sigma}_n$  The second row is centered by triangles density estimated on the whole graph (\mblapx is centered at approximate triangle density estimated from the whole graph) and normalized by $\hat{\sigma}_n$.}
\label{plot:cdf-error-smg}
\end{figure}

\begin{figure}[!htb]
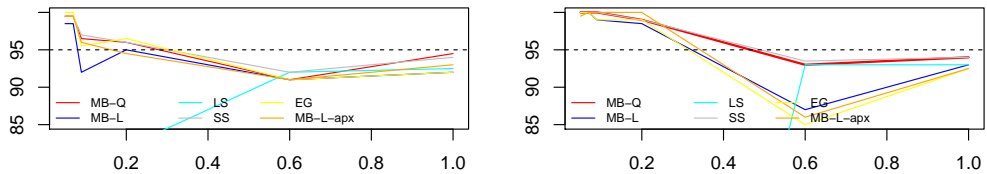

     \includegraphics[trim=2em 2em 0em 0em, clip, width=  .5\textwidth]{plots2/corrected-triangles-sbm.pdf}
     \includegraphics[trim=2em 2em 0em 0em, clip,width=  .5\textwidth]{plots2/corrected-trans-sbm-new.pdf}
\caption{ We present coverage of  95\% Bootstrap Percentile CI with correction for triangles (left) and transitivity coefficient (right) of the \smg model  with $\rho_n$ varying from $0.05$ to $1$.}
\label{plot:coverage-smg}
\end{figure}

 Figure~\ref{plot:cdf-error-smg} and Figure~\ref{plot:coverage-smg} show simulation results of triangles and transitivity for \smg. We show in Figure~\ref{plot:cdf-error-smg}  the maximum of (absolute) difference of the bootstrap CDF $F_n^*$ over the $[-3,3]$ range  (\KL{$F_n$}{$F_n^*$}) for triangles density from the true CDF $F_n$ for sparsity parameter $\rho_n$ varying from $0.05$ to $1$. We show the average of the expected difference over 30 independent runs along with error-bars. In Figure~\ref{plot:coverage-smg}, we show the $95\%$ CI coverage for triangles and transitivity. 
 
 Figure~\ref{plot:cdf-vstar-error} and Figure~\ref{plot:coverage-vs} show additional simulation results for two-stars for both \sbm and \smg. The results of two-stars are similar to those of triangles in the main paper. 

\begin{figure}[H]
\begin{tabular}{c}
    \includegraphics[width=  0.45\textwidth]{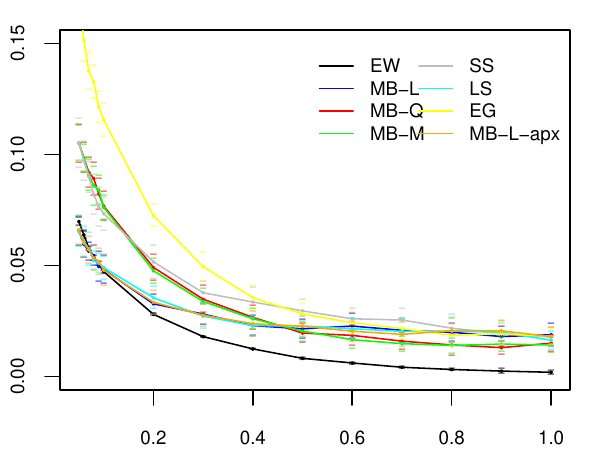} 
    \includegraphics[width=  0.45\textwidth]{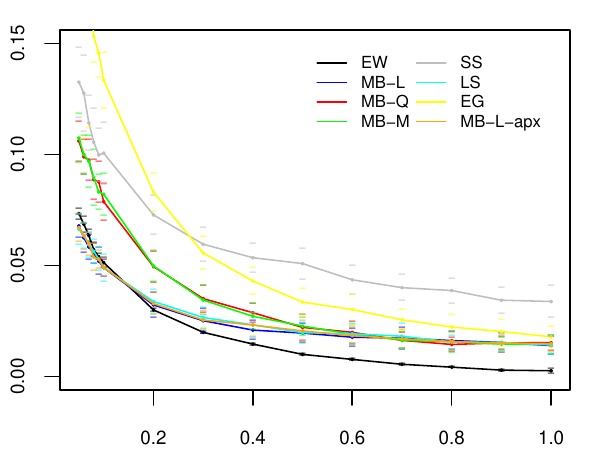}\\
    \includegraphics[width=  0.45\textwidth]{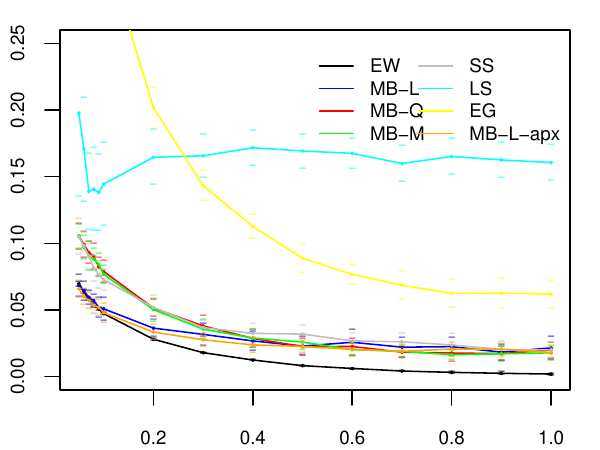}\includegraphics[width=  0.45\textwidth]{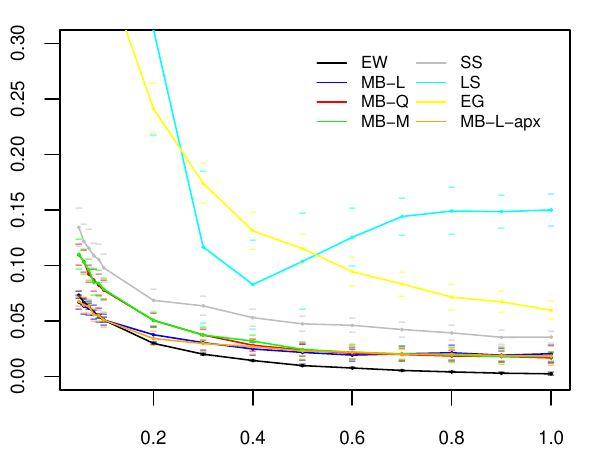}
    
\end{tabular}
 \caption{We plot  \KL{$F_n$}{$F_n^*$}  for two-star density for all methods on the $Y$ axis, where $F_n^*(t)$ corresponds to the appropriate resampling distribution. We vary the sparsity parameter $\rho_n$ on the $X$ axis. The networks in the left column are simulated from \sbm and those in the right column are simulated from \smg. The first row is centered at bootstrap mean and normalized by variance estimation from each method $\hat{\sigma}_n$  The second row is centered by triangles density estimated on the whole graph (\mblapx is centered at approximate triangle density estimated from the whole graph) and normalized by $\hat{\sigma}_n$. }
\label{plot:cdf-vstar-error}
\end{figure}

\begin{figure}[H]
\begin{tabular}{cc}
     \includegraphics[trim=2.5em 2.5em 0em 0em, clip, width=  0.45\textwidth]{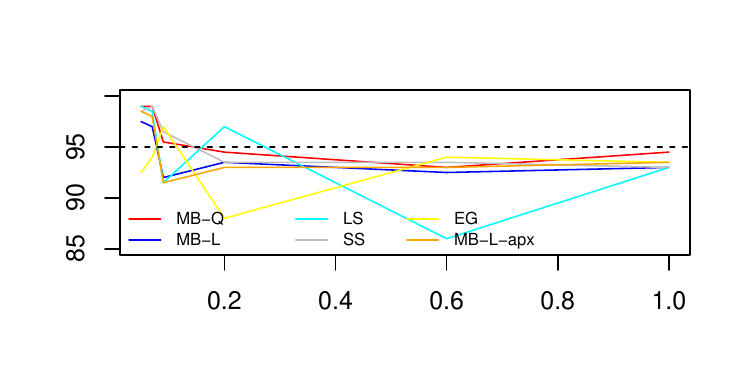}&
     \includegraphics[trim=2.5em 2.5em 0em 0em, clip,width=  0.45\textwidth]{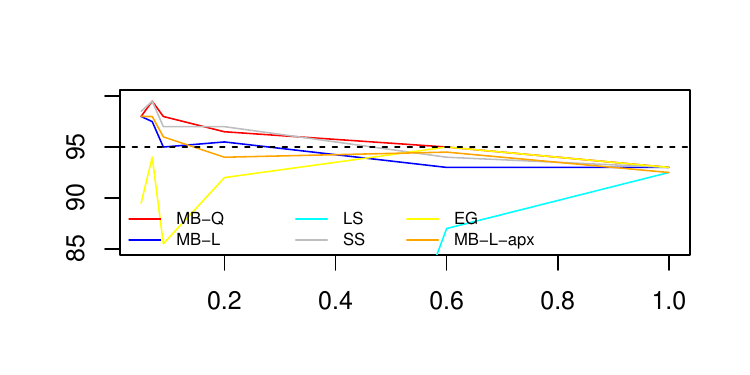}
\end{tabular}
\caption{We present coverage of  95\% Bootstrap Percentile CI with correction for two-stars of the \sbm (left) and \smg (right) models  in $\rho_n$ from $0.05$ to $1$. } \label{plot:coverage-vs}
\end{figure}

\paragraph{Timing results for \smg.}
In Figure~\ref{fig:4cycle-timing-ng2} we show logarithm of running time for four-cycles count against growing $n$ for \smg model. 

\begin{figure}[H]
\centering
    \includegraphics[width=  .5\textwidth]{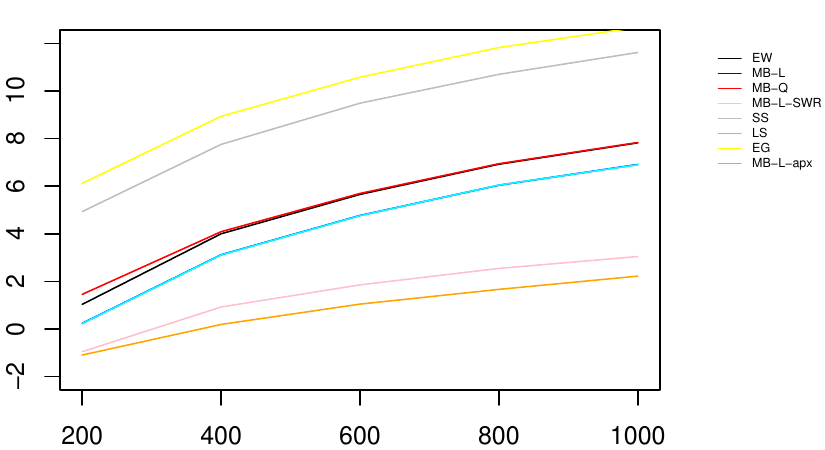} 
\caption{Logarithm of running time for  four-cycles in  \smg  against sample size $n$.}
\label{fig:4cycle-timing-ng2}
\end{figure}

\subsection{Real data experiments on voting similarities of U.S. Congress}
In this section, we apply our algorithms to compare networks representing the voting similarities of U.S. Congress.
\begin{figure}[H]
	\centering
	\includegraphics[width= \textwidth]{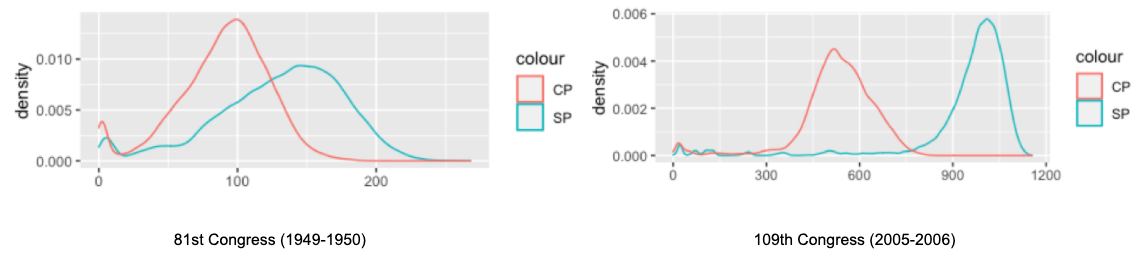}
	\caption{Threshold for forming edges between congress members calculated from histograms of same-party (SP) and cross-party (CP) agreements, illustrated with example of 81st Congress and 109th Congress}
	\label{fig:votethreshold}
\end{figure}
We use roll call vote data from the U.S. House of Representatives~\citep{voteview} from 1949 (commencement of the 81$^{st}$ Congress) to 2012 (adjournment of 112$^{nd}$ Congress). Each Congress forms a network of representatives (nodes).  An edge between a node pair is formed when the number of agreements, i.e. number of times they both vote \textit{yay} or \textit{nay}  exceeds a threshold computed by~\citep{andris2015votethreshold} of this congress. The threshold is computed by constructing histograms of same-party pairs’ number of agreements and cross-party pair’s number of agreements and using the intersection point of the two histograms as the threshold. We will denote same-party by \texttt{SP} and cross-party by \texttt{CP}. For example, the threshold value is $124$ for $81$st Congress and $766$ for $109^{th}$ Congress as illustrated in Figure~\ref{fig:votethreshold}.
For each network, we calculate the normalized cross-party edge density and cross-party triangle density, and perform our bootstrap methods on these quantities.  We construct 95\%  second-order corrected Confidence Intervals from  the MB-Q method and present the CIs over $81$st to $112$nd Congress.  Note that the CIs are adjusted by Bonferonni Correction where  $\alpha=1-0.05/32$ for $32$ experiment congresses. 

In Figure~\ref{fig:voteCI}, we can see a significant decrease in cross party edge densities and cross party triangle densities over the years,  suggesting a trend of decreasing bipartisan agreement.
\begin{figure}[H]
    \centering
    \includegraphics[width=0.45\textwidth]{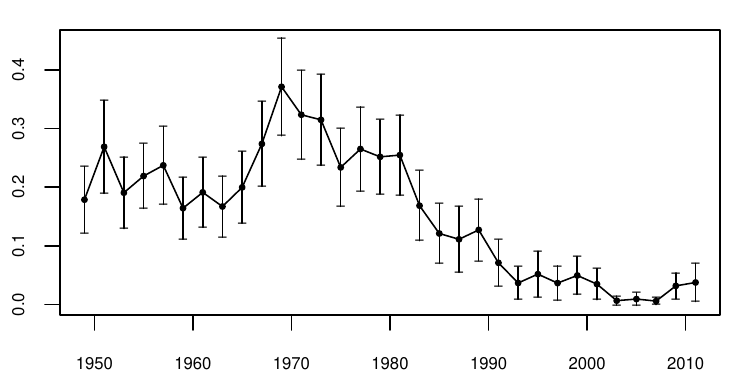}
    \includegraphics[width=0.45\textwidth]{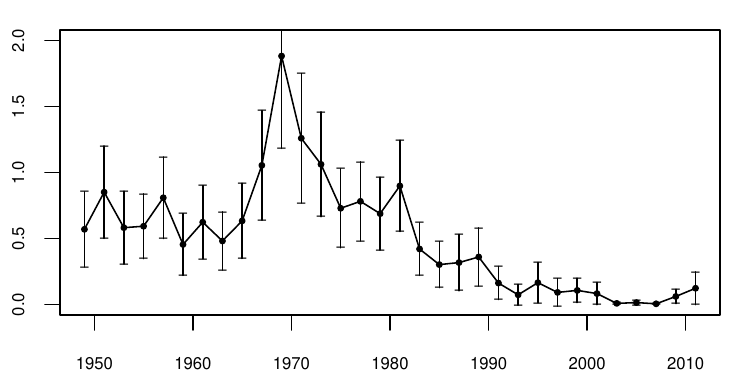}
    \caption{Bonferroni-adjusted 95\% second-order corrected CI for cross party edge density (left) and cross party triangle density (right) from 1949 (commencement of the 81st Congress) to 2012 (adjournment of 112nd Congress).}
    \label{fig:voteCI}
\end{figure}

\subsection{Detailed network statistics and timing for Facebook page networks}
Table~\ref{tab:fb-pages-summary} summarizes the size, edge density, and time cost for different count functionals using our approximate linear bootstraps in the example of Facebook page categories.

\begin{table}[h]
\centering
\scriptsize
\begin{tabular}{|c|c|c|c|c|c|}
\hline
            & Size  & Edge Density & Time Cost (Triangles) & Time Cost (two-stars) & Time Cost (Three-paths) \\ \hline
Athletes    & 13866 & 0.0009       & 11.47s                & 18.40s              & 65.51s                  \\ \hline
Politicians & 5908  & 0.0024       & 2.99s                 & 5.28s               & 47.54s                  \\ \hline
Artists     & 50515 & 0.0006       & 163.15s               & 278.03s             & 843.27s                 \\ \hline
\end{tabular}
\caption{Summaries of size, edge density and time cost for each network and count statistics in the Facebook page dataset.}
\label{tab:fb-pages-summary}
\end{table}

\section{Detailed pseudocode for algorithms}\label{sec:suppalgapx}

In this section, we provide detailed algorithms of our proposed procedures. Algorithm~\ref{alg:apx} presents explicit algorithms of our linear bootstrap and approximated linear bootstrap procedure. Algorithm~\ref{alg:ho} provided explicit algorithms for our proposed higher order correct procedures including quadratic and multiplicative bootstrap procedures. 

 
 Algorithm~\ref{alg:apx_pretri} and Algorithm~\ref{alg:apx_prevs} explicitly provide fast precomputation algorithms to calculate $\tdu$ and $\tdh(i)$ in Eq~\ref{eq:define-h1-tilde} for triangles, and two-stars as described in Section~\ref{subsec:fast-linear-boot} and implemented in real data examples in Section~\ref{sec:realdata}.
   Note that the corresponding algorithm for three-paths is similar to Algorithm~\ref{alg:apx_prevs}, except that each bucket of the permutation sequence has three nodes. Hence one needs to sample twice the degree of the root node to emulate the buckets which have a neighbor of the root node. We omit the details.

\begin{algorithm}	
\SetAlgoLined
\textbf{Input:} Network $A$, motif $R$, number of resamples $B$,  $\texttt{approximate}\in \{True,False\}$, parameter $u$\\
\textbf{If} $\texttt{approximate}=True$\\
\hspace{1em}Compute $\{\tdh(i)\}_{i=1}^n$,   $\tilde{T}_n$ (Eq~\ref{eq:define-h1-tilde}) and $\tilde{\tau}_n$ (Eq~\ref{eq:define-taun-tilde})\\
\textbf{Else}\\
\vskip 0.5em
\hspace{1em}Compute $\hat{T}_n$, (Eq~\ref{eq:cf}), $\{\hat{g}_1(i)\}_{i=1}^n$    (Eq~\ref{eq:g1hat}) and $\hat{\tau}_n$ (Eq~\ref{eq:tauhat})\\
\textbf{End}\\
 \For{$j\in\{1,\dots, B\}$}{	
  Generate $n$ weights $\boldsymbol{\xi}^{(j)}=\{\w^{(j)}_i,i=1,\dots, n\}_{j=1}^B$ using Eq~\ref{eq:gaussproduct}\\
  \textbf{If} $\texttt{approximate}=True$\\
  \text{\hspace{1em} $T_n^*(j)\leftarrow  \ttnl$ (using Eq~\ref{eq:define-tnl-tilde}.)}\\
  \textbf{Else}\\
   \text{\hspace{1em} $T_n^*(j)\leftarrow  \cfbl$ (using Eq~\ref{eq:cfbl}.)}\\
 \textbf{End}\\
}
 Return $\dfrac{1}{B}\sum_j \mathbbm{1}\left(\dfrac{T_n^*(j)-\tilde{T}_n}{\frac{r}{n^{1/2}}\tilde{\tau}_n}\leq u\right)$
 \caption{\label{alg:apx}Construction of linear or approximated linear bootstrap estimate of CDF}
\end{algorithm}

\begin{algorithm}[H]
\SetAlgoLined
\vskip 0.1em
\textbf{Input:} Network $A$, motif $R$, number of resamples $B$, choice of bootstrap procedure $a\in\{M,Q\}$, parameter $u$\\
\vskip 0.5em
Compute $\cfhat$ (Eq~\ref{eq:cf}), $\{\hat{g}_1(i)\}_{i=1}^n$ (Eq~\ref{eq:g1hat}), $\{\hat{g}_2(i,j)\}_{i=1}^n$ (Eq~\ref{eq:g2hat}) and $\hat{\tau}_n$ (Eq~\ref{eq:tauhat})
\vskip 0.5em
 \For{$j\in\{1,\dots, B\}$}{
  Generate $n$ weights $\boldsymbol{\xi}^{(j)}=\{\w^{(j)}_i,i=1,\dots, n\}_{j=1}^B$ using Eq~\ref{eq:gaussproduct}\\
   \textbf{If} $a=M$\\
  \text{\hspace{1em} $T_n^*(j)\leftarrow  \cfbm$ (using Eq~\ref{eq:multiplicative-bootstrap}.)}\\
  \textbf{Else}\\
  \text{\hspace{1em} $T_n^*(j)\leftarrow  \cfbq$ (using Eq~\ref{eq:cfbq}.)}\\
 \textbf{End}\\
}
 Return $\dfrac{1}{B}\sum_j \mathbbm{1}\left(\dfrac{T_n^*(j)-\cfhat}{\frac{r}{n^{1/2}}\hat{\tau}_n}\leq u\right)$
 
 \caption{\label{alg:ho}Construction of quadratic or multiplier bootstrap estimate of CDF}
\end{algorithm}

\rd
\rd
\begin{algorithm}	
\SetAlgoLined
\textbf{Input:} Network A, size of network n,  number of samples N, r=3;\\
 \For{$i\in\{1,\dots, n\}$}{	
  $\texttt{nbs}\_{\texttt{i}}\leftarrow$ neighbors of node i; 
  $\texttt{deg}\_i\leftarrow$degree of i; 
  count=0;\\
    \For{$j\in\{1,\dots, N\}$}{
   \texttt{pos}$\leftarrow$ randomly sampled $\texttt{deg}\_i$ indices without replacement from $\{1\dots n-1\}$; \\ 
   (\texttt{indices}, \texttt{sorted})$\leftarrow$ sort(\texttt{pos}); \\
   ii=0; \\
  \While{ii $<$ \texttt{deg\_i}}{
  \texttt{curr} = \texttt{indices}(ii);
  \texttt{next} = \texttt{indices}(ii+1);\\
  \texttt{currnb}= \texttt{nbs}\_{\texttt{i}}(\texttt{sorted}(ii));
  \texttt{nextnb}=\texttt{nbs}\_{\texttt{i}}(\texttt{sorted}(ii+1));\\
 \textbf{if} (mod(\texttt{curr}, r-1) ==1 and \texttt{next}==\texttt{curr}+1):\\
       $\qquad$ \texttt{count}=\texttt{count}+A(\texttt{currnb},\texttt{nextnb}); ii=ii+1;\\
        \textbf{Else}:\\
 $\qquad$ continue;\\
     ii=ii+1\\
  \textbf{End}\\
  }\vspace{4mm}
 \textbf{End}\\
} \vspace{4mm}
  M(i) $\leftarrow$ (r-1)$\times$\texttt{count}/(N(n-1));\\
   \textbf{End}\\
}
 Return M where $\tdh(i)=M(i)/\rho_n^s$ and $\tdu=\frac{1}{n}\sum_{i=1}^n\tdh(i)$
\caption{\label{alg:apx_pretri} Fast precomputation of approximating $\tdu$ and $\tdh(i)$ for triangles}
\end{algorithm}

\rd
\begin{algorithm}[htb]
\SetAlgoLined
\textbf{Input:} Network $A$, size of network $n$, number of samples $N$, r=3;\\
 \For{$i\in\{1,\dots, n\}$}{	
   $\texttt{nbs}\_{\texttt{i}}\leftarrow$ neighbors of node i; $\texttt{deg}\_{\texttt{i}}\leftarrow$degree of i;  \texttt{count}=0;\\
    \For{$j\in\{1,\dots, N\}$}{
  \texttt{pos}$\leftarrow$ randomly sampled $\texttt{deg}\_{\texttt{i}}$ indices without replacement from $\{1:n-1\}$; \\ 
  \texttt{posConc} $\leftarrow$ randomly sample $\min(\texttt{deg}\_{\texttt{i}},n-1-\texttt{deg}\_{\texttt{i}})$ indices without replacement from $\{1:n\}\setminus\{i \cup \texttt{nbs}\_{\texttt{i}}\}$;\\
    (\texttt{indices}, \texttt{sorted})$\leftarrow$ sort(\texttt{pos}); ConcurrentID=0; ii=0; \\
  \While{$ii < \texttt{deg}\_{\texttt{i}}$}{
    t=0; \texttt{curr} = \texttt{indices}(ii);  \texttt{next}=0; \texttt{nextnb}=0;\\
    \textbf{If} ($ii+1 > \texttt{deg}\_{\texttt{i}}-1$): \texttt{next} = -1;\\
    \textbf{Else}: \texttt{next} = \texttt{indices}(ii+1);
            \texttt{nextnb}=\texttt{nbs}\_{\texttt{i}}(\texttt{sorted}(ii+1));\\
    \texttt{currnb}= \texttt{nbs}\_{\texttt{i}}(\texttt{sorted}(ii));\\     
    \textbf{If} (\texttt{next}!=-1 and (mod(\texttt{curr},2)==1 and \texttt{next}==\texttt{curr}+1)):\\
            $\qquad$  \textbf{If}(A(\texttt{currnb},\texttt{nextnb})==0): t=1;\\
             $\qquad$ ii=ii+1;\\
    \textbf{Else}:\\
    \vspace{-3mm}
           $\qquad$    \texttt{concurrent} = \texttt{posConc}(ConcurrentID);\\
           $\qquad$    \textbf{If} (A(\texttt{currnb},\texttt{concurrent})==1): t=1;\\
           $\qquad$    ConcurrentID=ConcurrentID+1;\\
    \texttt{count}=\texttt{count}+t; \  ii=ii+1;\\
  \textbf{End}\\
  }\vspace{4mm}
 \textbf{End}\\
} \vspace{4mm}
  M(i)=(r-1)$\times$\texttt{count}/(N(n-1));\\
   \textbf{End}\\
}
\vspace{2mm}
 Return $M$ where $\tdh(i)=M(i)/\rho_n^s$ and $\tdu=\frac{1}{n}\sum_{i=1}^n\tdh(i)$
 \caption{\label{alg:apx_prevs} Fast precomputation of approximating $\tdu$ and $\tdh(i)$ for two-stars}
\end{algorithm}

\bk

\clearpage
\bibliography{main}

\begin{thebibliography}{49}
\providecommand{\natexlab}[1]{#1}
\providecommand{\url}[1]{\texttt{#1}}
\expandafter\ifx\csname urlstyle\endcsname\relax
  \providecommand{\doi}[1]{doi: #1}\else
  \providecommand{\doi}{doi: \begingroup \urlstyle{rm}\Url}\fi

\bibitem[Airoldi et~al.(2008)Airoldi, Blei, Fienberg, and Xing]{airoldi-mmsb}
E.~M. Airoldi, D.~M. Blei, S.~E. Fienberg, and E.~P. Xing.
\newblock Mixed membership stochastic blockmodels.
\newblock \emph{Journal of Machine Learning Research}, 9:\penalty0 1981--2014,
  2008.

\bibitem[Aldous(1981)]{aldous-representation-array}
D.~J. Aldous.
\newblock Representations for partially exchangeable arrays of random
  variables.
\newblock \emph{Journal of Multivariate Analysis}, 11:\penalty0 581--598, 1981.

\bibitem[Andris et~al.(2015)Andris, Lee, Hamilton, Martino, Gunning, and
  Selden]{andris2015votethreshold}
C.~Andris, D.~Lee, M.~J. Hamilton, M.~Martino, C.~E. Gunning, and J.~A. Selden.
\newblock The rise of partisanship and super-cooperators in the {US House of
  Representatives}.
\newblock \emph{PloS one}, 10\penalty0 (4):\penalty0 e0123507, 2015.

\bibitem[Athreya et~al.(2018)Athreya, Fishkind, Tang, Priebe, Park, Vogelstein,
  Levin, Lyzinskiand, Qin, and Sussman]{athreya-rdpg-survey}
A.~Athreya, D.~E. Fishkind, M.~Tang, C.~E. Priebe, Y.~Park, J.~T. Vogelstein,
  K.~Levin, V.~Lyzinskiand, Y.~Qin, and D.~L. Sussman.
\newblock Statistical inference on random dot product graphs: a survey.
\newblock \emph{Journal of Machine Learning Research}, 18\penalty0
  (226):\penalty0 1--92, 2018.

\bibitem[Bai and Zhao(1986)]{bai-zhao-independent-edgeworth}
Z.~Bai and L.~Zhao.
\newblock Edgeworth expansions of distribution functions of independent random
  variables.
\newblock \emph{Science in China Series A-Mathematics, Physics, Astronomy and
  Technological Science}, 29\penalty0 (1):\penalty0 851--896, 1986.

\bibitem[Bentkus et~al.(1997)Bentkus, G{\"o}tze, and van
  Zwet]{Bentkus-edgeworth-symmetric}
V.~Bentkus, F.~G{\"o}tze, and W.~van Zwet.
\newblock An edgeworth expansion for symmetric statistics.
\newblock \emph{Annals of Statistics}, 25\penalty0 (2):\penalty0 851--896,
  1997.

\bibitem[Bhattacharyya and
  Bickel(2015)]{Bhattacharyya-subsample-count-features}
S.~Bhattacharyya and P.~J. Bickel.
\newblock Subsampling bootstrap of count features of networks.
\newblock \emph{Annals of Statistics}, 43:\penalty0 2384--2411, 2015.

\bibitem[Bickel et~al.(1986)Bickel, G{\"o}tze, and
  Van~Zwet]{bickel1986edgeworth}
P.~Bickel, F.~G{\"o}tze, and W.~Van~Zwet.
\newblock The edgeworth expansion for u-statistics of degree two.
\newblock \emph{The Annals of Statistics}, pages 1463--1484, 1986.

\bibitem[Bickel and Chen(2009)]{Bickel-Chen-on-modularity}
P.~J. Bickel and A.~Chen.
\newblock A nonparametric view of network models and {Newman}-{Girvan} and
  other modularities.
\newblock \emph{PNAS}, 106:\penalty0 21068--21073, 2009.

\bibitem[Bickel et~al.(2011)Bickel, Chen, and
  Levina]{Bickel-Chen-Levina-method-of-moments}
P.~J. Bickel, A.~Chen, and E.~Levina.
\newblock The method of moments and degree distributions for network models.
\newblock \emph{Annals of Statistics}, 39:\penalty0 38--59, 2011.

\bibitem[Bollobas et~al.(2007)Bollobas, Janson, and
  Riordan]{bollobas-inhomogenous-graphs}
B.~Bollobas, S.~Janson, and O.~Riordan.
\newblock The phase transition in inhomogeneous random graphs.
\newblock \emph{Random Structures and Algorithms}, 31:\penalty0 3--122, 2007.

\bibitem[Borgs et~al.(2019)Borgs, Chayes, Cohn, and Zhao]{borgs-lp-part-one}
C.~Borgs, J.~T. Chayes, H.~Cohn, and Y.~Zhao.
\newblock An {$L^p$} theory of sparse graph convergence {I}: Limits, sparse
  random graph models, and power law distributions.
\newblock \emph{Transactions of the American Mathematical Society},
  372\penalty0 (5):\penalty0 3019--3062, 2019.

\bibitem[Bose and Chatterjee(2018)]{bose-chatterjee-u-stats-resampling}
A.~Bose and S.~Chatterjee.
\newblock \emph{U-Statistics, $M_m$-Estimators, and Resampling}.
\newblock Springer Verlag, New York, 2018.

\bibitem[Chen and Yuan(2006)]{chen-yuan-protein-protein-interaction-network}
J.~Chen and B.~Yuan.
\newblock Detecting functional modules in the yeast protein protein interaction
  network.
\newblock \emph{Bioinformatics}, 22\penalty0 (8):\penalty0 2283--2290, 2006.

\bibitem[Chen and Kato(2019)]{chen2019randomized}
X.~Chen and K.~Kato.
\newblock Randomized incomplete {U}-statistics in high dimensions.
\newblock \emph{Annals of Statistics}, 47\penalty0 (6):\penalty0 3127--3156,
  2019.

\bibitem[Daudin et~al.(2008)Daudin, Koskas, Schbath, and
  Robin]{daudin-exceptionality-motifs}
J.~Daudin, M.~Koskas, S.~Schbath, and S.~Robin.
\newblock Assessing the exceptionality of network motifs.
\newblock \emph{Journal of Computational Biology}, 15\penalty0 (1):\penalty0
  1--20, 2008.

\bibitem[Efron(1980)]{efron-jackknife-resampling}
B.~Efron.
\newblock The {Jackknife}, the {Bootstrap}, and other resampling plans.
\newblock Technical report, Stanford University University, December 1980.

\bibitem[Efron(1987)]{efron-better-ci}
B.~Efron.
\newblock Better bootstrap confidence intervals.
\newblock \emph{Journal of the American statistical Association}, 82\penalty0
  (397):\penalty0 171--185, 1987.

\bibitem[Feller(1971)]{feller-vol-2}
W.~Feller.
\newblock \emph{An introduction to probability theory and its applications.
  {V}ol. {II}.}
\newblock Second edition. John Wiley \& Sons Inc., New York, 1971.

\bibitem[Gao and Ma(2019)]{gao-ma-minimax-network-analysis}
C.~Gao and Z.~Ma.
\newblock Minimax rates in network analysis: Graphon estimation, community
  detection and hypothesis testing.
\newblock 2019.
\newblock URL \url{https://arxiv.org/pdf/1811.06055.pdf}.

\bibitem[Green and Shalizi(2017)]{green-shalizi-network-bootstrap}
A.~Green and C.~Shalizi.
\newblock Bootstrapping exchangeable random graphs.
\newblock \emph{ArXiv e-prints}, 2017.

\bibitem[Hall(1988)]{hall-theoretical-comparison-bootstrap-ci}
P.~Hall.
\newblock {Theoretical Comparison of Bootstrap Confidence Intervals}.
\newblock \emph{The Annals of Statistics}, 16\penalty0 (3):\penalty0 927 --
  953, 1988.
\newblock \doi{10.1214/aos/1176350933}.
\newblock URL \url{https://doi.org/10.1214/aos/1176350933}.

\bibitem[Hall(2013)]{hall-bootstrap-edgeworth}
P.~Hall.
\newblock \emph{The bootstrap and Edgeworth expansion}.
\newblock Springer, NY, 2013.

\bibitem[Hoeffding(1948)]{hoeffding1948}
W.~Hoeffding.
\newblock A class of statistics with asymptotically {Normal} distribution.
\newblock \emph{Annals of Mathematical Statistics}, 19\penalty0 (3):\penalty0
  293--325, 09 1948.

\bibitem[Hoff et~al.(2002)Hoff, Raftery, and
  Handcock]{hoff-raftery-handcock-latent-space-model}
P.~D. Hoff, A.~E. Raftery, and M.~S. Handcock.
\newblock Latent space approaches to social network analysis.
\newblock \emph{Journal of the American Statistical Association}, 97\penalty0
  (460):\penalty0 1090--1098, 2002.
\newblock URL \url{https://doi.org/10.1198/016214502388618906}.

\bibitem[Holland et~al.(1983)Holland, Laskey, and Leinhardt]{holland-sbm}
P.~W. Holland, K.~Laskey, and S.~Leinhardt.
\newblock Stochastic blockmodels: First steps.
\newblock \emph{Social Networks}, 5\penalty0 (2):\penalty0 109 -- 137, 1983.

\bibitem[Hoover(1979)]{hoover-exchangeability}
D.~N. Hoover.
\newblock \emph{Relations on probability spaces arrays of random variables}.
\newblock Institute for Advanced Study,RI, 1979.

\bibitem[Jeffrey~B. et~al.(2020)Jeffrey~B., Poole, Rosenthal, Boche, Rudkin,
  and Sonnet]{voteview}
L.~Jeffrey~B., K.~Poole, H.~Rosenthal, A.~Boche, A.~Rudkin, and L.~Sonnet.
\newblock {Voteview: Congressional Roll-Call Votes Database}.
\newblock {https://voteview.com/}, 2020.
\newblock Online; accessed Nov 29, 2020.

\bibitem[Jing and Wang(2010)]{jing2010unified}
B.-Y. Jing and Q.~Wang.
\newblock A unified approach to edgeworth expansions for a general class of
  statistics.
\newblock \emph{Statistica Sinica}, pages 613--636, 2010.

\bibitem[Karrer and Newman(2011)]{karrer-newman-dcsbm}
B.~Karrer and M.~E.~J. Newman.
\newblock Stochastic blockmodels and community structure in networks.
\newblock \emph{Physical Review E}, 83\penalty0 (016107):\penalty0 1--24, 2011.

\bibitem[Kim et~al.(2014)Kim, Wozniak, Mueller, Shen, and Pan]{kim2014}
J.~Kim, J.~R. Wozniak, B.~A. Mueller, X.~Shen, and W.~Pan.
\newblock Comparison of statistical tests for group differences in brain
  functional networks.
\newblock \emph{NeuroImage}, 101:\penalty0 681--694, 11 2014.
\newblock URL \url{https://www.ncbi.nlm.nih.gov/pubmed/25086298}.

\bibitem[Levin and Levina(2019)]{levin-levina-rdpg-bootstrap}
K.~Levin and E.~Levina.
\newblock Bootstrapping networks with latent space structure.
\newblock \emph{ArXiv e-prints}, 2019.

\bibitem[Lin et~al.(2020)Lin, Lunde, and Sarkar]{network-jackknife-theory}
Q.~Lin, R.~Lunde, and P.~Sarkar.
\newblock On the theoretical properties of the network jackknife.
\newblock \emph{ArXiv e-prints}, 2020.

\bibitem[Liu(1988)]{liu-non-iid-bootstrap}
R.~Liu.
\newblock Bootstrap procedures under some non-i.i.d. models.
\newblock \emph{Annals of Statistics}, 16\penalty0 (4):\penalty0 1696--1708,
  1988.

\bibitem[Lunde and Sarkar(2019)]{subsampling-sparse-graphons}
R.~Lunde and P.~Sarkar.
\newblock Subsampling sparse graphons under minimal assumptions.
\newblock \emph{ArXiv e-prints}, 2019.

\bibitem[Milo et~al.(2002)Milo, Shen-Orr, Itzkovitz, Kashtan, Chklovskii, and
  Alon]{milo-network-motifs}
R.~Milo, S.~Shen-Orr, S.~Itzkovitz, N.~Kashtan, D.~Chklovskii, and U.~Alon.
\newblock Network motifs: Simple building blocks of complex networks.
\newblock \emph{Science}, 298\penalty0 (5594):\penalty0 824--827, 2002.

\bibitem[Myers et~al.(2014)Myers, Sharma, Gupta, and
  Lin]{meyers-social-network-twitter}
S.~A. Myers, A.~Sharma, P.~Gupta, and J.~Lin.
\newblock Information network or social network? the structure of the twitter
  follow graph.
\newblock In \emph{WWW'14}, pages 1--6, 2014.

\bibitem[Newman(2001)]{newman-collaboration-network}
M.~E. Newman.
\newblock The structure of scientific collaboration networks.
\newblock \emph{PNAS}, 98\penalty0 (2):\penalty0 404--409, 2001.

\bibitem[Newman(2003)]{newman-structure-networks}
M.~E. Newman.
\newblock The structure and function of complex networks.
\newblock \emph{SIAM Review}, 45\penalty0 (2):\penalty0 167–256, 2003.

\bibitem[Petrov(2012)]{petrov2012sums}
V.~V. Petrov.
\newblock \emph{Sums of independent random variables}, volume~82.
\newblock Springer Science \& Business Media, 2012.

\bibitem[Politis and
  Romano(1994)]{politis-romano-subsampling-minimal-assumptions}
D.~N. Politis and J.~P. Romano.
\newblock Large sample confidence regions based on subsamples under minimal
  assumptions.
\newblock \emph{Annals of Statistics}, 22\penalty0 (4):\penalty0 2031--2050,
  1994.

\bibitem[Rozemberczki et~al.(2019)Rozemberczki, Davies, Sarkar, and
  Sutton]{rozemberczki2019gemsec}
B.~Rozemberczki, R.~Davies, R.~Sarkar, and C.~Sutton.
\newblock Gemsec: Graph embedding with self clustering.
\newblock In \emph{Proceedings of the 2019 ASONAM}, pages 65--72. ACM, 2019.

\bibitem[Rubin-Delanchy et~al.(2018)Rubin-Delanchy, Priebe, Tang, and
  Cape]{generalized-rdpg}
P.~Rubin-Delanchy, C.~E. Priebe, M.~Tang, and J.~Cape.
\newblock A statistical interpretation of spectral embedding: the generalised
  random dot product graph.
\newblock \emph{ArXiv e-prints}, 2018.

\bibitem[Serfling(1974)]{serfling-concentration-without-replacement}
R.~J. Serfling.
\newblock Probability inequalities for the sum in sampling without replacement.
\newblock \emph{The Annals of Statistics}, 2\penalty0 (1):\penalty0 39--48, 03
  1974.

\bibitem[Ugander et~al.(2011)Ugander, Karrer, Backstrom, and
  Marlow]{ugander-facebook-graph}
J.~Ugander, B.~Karrer, L.~Backstrom, and C.~Marlow.
\newblock The anatomy of the facebook social graph.
\newblock \emph{ArXiv e-prints}, 2011.

\bibitem[Wang and Jing(2004)]{wang-jing-weighted-bootstrap-u-statistics}
Q.~Wang and B.-Y. Jing.
\newblock Weighted bootstrap for {U}-statistics.
\newblock \emph{Journal of Multivariate Analysis}, 91\penalty0 (2):\penalty0
  614, 2004.

\bibitem[Young and Scheinerman(2007)]{young-schneiderman-rdpg}
S.~J. Young and E.~R. Scheinerman.
\newblock Random dot product graph models for social networks.
\newblock In \emph{Proceedings of WAW 2007}, pages 138--149, 2007.

\bibitem[Zhang and Xia(2020)]{zhang-xia-network-edgeworth}
Y.~Zhang and D.~Xia.
\newblock Edgeworth expansions for network moments.
\newblock \emph{ArXiv e-prints}, 2020.

\bibitem[Zhang et~al.(2017)Zhang, Levina, and Zhu]{zhang2017smooth}
Y.~Zhang, E.~Levina, and J.~Zhu.
\newblock {Estimating network edge probabilities by neighbourhood smoothing}.
\newblock \emph{Biometrika}, 104\penalty0 (4):\penalty0 771--783, 09 2017.
\newblock ISSN 0006-3444.
\newblock \doi{10.1093/biomet/asx042}.
\newblock URL \url{https://doi.org/10.1093/biomet/asx042}.

\end{thebibliography}

\end{document}